\documentclass[A4,11pt]{article}
\usepackage[utf8]{inputenc}
\usepackage{amsmath,amssymb}
\usepackage{bbm}
\usepackage{algorithm2e}
\usepackage[utf8]{inputenc}
\usepackage{graphicx}
\usepackage{booktabs}
\usepackage{pdflscape}
\usepackage{multirow}
\usepackage{caption}
\usepackage{subcaption}
\usepackage{relsize}
\usepackage{graphicx}
\usepackage{float}
\usepackage[nolists]{endfloat}

\usepackage{natbib}

\newtheorem{theorem}{Theorem}

\newtheorem{definition}[theorem]{Definition}

\newtheorem{lemma}[theorem]{Lemma}

\newtheorem{proposition}[theorem]{Proposition}

\newenvironment{proof}[1][Proof]{\noindent\textbf{#1.} }{\ \rule{0.5em}{0.5em}}
\graphicspath{{Figures/}}

\title{{\LARGE Testing Instrument Validity with Covariates}{ \thanks{%
We would like to thank Tymon S\l oczy\'{n}ski and the seminar participants at Brandeis University, Sant'Anna School of Advanced Studies, University of Exeter, and University of Siena for helpful comments. We would also like to thank the participants in IAAE 2022, ICEEE 2021, NASMES 2022, SETA 2022 for helpful comments. The authors gratefully acknowledge financial support from ERC grant (number 715940) and the
ESRC Centre for Microdata Methods and Practice (CeMMAP) (grant number RES-589-28-0001).}}}
\author{ \Large Thomas Carr\thanks{Department of Economics, Brown University}, \ Toru Kitagawa\thanks{ Department of Economics, Brown University; Department of Economics, University College London and the Institute of Economic Research at Hitotsubashi University. Email: toru\_kitagawa@brown.edu}}
\date{\today}

\begin{document}

\maketitle

\begin{abstract}
We develop a novel test of the instrumental variable identifying assumptions for heterogeneous treatment effect models with conditioning covariates. We assume semiparametric dependence between potential outcomes and conditioning covariates. This allows us to obtain testable equality and inequality restrictions among the subdensities of estimable partial residuals. We propose jointly testing these restrictions. To improve power, we introduce \textit{distillation}, where a trimmed sample is used to test the inequality restrictions. In Monte Carlo exercises we find gains in finite sample power from testing restrictions jointly and distillation. We apply our test procedure to three instruments and reject the null for one.


\end{abstract}


\pagebreak

\section{Introduction}

Instrumental variable (IV) methods are widely used in empirical economics. The credibility of IV estimates relies on the validity of the identifying assumptions, which are often subject to debate. In IV estimation, conditioning covariates are frequently used to enhance the credibility of the identifying assumptions or improve the precision of estimates. For example, when using proximity to college as an instrument to estimate the return to schooling, researchers control for location characteristics which may otherwise induce dependence between the instrument and the outcome. 
\citet{hubermellace15Testing}, \citet{Kitagawa2015}, and \citet{mourifiewan} propose statistical tests for instrument validity in the heterogeneous treatment effect model of \citet{imbensangrist}. However, the state of the art literature lacks high-power tests for instrument validity that remain computationally feasible for more than a small number of covariates.


The first contribution of this paper is to propose an easy-to-implement test procedure for instrumental validity in the heterogeneous treatment effect model that can accommodate a moderate to large number of conditioning covariates. Building on the semiparametric approach of \citet{CHV11}, we specify that conditioning covariates are independent of unobserved hetergoneity and affect potential outcomes linearly. Under these restrictions, the conditional expectation of the outcome given conditioning covaraites and the instrument follows a partially linear model that depends non-parametrically on the probability of receiving treatment (propensity score). This allows the conditioning covariates to be partialed out and partial residuals obtained. Testable implications for the instrumental variable identifying assumptions which hold without conditioning on covariates can be stated for the subdensities of these partial residuals. With partial residuals in hand, testing can be performed without reference to covariate values. In comparison to the test procedure with conditioning covariates proposed in \cite{Kitagawa2015}, this greatly reduces the computational burden. Increasing the number of conditioning covariates increases only the complexity of estimating the propensity score and partially linear model, not the number of implications to be tested.

Two testable implications are available: \emph{index sufficiency} and \emph{nesting inequalities}. Index sufficiency states that the partial residual subdensities depend on the instrument through the propensity score only. The nesting inequalities are monotonicity inequalities for the joint distributions of the partial residuals and treatment status. They state that, for two distinct values of the propensity score, the distribution of partial residuals for treated units with the higher propensity score nests that of units with the lower propensity score, and the converse is true for untreated units. The two testable implications complement each other: index sufficiency holds the propensity score fixed and compares subdensities for different values of the instrument, while the nesting inequalities compare subdensities for different propensity scores. To exploit this complementarity, we propose jointly testing the two.     

As the subdensity nesting inequalities are ordered with respect to the propensity score, a natural approach to testing them would be to coarsen the propensity score and test the monotonicity of the subdensities conditional on this coarsened propensity score. When the propensity score is strongly correlated with the conditioning covariates, this approach can result in a test with low power. In such cases, classifying observations by propensity score will homogenise the distribution of instrument values across propensity score bins, which can mask violations of instrument validity. 

Our second contribution is to propose a novel approach to testing the nesting inequalities. We suggest testing instead monotonicity of the conditional partial residual subdensities given the instrument. This approach ensures that observations with different instrument values are kept separate. We show that, when the identifying assumptions are satisfied, the nesting inequalities hold over instrument values if the conditional distribution of propensity scores is monotonic in the instrument in the sense of first-order stochastic dominance. If first-order stochastic dominance does not hold for the instrument, we propose trimming the sample and testing the nesting inequalities for a sub-sample where first-order stochastic dominance is satisfied. We refer to this process as \emph{distillation} as it distills the sample down to the observations best able to detect violations.   

We show analytically and numerically the benefits of implementing distillation in terms of the ability to detect violation of instrument validity conditional on covariates. 
Specifically, we present a class of data generating processes for which distillation dominates conditioning on the coarsened propensity score in the following sense. A violation of the nesting inequalities conditional on the coarsened propensity score implies a larger violation of the nesting inequalities with a distilled sample, but not vice versa. That is, a violation of the nesting inequalities with a distilled sample does not imply the existence of a violation of the nesting inequalities conditional on the coarsened propensity score. Consequently, there exist data generating processes where instrument validity fails but the violation can only be detected by the nesting inequalities when performing distillation. 

We perform Monte Carlo exercises to evaluate the size and power of the proposed test procedure and compare it to two alternatives: the \citet{Kitagawa2015} test that does not control for covariates, and the test of nesting inequalities using a coarsened propensity score. We consider processes where the instrument is independent of the covariates, and processes where it is not. In the first case, an instrument that is valid conditional on the covariates will also be valid without conditioning on covariates, so the size of the \citet{Kitagawa2015} test will not be distorted. In the second, size will only be controlled if the conditioning covariates are appropriately controlled for. We find that the size of the proposed test is controlled in all cases. For processes where the \citet{Kitagawa2015} test has the correct size, we find our test procedure has comparable power. For all processes, it has higher power than the test of nesting inequalities using a coarsened propensity score. We also investigate the contributions of the index sufficiency and nesting inequality tests to the power of the joint test. These results suggest that index sufficiency has most power when the distributions of propensity scores for different instrument values are similar, so that there are observations with different instrument values and a similar propensity score. On the other hand, for the processes considered, rejection probabilities when testing the nesting inequalities do not depend strongly on the similarity of the propensity score distributions across instrument values. This is consistent with the two testable implications being complementary. 

To demonstrate the feasbility of the test procedure and illustrate how it can be used to resolve questions of instrument credibility, we present three applications. The first is \citet{Card1993}. This instrument has been tested in \citet{hubermellace15Testing}, \citet{Kitagawa2015}, \citet{mourifiewan} and \citet{sun2020} but, for computational reasons, in each case only a subset of the conditioning covariates of the original specification is used. The proposed test procedure allows all conditioning covariates to be included and remains computationally feasibile. Second, we apply the test to the same-sex instrument of \citet{angristevans1998}. The validity of this instrument was been questioned in \citet{RosenzweigWolpin2000} and its validity has been formally tested by \citet{huber2015} and \citet{mourifiewan}. As with \citet{Card1993}, these studies use only a subset of the conditioning covariates, while our test procedure allows us to include all conditioning covariates. Finally, we turn to identifying strategies that rely on geographical variation in the timing of changes in compulsory schooling laws. Using data for America, \citet{stephensyang2014} presents evidence that changes in education laws are correlated with other unobserved state-level changes, which suggests that this instrument is not valid. We test its validity in the case of the UK using the data of \citet{oreopoulos2006}. 


The remainder of the paper proceeds as follows. The related literature is reviewed in Section \ref{sec:literature review} below. Section 2 introduces the semiparametric testable implications. Section 3 describes the test procedure. Section 4 presents Monte Carlo exercises. Section 5 discusses our three empirical applications. Section 6 concludes.  

\subsection{Related Literature}\label{sec:literature review}

This paper contributes to the active literature on the identification and estimation of the local average treatment effect (LATE) model of \cite{angristetal} and \citet{imbensangrist} and the marginal treatment effect model (MTE) of \citet{ heckmanvytlacil99, heckmanvytlacil01, heckman_vytlacil_2001, heckmanvytlacil05}. This literature has recently been reviewed by \citet{MogstadTorgovitsky2018}. A testable implication of the LATE identifying restrictions was derived in \citet{balkepearl}, \citet{imbensrubin} and \citet{heckmanvytlacil05}. \citet{Kitagawa2015} shows that this testable implication is sharp, and proposes a test procedure using a Kolmogorov-Smirnov statistic. \citet{mourifiewan} reformulates this testable implication as a conditional moment inequality, and proposes a test based on the intersection bounds framework of \citet{chernozhukovleerosen2013}. It is straightforward to generalise this testable implication to incorporate conditioning covariates, but assessing it nonparametrically in a finite sample is limited by the capacity of nonparametric methods to handle multiple conditioning covariates. 

To facilitate the estimation of marginal treatment effects while controlling for multiple conditioning covariates, \citet{CHV11} further imposes that potential outcomes are linear in the conditioning covariates, unobserved heterogeneities enter additively, and these unobservables are statistically independent of the covariates and instrument. The test procedure presented in this paper corresponds to a test for the validity of the \citet{CHV11} identifying assumptions. \citet{MaestasMullenStrand2013}, \citet{Eisenhaueretal2015}, \citet{KlineWalters2016}, \citet{Cornelissenetal2018}, \citet{FelfeLalive2018}, \citet{Bhulleretal2020}, and \citet{CKST2022} among others apply the approach of \citet{CHV11} to recover marginal treatment effects. Our test procedure can be applied in any of these contexts to assess the validity of the identifying assumptions. Furthermore, even if the validity of the instrument is undisputed, the test procedure can still be used to detect errors in the functional form specification for the relationship between potential outcomes and conditioning covariates. \citet{brinchmogstadwiswall} further develops the approach of \citet{CHV11}, and show how marginal treatment effects can be identified in a variety of settings. \citet{mogstadetal} discusses how to extrapolate marginal treatment effects to set identify policy relevant treatment effects.

Several alternative tests of the LATE identifying assumptions have been proposed. \citet{hubermellace15Testing} considers a weaker set of identifying restrictions than \citet{Kitagawa2015}, derives a testable implication, and proposes several test procedures. \citet{LaffersMellace2017} shows that the testable implication of \citet{hubermellace15Testing} is sharp for these identifying assumptions. \citet{sun2020} derives testable implications when the treatment is an ordered or unordered multi-valued variable, and proposes a power improving refinement of the \citet{Kitagawa2015} test procedure. \citet{araietal2022} extends the \citet{Kitagawa2015} test to fuzzy regression discontinuity designs. \citet{farbmacherguberklaassen2022} proposes a test procedure that uses random forests and classification and regression trees to find violations of the monotonicity and exclusion assumptions. \citet{kedagnimourifie} derives a testable implication of the instrument independence assumption in the form of a set of inequalities, reformulates these inequalities as a set of conditional expectations, and develops a test procedure using the framework of \citet{chernozhukovleerosen2013}. \citet{machadoetal} discusses tests for instrument validity in the case of a binary outcome which responds monotonically with respect to the treatment. \citet{frandsenetal2023} develops a test of the identifying assumptions of research designs which exploit random assignment of judges. The closest paper to ours is \citet{MaoSantAnna20}, which also uses the MTE framework of \citet{CHV11}, and derives index sufficiency and the nesting inequalities. This paper focuses on testing the nesting inequalities, and does not jointly test index sufficiency. The nesting inequalities are tested by considering many possible discretisations of the propensity score. We differ in that we jointly test index sufficiency and the nesting inequalities. In addition, for the nesting inequalities, rather than examining subdensities indexed by the propensity score, we propose constructing a distilled sample so that the inequalities hold for subdensities indexed by the instrument, and show the advantages of this approach. 

We contribute to several empirical literatures by testing the validity of commonly used instruments and identification strategies. The college proximity instrument of \citet{Card1993} has also been employed in \citet{CameronTaber2004} to estimate the effects of borrowing constraints on education, \cite{CHV11} to estimate the marginal treatment effect of education, and \citet{Heckmanetal2018} to estimate the effects of education on labour market and health outcomes. The same-sex instrument of \citet{angristevans1998} has been used widely to identify the effects of family size. \citet{BlackDeveruxSalvanes2010},\citet{ConleyGlauber2006}, \citet{AngristLavySchlosser2010} and \citet{Beckeretal2010} use this instrument to empirically investigate the existence of a quantity-quality trade-off for children, while \citet{CrucesGaliani2007} employs it to investigate the relationship between family size and maternal labour supply. Other than \citet{oreopoulos2006}, geographical variation in schooling laws has been employed by \citet{acemogluangrist2000} to estimate human capital externalities, \citet{lochnermoretti2004} to estimate the effect of education on imprisonment rate, and \citet{ClarkRoyer2013} to estimate the effects of education on life expectancy. 

\section{Model and Testable Implications}

\subsection{Setting}

The data is a random sample $\left( Y,D,X,Z\right) $,
where $Y$ is an observed outcome continuously distributed on $\mathbb{R}$, $D\in \left\{ 1,0\right\} $ is an observed binary treatment status
indicating treated $(D=1)$ or non-treated $\left( D=0\right) $, $X\in 
\mathcal{X}\subset \mathbb{R}^{k_{X}}$ is a vector of individual pretreatment observable covariates, and $Z\in \mathcal{Z}\subset \mathbb{R}%
^{k_{Z}}$ is a vector of instrumental variables. For the moment, we impose no restrictions on the instrument.
\ Let $\left\{ Y_{dz}\right\} _{d\in \left\{ 1,0\right\} ,z\in \mathcal{Z}}$
be the set of potential outcomes indexed by treatment status $d\in \left\{
1,0\right\} $ and instrument value $z\in \mathcal{Z}$. 
The observed outcome 
$Y$ can be written as
\begin{equation*}
Y=Y_{1Z}D+Y_{0Z}(1-D)\text{.}
\end{equation*}
We write a selection equation of unconstrained form as
\begin{equation*}
D=1\{v(X,Z,U_D)\geq 0\}.
\end{equation*}%
We interpret this selection equation as follows. With $U_{D}$ fixed at $u_{D}$, $D_{xz}=1\{v(x,z,u_{D})\geq 0\}$ generates the counterfactual selection response for an individual when their pre-treatment characteristics $X$ and instrument $Z$ are exogenously set at $x$ and $z$. \\ 

\bigskip

The following three assumptions identify the marginal treatment effect (MTE):

\bigskip

\textbf{MTE Identification Assumptions (Heckman and Vytlacil (2005)):}

\begin{description}
\item[(A1)] \textit{Instrument exclusion restriction}: Potential outcomes do not depend on the instrument in the sense $Y_{1z}=Y_{1}$ and $%
Y_{0z}=Y_{0}$ for all $z\in \mathcal{Z}$ with probability one.

\item[(A2)] \textit{Random assingnment}: $(\left\{ Y_{1z}\right\} _{z\in 
\mathcal{Z}},\left\{ Y_{0z}\right\} _{z\in \mathcal{Z}},U_{D})$ are
statistically independent of $Z$ conditional on $X$, and $U_{D}$ is
statistically independent of $X$.

\item[(A3)] \textit{Instrument monotonicity}: For every $z$, $z^{\prime }\in 
\mathcal{Z}$, $v(x,z,u_{D})\geq v(x,z^{\prime },u_{D})$ holds for
all $u_{D}\in \mathcal{U}$ or $v(x,z,u_{D})\leq v(x,z^{\prime},u_{D})$ holds for all $u_{D}\in \mathcal{U}$.
\end{description}

\bigskip

The instrument exclusion restriction (A1) precludes the instrument from having a direct causal impact on the outcome for any member of the population.
With this restriction imposed, an individual's potential outcomes are reduced to
a pair indexed only by their treatment status, i.e, $Y_{1}$ denotes
their potential outcome with treatment and $Y_{0}$ their potential
outcome without treatment.  The observed outcome $Y$ can then be
written as $Y=Y_{1}D+Y_{0}(1-D)$.

Assumption (A2) states that, conditional on $X$, the instrument $Z$ is assigned without reference to underlying potential outcomes and unobserved heterogeneity in selection response. This corresponds to the conventional assumption of instrument exogeneity. In the heterogeneous treatment effect model, however, it should be noted that instrument exogeneity takes the form of joint statistical independence of potential outcomes and selection heterogeneities, rather than the zero-correlation exogeneity restriction. Assumption (A3) states that a hypothetical change in the instrument from $z^{\prime }$ to $z$ can induce some individuals to opt in to treatment or some individuals to opt out of treatment, but never both simultaneously. Equivalently, in the terminology of \citet{imbensangrist}, there is a sub-population of compliers associated with this hypothetical change, but no sub-population of defiers. Instrument monotonicity of the form (A3) allows for multiple instrumental variables. \citet{Mogstadetal2020} argue that, in the presence of multiple instruments, (A3) implies severe restrictions on individual selection responses.  Our test procedure allows for multiple instruments and can be used to assess a necessary testable implication of (A3) jointly with the exclusion and random assignment restrictions. Section \ref{subsec:multipleIV} discusses in detail the implications of multiple instruments for instrument validity testing.    

Following \citet{heckmanvytlacil01, heckmanvytlacil05}, assumptions (A1) - (A3) imply a heterogeneous treatment effect model of the following form: 
\begin{equation}
\left\{ 
\begin{array}{l}
Y_{1}=m_{1}(X)+\tilde{U}_{1}; \\ 
Y_{0}=m_{0}(X)+\tilde{U}_{0}; \\ 
Y=DY_{1}+(1-D)Y_{0}; \\ 
D=\mathbbm{1}\{p(X,Z)- V \geq 0\},%
\end{array}%
\right.  \label{Model1}
\end{equation}%
with%
\begin{equation}
\left( Y_{1},Y_{0},V\right) \perp Z|X\text{ \ and } V\perp X.
\label{IV exclusion & exogeneity}
\end{equation}
Here, $E[Y_{1}|X=x]=m_{1}(x)$ and $E[Y_{0}|X=x]=m_{0}(x)$
are regression equations of potential outcomes $(Y_{1},Y_{0})$ on
the control covariates, $\tilde{U}_{1}$ and $\tilde{U}_{0}$ are
unobserved heterogeneities in the potential outcomes $(Y_{1},Y_{0})$, $p(x,z) = \Pr (D=1|X=x,Z=z)$ is the propensity score, and $V$ is uniformly distributed on $[0,1]$ and is independent of $(X,Z)$. \

\subsection{Semiparametric Testable Implications}

The MTE identifying assumptions introduced above involve counterfactual
variables that are never jointly observed for the same individual, so cannot be directly verified using the distribution of observables. However, \textit{necessary} testable implications have been established. These conditions relate to the distribution of observables, so may be examined empirically. They consist of \textit{single index sufficiency} and \textit{nesting inequalities}. We present a proof in Appendix \ref{sec:proposition 1 and 2 proofs} for completeness.

\begin{proposition}
\label{prop: joint-ineq} Assume (A1) - (A3) and that, conditional on $X=x$, the propensity score $%
p\left( X, Z\right) $ has a nondegenerate distribution. \ Then, at any such $x$,
the following two conditions hold:

\noindent (i) \textbf{Index sufficiency}: The conditional distribution of $(Y,D)$ given $(X,Z)$ depends on $Z$ only through the propensity score $p(X,Z)$, i.e., for any measurable set $A \subset \mathcal{Y}$ and $d \in \{0 , 1 \}$,
\begin{equation}
    \Pr(Y \in A ,D = d|X = x,Z=z) = \Pr(Y \in A ,D = d|p(X,Z)=p, X = x) \label{propostion one index sufficiency}
\end{equation}
holds for any $(x,z)$ where $p = p(x,z)$.

\noindent (ii) \textbf{Nesting inequalities}: The following inequalities hold for every $p\geq p^{\prime }$ in the support of $p\left( X,Z\right)$ and every measurable subset $A \subset \mathcal{Y}$:
\begin{eqnarray}
&&\Pr(Y \in A,D=1|p\left( X,Z \right) =p,X=x)\geq \Pr(Y \in A,D=1|p\left(X,Z \right) =p^{\prime },X=x),  \label{proposition one density} \\
&&\Pr(Y \in A ,D=0|p\left( X,Z \right) =p,X=x)\leq \Pr(Y \in A,D=0|p \left(
X,Z\right) =p^{\prime },X=x).  \notag
\end{eqnarray}
\end{proposition}

\bigskip

Condition (i), index sufficiency, restricts the conditional distribution of $(Y,D)$ given $(X,Z)$ to depend on $Z$ through the propensity score $p(X,Z)$ only. That is, if there are multiple values of $Z$, say $z \neq z' \in \mathcal{Z}$, that yield the same value of the propensity score given $X=x$, the conditional distributions of $(Y,D)$ given $(X=x,Z=z)$ and $(X=x,Z=z')$ must be identical. Condition (ii) provides distributional monotonicity inequalities among the conditional distributions of $(Y,D)$ given $(p,x)$ with respect to $p = p(X,Z)$. This corresponds to the testable monotonicity shown by \citet{heckmanvytlacil05}.  The difference between the two sides of (\ref{proposition one density}) can be expressed as
\begin{eqnarray*}
&&\Pr (Y_{1}\in A,p^{\prime }<U_{D}\leq p|X=x) \\
&=&\Pr (Y_{1}\in A|p^{\prime }<U_{D}\leq p,X=x)\times \Pr \left( p^{\prime
}<U_{D}\leq p|X=x\right).
\end{eqnarray*}%
Individuals whose selection heterogeneity $U_{D}$ falls in $(p^{\prime },p]$ can be viewed as compliers conditional on $X=x$. (\ref{proposition one density}) imposes non-negativity of the treated outcome probability density function for these conditional compliers. Detecting any violation of the conditions of Proposition 1 allows the joint restrictions (A1) - (A3) to be refuted. 
Neither of these testable implications restrict the support of the instrument $Z$, i.e., $Z$ can be discrete, continuous, or multi-dimensional. 

The two testable implications shown in Proposition \ref{prop: joint-ineq} are distinct and assess different aspects of the distribution of observables. 
Holding $X$ fixed, index sufficiency compares conditional distributions of $(Y,D)$ given $(p(X,Z),X)$ at a \textit{fixed} value of the propensity score, while the nesting inequalities compare conditional distributions across \textit{different} values of the propensity score.
Consequently, there are scenarios in which one testable implication is useful for assessing instrument validity but the other is not. 
For instance, if $Z$ is a scalar and $p(x,z)$ is strictly monotonic in $z$ at every $x$, instrument validity cannot be assessed through index sufficiency since their is no variation in the instrument conditional on $(p(X,Z),X)$, whereas it can be assessed using the nesting inequalities. 
Conversely, when the propensity score does not vary with the instrument conditional on $X$ (i.e. the instrument is irrelevant), index sufficiency does have content for assessing instrument validity while the nesting inequalities do not. As such, the two testable implications complement each other and joint assessment of them is desirable.

When the dimension of $X$ is not small, nonparametrically testing the implications listed in Proposition \ref{prop: joint-ineq} is challenging. 
\citet{Kitagawa2015} proposes testing the inequalities of Proposition \ref{prop: joint-ineq} (ii) with a Kolmogorov-Smirnov type test statistic where a supremum is taken over a large class of instrument functions defined on the product space of $Y$ and $X$. Similar to optimization in the empirical risk minimizing classification problem, the computational complexity of searching for this supremum increases with the dimension of $X$. Due to this computational barrier, the finite sample power of this test is unknown, and it has rarely been implemented in practice. Furthermore, accommodating a continuous instrument in this procedure remains an open problem.

\bigskip

We develop a test for instrument validity which circumvents the implementation issues that arise due to conditioning covariates. This test \textit{jointly} assesses index sufficiency and the nesting inequalities. In addition, the test procedure can accommodate continuous instruments. 

To achieve this, we impose the semiparametric restrictions proposed by \citet{CHV11}, which have since been used in many empirical studies of marginal treatment effects. 
\begin{description}
\item[(A4)] \textit{Linear Functional Form}: The regression equations for potential outcomes have the parametric form
\begin{equation}
\left\{ 
\begin{array}{c}
Y_{1}=\alpha _{1}+X^{\prime }\theta _{1}+\tilde{U}_{1}; \\ 
Y_{0}=\alpha _{0}+X^{\prime }\theta _{0}+\tilde{U}_{0},%
\end{array}%
\right.  \label{Y linear in X}
\end{equation}%
\end{description}

In addition, strengthen the exogeneity restriction (A2) to:

\begin{description}
\item[(A5)] \textit{Strong Exogeneity}: Let $(\tilde{U}_1, \tilde{U}_0)$ be the residuals of the potential outcome regression equations defined in (\ref{Y linear in X}). $(\tilde{U}_{1},\tilde{U}%
_{0},V)$ is independent of $\left( X,Z\right)$.
\end{description}

\bigskip

(A4) restricts the conditional mean potential outcomes to depend linearly on the covariates. (A5) implies that the mean zero unobserved terms $\left( \tilde{U}_{1},\tilde{U}%
_{0}\right) $ are independent of the covariates. (A4) and (A5) jointly imply that  covariates affect the potential outcomes only through their conditional means. Note that this assumption
does not imply that the treatment effect is homogeneous, since the individual
causal effect $Y_{1}-Y_{0}$ remains random even conditional on $X$. Note also that, while only conditional mean treatment effects may depend on $X$, the distribution of treatment effects remains otherwise unconstrained. 

In general, under (A1) - (A3), the MTE at $X=x$ and $V=p\in \mathcal{P}$ is
nonparametrically identified by%
\begin{equation*}
\frac{\partial }{\partial p}E\left[ Y|p\left( X,Z\right) =p,X=x\right] .
\end{equation*}%
\citet{CHV11} shows that imposing the parametric functional forms of (A4) together with strengthening the exogeneity assumption of (A2) to (A5), yields the following partially linear regression equation for the observed outcome $Y$:
\begin{equation}
E\left[ Y|p\left( X,Z\right) =p,X=x\right] =x^{\prime }\theta
_{0}+px^{\prime }\left( \theta _{1}-\theta _{0}\right) +\phi \left( p\right) 
\text{,}  \label{reg y on p and x}
\end{equation}%
where $\phi \left( \cdot \right) $ is a unknown function of the propensity
score, which absorbs the intercept term $p\alpha _{1}+\left( 1-p\right)
\alpha _{0}$. As a result, the MTE at $X=x$ and $V=p$ is identified by 
\begin{equation}
MTE(x,p) \equiv E[Y_1 - Y_0 | X=x, V = p] =  x^{\prime }\left( \theta _{1}-\theta _{0}\right) +\phi ^{\prime }\left(
p\right). \label{eq:simplified_MTEs}
\end{equation}
Since the only nonparametric component in this regression equation is $\phi \left( \cdot \right)$, which is a function of the
scalar-valued propensity score, the computational complexity of estimation does not explode as the dimension of $X$ increases. \citet{CHV11} argues that this is a practical advantage of the functional form specification (A4) and the stronger instrument exogeneity assumption
(A5). 

Since the partial linear regression (\ref{reg y on p and x}) identifies $\left( \theta _{1},\theta _{0} \right)$, it is possible to compute
\begin{equation*}
U_{1}\equiv \alpha _{1}+\tilde{U}_{1} = Y_{1}-X^{\prime}\theta _{1}    
\end{equation*}
for every treated observation, and 
\begin{equation*}
U_{0}\equiv
\alpha _{0}+\tilde{U}_{0} = Y_{0}-X^{\prime }\theta _{0}
\end{equation*}
for every untreated observation. 

\bigskip

The joint restrictions (A1), (A3), (A4) and (A5) simplify not only estimation of MTE, but also the testable implications of MTE identifying assumption. To see how, consider
the distribution of $U_{1}$ conditional on $\left( D=1,p\left( X,Z\right)
=p,X=x\right) $. \ Under (A1), (A3) and (A4),%
\begin{eqnarray}
U_{1}|\left( D=1,p\left( X,Z\right) =p, X=x\right) &\sim &U_{1}|\left(
V\leq p,p\left( X,Z\right) =p,X=x\right) \notag \\
&\sim &U_{1}|\left( V \leq p\right), \label{eq.U1} 
\end{eqnarray}
and similarly,
\begin{eqnarray}
U_{0}|\left( D=0,p\left( X,Z\right) =p, X=x\right) &\sim &U_{0}|\left(
V > p,p\left( X,Z\right) =p,X=x\right) \notag \\
&\sim &U_{0}|\left( V > p\right). \label{eq.U0}
\end{eqnarray}
That is, for $d=0,1$, the distribution of the residual, $U_{d}=Y_{d}-X^{\prime }\theta _{d}$,
conditional on $\left( D=d,p\left( X,Z\right) ,X\right)$ depends only on $p\left( X,Z\right)$. Exploiting this single index sufficiency for the distributions of $U_1$ and $U_0$, we obtain the following necessary testable implications for the joint restrictions (A1), (A3), (A4) and (A5).

\begin{proposition} \label{prop2}
 Let $U=DU_{1}+(1-D)U_{0}$. Assume the propensity score $p\left( X,Z\right)$ has a nondegenerate distribution. If (A1), (A3), (A4) and (A5) hold, then the following two conditions hold: 

\noindent (i) \textbf{Index sufficiency}: The conditional distribution of $(U,D)$ given $(X,Z)$ depends only on the propensity score $p(X,Z)$, i.e., for any measurable set $A \subset \mathbb{R}$ and $d \in \{0 , 1 \}$
\begin{equation}
    \Pr(U \in A ,D = d|X=x,Z=z) = \Pr(U \in A ,D = d|p(X,Z)=p)\label{eq:index sufficiency statement}
\end{equation}
holds for any $(x,z)$ where $p = p(x,z)$.

\noindent (ii) \textbf{Nesting inequalities}: For every $p\geq p^{\prime }$ in the support of the propensity score distribution and every measurable subset $A \subset \mathcal{Y}$:
\begin{eqnarray}
&&\Pr(U \in A,D=1|p\left( X,Z \right) =p)\geq \Pr(U \in A,D=1|p\left(X,Z \right) =p^{\prime }),  \label{density} \\
&&\Pr(U \in A ,D=0|p\left( X,Z \right) =p)\leq \Pr(U \in A, D=0|p \left(X,Z\right) =p^{\prime }).  \notag
\end{eqnarray}
\end{proposition}

\bigskip

We refer to the joint restrictions of (A1), (A3), (A4), and (A5) as \textit{instrument validity} and base our test on the testable implications of Proposition \ref{prop2}. Compared to Proposition \ref{prop: joint-ineq}, Proposition \ref{prop2} replaces the outcome variable $Y$ with the outcome residuals $U$, and removes the requirement to condition on $X$. Index sufficiency is strengthened such that the distribution of $(U,D)$ depends on $(X,Z)$ only through the propensity score $p(X,Z)$. The nesting inequalities are strengthened to imply distributional monotonicity with respect to changes in $p(X,Z)$ regardless of the underlying variation in $(X,Z)$. These testable implications are also shown in \citet{MaoSantAnna20}.  

As the testable implications of Proposition \ref{prop2} follow from assumptions (A1), (A3), (A4) and (A5), any test based on them is a joint test of these assumptions. That is, in addition to the conventional exclusion, random assignment and monotonicity, we will be testing the functional form assumption of (A4) and strong exogenienty assumption of (A5). Violation of testable implications can be due to misspecification in the functional form (\ref{Y linear in X}) or the restricted heterogeneity of MTEs with respect to the observable characteristics, i.e., additive separability of MTEs in (\ref{eq:simplified_MTEs}). In addition, when the propensity score is estimated, misspecification of the propensity score function may also lead to a rejection of the null. 

As index sufficiency should hold for all $(X,Z)$  it leads to a large number of inequalities, and testing all of them may be impractical. In this paper our focus is on testing the validity of the instrument, assuming that the strong exogeneity of $X$ holds. As such, we focus on testing that the conditional joint distribution of $(U, D)$ given the propensity score does not vary with $Z$. 

\subsection{Potential Difficulties Detecting Violation of Instrument Validity}\label{sec: distilled instrument discussion}

In this section we show how the conditional distribution of the propensity score given $Z$ affects the usefulness of each of the testable implications characterized in Proposition \ref{prop2}. We use the simple case of a single binary instrument to illustrate this point. 

Index sufficiency (\ref{eq:index sufficiency statement}) can be tested by examining whether the difference
\begin{equation}\label{eq:index sufficiency practical test}
    \Pr(U\in A, D=d | p(X,Z)=p, Z=1) - \Pr(U\in A, D=d| p(X,Z)=p, Z= 0)
\end{equation}
is zero.

For a data generating process to have empirical content to asses index sufficiency, there must exist some $p$ where observations with both $Z=0$ and $Z=1$ are available. This occurs if we can find covariate values $x \neq x'$ such that $p(x,1) = p(x^{\prime},0)$. If no such $p$ exists, then for all $p$ either $\Pr(p(X,Z) = p, Z=0) = 0$ or $\Pr(p(X,Z) = p, Z=1) = 0$ and (\ref{eq:index sufficiency practical test}) cannot be calculated. Hence, the practical relevance of (\ref{eq:index sufficiency statement}) is determined by the degree of overlap between the propensity score distributions conditional on $Z=0$ and $Z=1$.  


Now consider the nesting inequalities (\ref{density}). These can be viewed as analogous to those tested in \citet{Kitagawa2015} with the observed outcome $Y$ replaced by the outcome residual $U$ and the instrument replaced by the propensity score $p(X,Z)$. If the propensity score were coarsened, e.g., by binning, the test procedure with no conditioning covariates and a discrete instrument of \citet{Kitagawa2015} could be applied. Although this approach would attain an asymptotically valid test size,\footnote{Assuming the estimation errors in $U$ and $p(X,Z)$ asymptotically vanish} there are circumstances where conditioning on the propensity score has a detrimental effect on power. 

Maintain the simplifying assumption of a single binary instrument. Assume that the conditioning covariates $X$ are statistically independent of the unobservables $(\tilde{U}_1, \tilde{U}_0, V)$,as required by (A5). Furthermore, assume that they are also independent of $Z$. Consider the following simple violation of the exclusion restriction (A1) where the residuals of the potential outcome regression for $D = 1$ depend on $Z$
\begin{equation}
    U_1 = \alpha_1  + \nu Z + \tilde{U}_1, \label{eq:section 2.3 violation process}
\end{equation}
with $\nu \neq 0$. As $U_1$ is a sum of quantities that are independent of $Z$, it will be independent of $Z$.

If $\Pr(Z=1|p(X,Z)=p)$ is bounded away from 0 and 1, the additive separability of the selection equation and the independence of $X$ leads to the following decomposition:
\begin{align}
    \Pr(U\in A, D=1|p(X,Z)=p) = & \Pr(U\in A, D=1|p(X,Z)=p,Z=0)\Pr(Z=0|p(X,Z)=p) \notag \\ & + \Pr(U\in A, D=1|p(X,Z)=p,Z=1)\Pr(Z=1|p(X,Z)=p) \notag \\
    = & \Pr(U_1\in A, V \leq p |p(X,Z) = p, Z=0)\Pr(Z=0|p(X,Z)=p) \notag \\ 
    & + \Pr(U_1 \in A, V \leq p|p(X,Z) = p, Z=1)\Pr(Z=1|p(X,Z)=p) \notag\\
    = & \Pr(U_1\in A, V \leq p |Z=0)\Pr(Z=0|p(X,Z)=p) \notag \\ 
    & + \Pr(U_1 \in A, V \leq p|Z=1)\Pr(Z=1|p(X,Z)=p).
    \label{eq:decomposition 2}
\end{align}
The first equality states that, for a given $p$ where $\Pr(Z=1|p(X,Z)=p)$ is bounded away from 0 and 1, we have a mixture of observations with $Z=0$ and $Z=1$ with conditioning covariates implicitly adjusting to hold the propensity score constant. The second equality uses that the event $(U \in A, D = 1)$ is equivalent to $(U_1 \in A, V \leq p)$. To see the third, note that conditioning on $(p(X,Z) = p, Z=z)$ is equivalent to conditioning on $Z = z$ and $X \in \{x|p(X = x, Z = z) = p\}$. $X$ is independent of $V$ and $U_1$, so $X \in \{x|p(X = x, Z = z) = p\}$ can be dropped from the conditioning set.

Assume that the violation of the exclusion restriction assumption leads there to be some $A \subset \mathcal{Y}$ and $p^{\prime} < p$ where the following inequality holds:
\begin{equation}
    \Pr(U_1 \in A, V \leq p^{\prime} | Z= 0) - \Pr(U_1 \in A, V \leq p |Z=1) > 0. \label{eq:source violation}
\end{equation}
Notice that if the instrument were valid the reverse inequality would hold.

To see how (\ref{eq:source violation}) can drive a violation of the first nesting inequality (\ref{density}), apply the decomposition \eqref{eq:decomposition 2} to the difference between the left and right side of (\ref{density})
\begin{align}
    & \Pr(U \in A, D = 1 | p^{\prime}) - \Pr(U \in A, D=1 |p) \label{eq:nesting violation} \\ 
    = &\left[(\Pr(Z=1|p) - \Pr(Z=1|p')\right]\left[ \Pr(U_1 \in A, V \leq p^{\prime}|Z=0) - \Pr(U_1 \in A, V \leq p |Z=1) \right] \notag \\
    & - \Pr(Z=0|p)  \cdot \Pr(U_1 \in A, p^{\prime} < V  \leq p| Z=0) - \Pr(Z=1|p^{\prime}) \cdot \Pr(U_1 \in A, p^{\prime} < V \leq p | Z=1). \notag
\end{align}
The violation (\ref{eq:source violation}) appears in the second line of (\ref{eq:nesting violation}) multiplied by $(\Pr(Z=1|p) - \Pr(Z=1|p')$. The remaining terms in the third line are all negative. In order for (\ref{eq:source violation}) to result in a violation of the nesting inequality  (i.e.,for (\ref{eq:nesting violation}) to be positive), it is desirable that $\Pr(Z=1|p) - \Pr(Z=1|p')$ is of large magnitude. 

This magnitude depends on the strength of the relationship between the covariates $X$ and the propensity score. If variation in the propensity scores is driven mainly by $X$, $\Pr(Z=1|p)$ will change little in $p$, and the magnitude of $(\Pr(Z=1|p) - \Pr(Z=1|p')$ will be small. This will result in limited detectability of violations of instrument validity through the nesting inequality (\ref{eq:nesting violation}). 

In the Monte Carlo studies in Section \ref{sec:MC}, we show numerically that tests of the nesting inequalities that directly examine the monotonicity of $\Pr(U \in A, D = 1 | p)$ in $p$ suffer from low power if variation in the propensity score is driven by the conditioning covariates $X$. 

\subsection{Distillation}\label{subsec:distilledIV}

To improve power when testing nesting inequalities, we consider conditioning on $Z$ rather than conditioning on the propensity score. Consider the difference between the probabilities of $\{U \in A,D=1 \}$ conditional on $Z$,
\begin{align}
    &\Pr(U \in A, D = 1|Z = 0) - \Pr(U \in A, D=1| Z = 1) \label{eq:prob_difference} \\
    =&\int_0^1 \left[ \Pr(U \in A, D=1 |p(X,Z)=p, Z=0) f(p|Z=0) - \Pr(U \in A, D=1 |p(X,Z)=p,Z=1) f(p|Z=1) \right] dp,  \notag
\end{align}
where $f(p|Z=z)$ denotes the conditional probability density (or probability mass function) of the propensity score $p(X,Z)$ given $Z=z$. 

Observing a positive sign for (\ref{eq:prob_difference}) can imply three possibilities: 
\begin{enumerate}
    \item[\textbf{Case 1}:] If the distribution of $p(X,Z)|Z=0$ is first-order stochastically dominated by the distribution of $p(X,Z)|Z=1$, a positive sign implies index sufficiency is violated, the nesting inequalities are violated, or both. In particular, if index sufficiency holds it implies violation of the nesting inequalities, i.e., monotonicity of $\Pr(U \in A, D=1 |p(X,Z)=p)$ in $p$ fails for some value of $p \in (0,1)$ with a positive measure.
    
    \item[\textbf{Case 2}:]  If $f(p|Z=0) = f(p|Z=1)$, a positive sign for (\ref{eq:prob_difference}) implies violation of index sufficiency at some value of $p \in (0,1)$ with a positive measure.
    
    \item[\textbf{Case 3}:]  If the distribution of $p(X,Z)|Z=0$ is not first-order stochastically dominated by the distribution of $p(X,Z)|Z=1$, a positive sign does
    \textit{not} imply violation of index sufficiency or the nesting inequalities. i.e., (\ref{eq:prob_difference}) can be positive even under the null.
\end{enumerate}
Hence, (\ref{eq:prob_difference}) is informative for detecting violations of instrument validity only if first-order stochastic ordering holds between the distributions of $p(X,Z)|Z=0$ and $p(X,Z)|Z=1$. 
First-order stochastic dominance would hold, for instance, if the propensity score were given by a probit function $\Phi(X'\delta + \gamma Z)$ with $\gamma > 0$, and $X$ and $Z$ were uncorrelated. However, first-order stochastic dominance is not implied by (A1), (A3), (A4) and (A5). If $\gamma > 0$, $\delta > 0$ and $X$ and $Z$ are negatively correlated, first-order stochastic dominance can fail even if (A1) - (A5) hold. In general, a positive value of  (\ref{eq:prob_difference}) cannot be taken as evidence against instrument validity.  

Case 3 renders a positive sign of (\ref{eq:prob_difference}) inconclusive about violations of instrument validity. To rule out Case 3, first-order stochastic dominance can be checked as part of the test procedure. In cases where it does not hold, we consider trimming the sample so that the conditional distribution of the propensity score $p(X,Z)$ given $Z$ is stochastically monotonic in $Z$. We then assess the sign of an inequality analogous to (\ref{eq:prob_difference}) using the potentially trimmed sample. We refer to this process as \textit{distillation}.

Given an instrument, $Z \in \mathcal{Z} \subset \mathbb{R}$, we define a binary indicator $S_1 \in \{ 0, 1 \}$ such that $S_1=1$ indicates that an observation is included in the sample used for testing nesting inequalities, and $S_1=0$ indicates that the observation is trimmed and discarded. The distribution of $S_1 \in \{0, 1\}$ can depend on $(X,Z)$, but we require $S_1$ to be independent of $(U_1,U_0,V)$.

Under the assumptions of Proposition \ref{prop2}, index sufficiency of the propensity score for the distribution of $(U_d,V)$ implies
\begin{align}
    & \Pr(U \in A, D=1|Z=z, S_1=1) \notag \\ 
    = & \int_0^1 \Pr(U_1 \in A, D=1|p(X,Z)=p,Z=z, S_1=1) dP_{p|ZS_1}(p|z^,S_1=1) \notag \\
    = & \int_0^1 \Pr(U_1 \in A, V \leq p |p(X,Z)=p,Z=z, S=1) dP_{p|ZS_1}(p|z,S_1=1) \notag \\
    = & \int_0^1 \Pr(U_1 \in A, V \leq p) dP_{p|ZS_1}(p|z,S_1=1), \label{eq:distilledIV1}
\end{align}
where $P_{p|Z S_1}$ is the conditional distribution of the propensity score $p(X,Z)$ given $Z=z$ and $S_1=1$. The third equality follows by Assumption (A5) and $S_1$ being independent of $(U_1,U_0,V)$. Since the integrand in (\ref{eq:distilledIV1}) is monotonic in $p$ by Proposition \ref{prop2}, $\Pr(U \in A, D=1|Z=z,S_1=1)$ is monotonically increasing in $z$ if $P_{p|ZS_1}(\cdot|z,S_1=1)$ is monotonic in $z$ in terms of first-order stochastic dominance. Similarly, the stochastic monotonicity of $P_{p|ZS_1}(\cdot|z,S_1=1)$ in $z$ implies 
\begin{equation*}
    \Pr(U \in A, D=0|Z=z,S_1=1) = \int_0^1 \Pr(U_0 \in A, V > p) dP_{p|ZS_1}(p|z,S_1=1) \label{eq:distilledIV2}
\end{equation*}
is monotonically decreasing in $z$. 

We now formally define distillation
\begin{definition}[Distillation] \label{def:distilledIV}
Let $S_1 \in \{0,1 \}$ be an indicator for inclusion in the sample. $(Z,S_1)$ is a distillation rule if $(U_1,U_0,V) \perp S_1$ and the conditional distribution of the propensity score given $\{ Z,S_1 \}$, $P_{p|Z S_1}(\cdot|z,S_1=1)$, is stochastically increasing in $z \in \mathcal{Z}$ in terms of the first-order stochastic dominance.
\end{definition}

As discussed above, the nesting inequalities shown in Proposition \ref{prop2} remain valid even when we replace the propensity score $p(X,Z)$ with the instrument $Z$ and restrict the sample to a subsample with $S_1=1$. We hence obtain the next proposition.

\begin{proposition}\label{Distilled Instrument Testable Implications} Suppose the assumptions of Proposition \ref{prop2} holds. Let $(Z,S_1) \in \mathcal{Z} \times \{0,1\}$ be a 
distilled sample as defined in Definition \ref{def:distilledIV}. For any $z \leq z^{ \prime} \in \mathcal{Z}$ and every measurable set $A \subset \mathcal{Y}$,
\begin{eqnarray}
&&\Pr(U \in A,D=1| Z =z,S_1=1)\geq \Pr(U \in A,D=1|Z =z^{\prime },S_1=1), \label{eq:nesting_inequalities_distilled_IV} \\
&&\Pr(U \in A ,D=0|Z =z,S_1=1) \leq \Pr(U \in A, D=0|Z =z^{\prime },S_1=1).  \notag
\end{eqnarray}
\end{proposition}

Proposition \ref{Distilled Instrument Testable Implications} shows that a testable implication exists in terms of the subdensities conditional on $Z$ provided that the sample is appropriately trimmed. However, it does not show any advantages of this approach relative to a direct test of the nesting inequalities in Proposition \ref{prop2} with a coarsened propensity score. In the proposition below, we characterise a class of processes and a restriction on trimming such that distillation will outperform a test of the nesting inequalities for observations binned by propensity score values in terms of detecting violations of instrument validity.


\begin{proposition} \label{prop:distillation better}
Let $P^{-}$ and $P^{+}$ be disjoint intervals with the lower bound of $P^{+}$ weakly greater than the upper bound of $P^{-}$ such that $\Pr(p(X,Z) \in P^{-}, Z = z)$ and $\Pr(p(X,Z) \in P^{+}, Z = z)$ are positive for $z = 0, 1$. Consider the differences
\begin{eqnarray}
&&\Pr(U \in A,D=1|p \in P^{-}) - \Pr(U \in A,D=1|p \in P^{+}) , \label{eq:nesting_comparison_1} \\
&&\Pr(U \in A ,D=0|p \in P^{+}) - \Pr(U \in A, D=0|p \in P^{-}),  \label{eq:nesting_comparison_2}
\end{eqnarray}
and
\begin{align}
&\Pr(U \in A,D=1|Z =0,S_1=1, p \in P^{-} \cup P^{+}) - \Pr(U \in A,D=1| Z =1,S_1=1, p \in P^{-} \cup P^{+})) , \label{eq:distilled_comparison_1} \\
&\Pr(U \in A ,D=0|Z =1,S_1=1, p \in P^{-} \cup P^{+}) - \Pr(U \in A, D=0|Z =0,S_1=1, p \in P^{-} \cup P^{+}). \label{eq:distilled_comparison_2}
\end{align}

Assume that the selection indicator $S_1$ is chosen so that the following hold \\ 
\noindent (i) \textbf{Valid Distillation}: $S_1$ satisfies the conditions of Definition \ref{def:distilledIV} \\
(ii) \textbf{Trimming Rule}: $S_1 = 1$ for all observations with $Z = 0$ and $p \in P^{-}$, and all observations with $Z = 1$ and $p \in P^{+}$.

Let $\mathbb{P}$ be the class of data generating processes satisfying \\
\noindent (iii) \textbf{Labeling}: $Z=0$ and $Z=1$ are chosen such that
\begin{equation}
    \Pr(p \in P^{+}|Z = 0, p \in P^{-} \cup P^{+}) \leq \Pr(p \in P^{+}|Z = 1, p \in P^{-} \cup P^{+}) \label{eq:labeling restriction} 
\end{equation}\\
(iv) \textbf{Conditional Nesting}: For every $p\geq p^{\prime}$ in the support of the propensity score distribution, every measurable subset $A \subset \mathcal{Y}$, and $z \in \{0,1\}$:
\begin{eqnarray}
&&\Pr(U \in A,D=1|p\left( X,Z \right) =p, Z = z)\geq \Pr(U \in A,D=1|p\left(X,Z \right) =p^{\prime }, Z = z), \label{eq:conditional nesting D = 1} \\
&&\Pr(U \in A ,D=0|p\left( X,Z \right) =p, Z = z)\leq \Pr(U \in A, D=0|p \left(X,Z\right) =p^{\prime }, Z = z). 
\end{eqnarray}

The following statements hold:
\begin{enumerate}
    \item For all data generating processes in $\mathbb{P}$, for any measurable subset $A \subset \mathbb{R}$ and $P^{-}$, $P^{+}$ where (\ref{eq:nesting_comparison_1}) is positive
    \begin{multline*}
        \Pr(U \in A,D=1|Z =0,S_1=1, p \in P^{-} \cup P^{+}) - \\ \Pr(U \in A,D=1| Z =1,S_1=1, p \in P^{-} \cup P^{+})  \geq \\ \Pr(U \in A,D=1|p \in P^{-}) - \Pr(U \in A,D=1|p \in P^{+}).\label{eq:distilled_comparison_1}
    \end{multline*}
    Similarly, for any measurable subset $A \subset \mathbb{R}$ where (\ref{eq:nesting_comparison_2}) is positive
    \begin{multline*}
     \Pr(U \in A ,D=0|Z =1,S_1=1, p \in P^{-} \cup P^{+}) - \\ \Pr(U \in A, D=0|Z =0,S_1=1, p \in P^{-} \cup P^{+}) \geq \\
       \Pr(U \in A ,D=0|p \in P^{+}) - \Pr(U \in A, D=0|p \in P^{-}). 
    \end{multline*}
    \item There exist data generating processes in $\mathbb{P}$ where,  given $P^{-}$ and $P^{+}$, a positive value of (\ref{eq:distilled_comparison_1}) for some $A \subset \mathbb{R}$ does not imply a positive value of (\ref{eq:nesting_comparison_1}) for the same $A \subset \mathbb{R}$. Similarly, there exist processes in $\mathbb{P}$ where,  given $P^{-}$ and $P^{+}$,  a positive value of \eqref{eq:distilled_comparison_2} for some  $A \subset \mathbb{R}$ does not imply a positive value of \eqref{eq:nesting_comparison_2}. 
    \item Given $P^{-}$ and $P^{+}$, there exist data generating processes in $\mathbb{P}$ such that distillation can detect violation of instrument validity while the original nesting inequalities cannot, i.e., (\ref{eq:distilled_comparison_1}) or (\ref{eq:distilled_comparison_2}) is positive at some $A \subset \mathbb{R}$, while (\ref{eq:nesting_comparison_1}) and (\ref{eq:nesting_comparison_2}) are less than or equal to zero for any $A \subset \mathbb{R}$.
\end{enumerate}
\end{proposition}
A proof is presented in Appendix \ref{sec:distillation power improvement}

Regarding the construction of $S_1$, condition (ii) requires that all observations with $Z = 0$ and  $p \in P^{-}$ are retained, and all observation with  $Z = 1$ and $p \in P^{+}$ are retained. On the other hand, no restriction is placed on the trimming of observations with $p \in P^{+}$ and $Z = 0$ or observations with $p \in P^{-}$ and $Z = 1$. That is, first order stochastic dominance should be attained solely through trimming observations with $Z = 0$ and a high propensity score, and observations with $Z = 1$ and a low propensity score. In Section \ref{sec:distillation procedure} we present an algorithm that satisfies these conditions.  

The inequality \eqref{eq:labeling restriction} can always be satisfied by setting $Z = 1$ for the instrument value with a higher proportion of observations in $P^{+}$. Condition (iv) imposes restrictions on the dependence between the partial residuals and $Z$. For example, condition (iv) follows if $Z$ affects the density of partial residuals, and this relationship is independent of $X$ i.e. the partial residuals admit the representation
\begin{equation*}
    U_d = \nu Z + \tilde{U}_d \quad \tilde{U}_d \perp (X,Z).
\end{equation*}

The first statement of Proposition \ref{prop:distillation better} says that, for any $A$ where a nesting inequality calculated using propensity score bins is positive, the nesting inequality calculated using a distilled sample will also be positive and at least as large in magnitude. 
That is, any violation of instrument validity that can be detected by a comparison of the data across the propensity score bins can also be detected by a comparison of the data across different instrument values, and the magnitude of the violation can be larger when comparing across instrument values. The second statement says that, for some $A$ where the nesting inequality calculated using a distilled sample is positive, the nesting inequality calculated with propensity score bins is not positive. In other words, at such $A$ the violation becomes undetectable if we compare across propensity score bins rather than instrument values. The second statement compares the nesting inequalities and distilled inequalities at common $A$, while the third statement compares their maxima over $A$. The third statement says that there exist data generating processes where instrument validity does not hold, but this violation can only be detected by comparing across instrument values. This proposition characterises a class of data generating processes where distillation is guaranteed to outperform the coarsened propensity score approach in terms of detecting invalid instruments. It thus provides theoretical support for implementing distillation.

\subsection{Example Process}\label{sec:example process}

  Proposition \ref{prop:distillation better} implies that we can more easily detect violations of instrument validity by testing the nesting inequalities with a distilled sample. To verify the practical relevance of this gain, we perform a numerical exercise using the following process
\begin{align*}
    & Y = \nu Z - \nu (1 - Z) + D (X^{\prime}\theta_1 + \tilde{U}_1) + (1 - D) (X^{\prime}\theta_0 + \tilde{U}_0), \\
    & D = \mathbbm{1}\{Z(X^{\prime}\delta_1 + \alpha) + (1 - Z)(X^{\prime}\delta_0 - \alpha) + U_D \geq 0\}, \\
    & \Pr(Z = 1) = 1 / 2, \\
    & X \sim N(0, I), \\
    & (U_D, \tilde{U}_1, \tilde{U}_0) \sim N(\mu, \Sigma), \\
    & \mu = \begin{pmatrix}
0 \\
\mu_1 \\
\mu_0 \\
    \end{pmatrix} \quad
      \Sigma = \begin{pmatrix}
        1 &  \sigma_{1} \rho_{D, 1} & \sigma_0 \rho_{D, 0} \\
        \sigma_1 \rho_{D, 1} & \sigma_1^2 & \sigma_1 \sigma_0\rho_{1, 0} \\
        \sigma_0 \rho_{D,0} & \sigma_1 \sigma_0 \rho_{1, 0} & \sigma_0^2\ \\
    \end{pmatrix}.
\end{align*} 
$X$ is assumed to consist of $k_X \geq 3$ covariates. We assume $\nu \neq 0$, so the instrument has a direct causal effect on the outcome. In addition, we impose $\alpha > 0$ and $\delta_0^{\prime}\delta_0 = \delta_1^{\prime}\delta_1 = \delta^{\prime}\delta > 0$. 

We consider the asymptotic setting, with both propensity scores and the parameters $\theta_1$ and $\theta_0$ estimated. We assume the propensity scores are estimated consistently. Given the above specification, the distribution of propensity scores conditional on $Z = 1$ will first-order stochastically dominate the distribution conditional on $Z = 0$. Conditions (i) and (ii) of Proposition \ref{prop:distillation better} can then be satisfied by including the entire sample, and condition (iii) is immediate. 

Let $\hat{\theta}_1^{*} $ and $\hat{\theta}_0^{*}$ be the probability limits of the estimates of $\theta_1$ and $\theta_0$. We assume these estimates are obtained semiparametrically.\footnote{Specifically, we assume these estimates are obtained as described in Section \ref{sec:test procedure}.} The partial residuals are
\begin{equation*}
    U = Y - DX^{\prime}\hat{\theta}_1^{*} - (1 - D) X^{\prime}\hat{\theta}_0^{*}.
\end{equation*}
In Appendix \ref{app:numerical example} we show that condition (iv) of Proposition \ref{prop:distillation better} follows immediately if these estimates are consistent i.e. $\hat{\theta}_1^{*} = \theta_1$ and $\hat{\theta}_0^{*} = \theta_0$. However, as the exclusion restriction is violated, in general $\hat{\theta}_1^{*} \neq \theta_1$ and $\; \hat{\theta}_0^{*} \neq \theta_0$. Proposition \ref{prop:bias simplification} in Appendix \ref{app:numerical example} shows that the effect of the asymptotic bias on the partial residual subdensities can be summarised by two constants which depend on the parameters $\nu, \alpha$ and $\delta^{\prime}\delta$ introduced above, and $\rho_{\delta}$ 
\begin{equation*}
    \rho_{\delta} \equiv \frac{\delta_0^{\prime}\delta_1}{\delta^{\prime}\delta}.
\end{equation*}
Furthermore, $\nu$ enters only as a scalar multiplier. Hence, the numerical exercise considers a class of processes indexed by values of $\delta^{\prime}\delta$ and $\rho_{\delta}$. Proposition \ref{prop:bias simplification} further shows that the effect of the bias is the same for the density of $(U, D = 1)$ conditional on $Z = 1$ and the density of $(U, D = 0)$ conditional on $Z = 0$, and the density of $(U, D = 0)$ conditional on $Z = 1$ and the density of $(U, D = 1)$ conditional on $Z = 0$. As the conditional densities for $(U, D = 0)$ thus mirror those for $(U, D = 1)$, we present results for the $(U, D = 1)$ densities only.  

The numerical exercise is as follows. $\mu_1$ is set to 0.3, $\mu_0$ to 0, and the elements of $\Sigma$ are drawn randomly. We set $\alpha = 0.3$, $\nu = 0.2$, and loop over values of $\rho_{\delta}$ and $\delta^{\prime}\delta$.
We define $P^{-}$ to include propensity scores below the median and $P^{+}$ to include propensity scores above the median. At each set of values, we calculate the violation of the nesting inequality with observations split into bins by the median propensity score as
\begin{equation}
    \underset{A \in \mathcal{C}(\mathcal{Y})}{\text{max}} \left\{\Pr(U \in A ,D=1|p \in P^{-}) - \Pr(U \in A, D=1|p \in P^{+})\right\} \label{eq:numerical example nesting violation definition}.
\end{equation}
and the violation of the nesting inequality with a distilled sample as
\begin{equation}
      \underset{A \in \mathcal{C}(\mathcal{Y})}{\text{max}} \left\{ \Pr(U \in A ,D=1|Z =0,S_1=1, p \in P^{-} \cup P^{+}) - \\ \Pr(U \in A, D=1|Z =1,S_1=1, p \in P^{-} \cup P^{+})\right\} \label{eq:numerical example distilled violation definition}. 
\end{equation}
where $\mathcal{C}(\mathcal{Y})$ is the set of connected intervals. We then compare these two violations. Figure \ref{fig:numerical example difference} plots the difference between the violation of the nesting inequality with a distilled sample and the nesting inequality with propensity score bins. For most parameter values, the violation with a distilled sample is larger, with the gain largest for high values of $\delta^{\prime}\delta$ and $\rho_{\delta}$ close to zero. Figure \ref{fig:numerical example regions} plots the region of the parameter space where the violation with a distilled sample is larger than the violation with propensity score bins, and the region where Statement 3 of Proposition \ref{prop:distillation better} holds. That is, the violation can only be detected by using a distilled sample. Testing with a distilled sample outperforms coarsening the propensity score for the majority of parameter values, and the region where Statement 3 holds is large.  

\begin{figure}
    \centering
    \includegraphics[scale = 0.8]{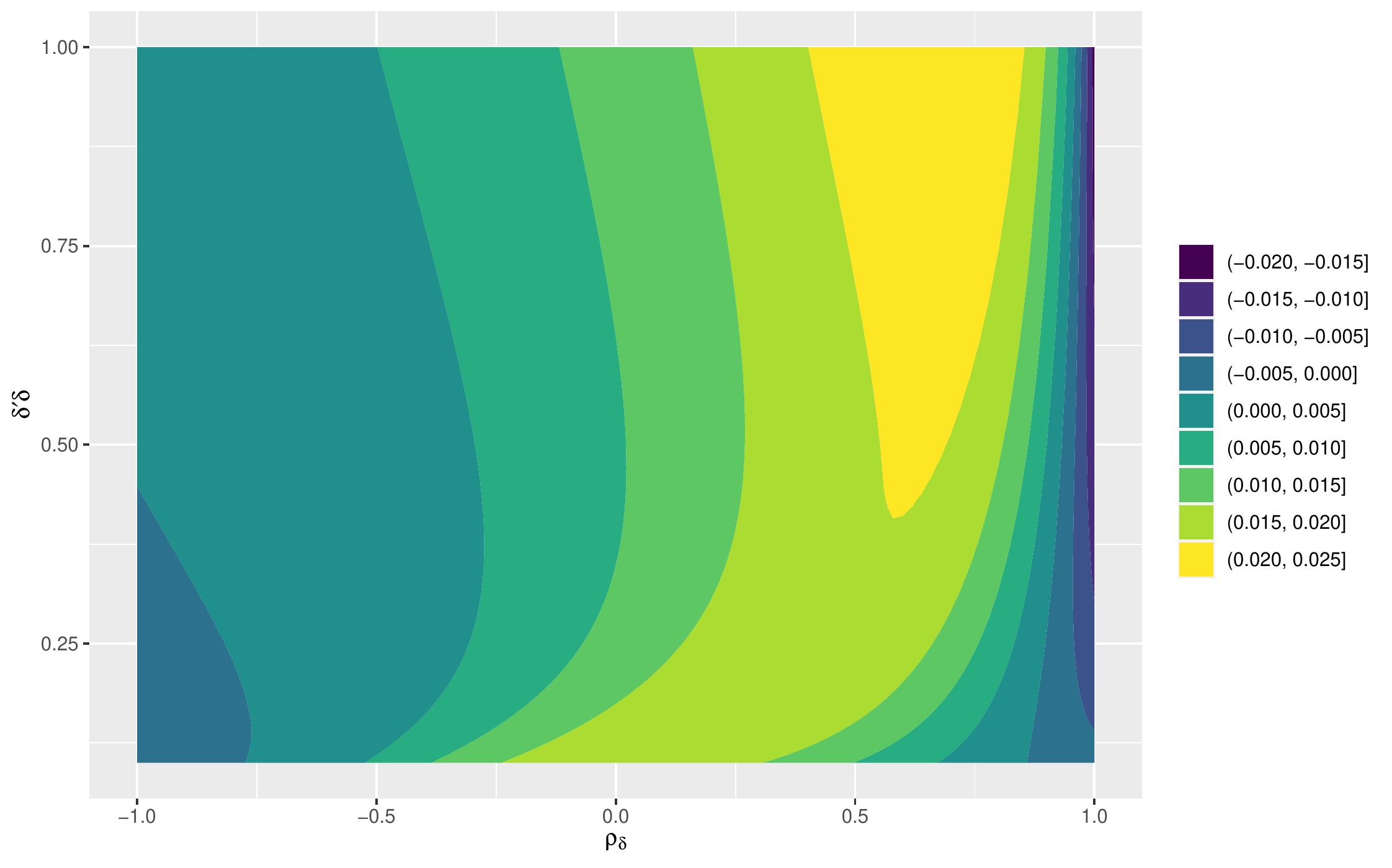}
    \caption{Contour plot of the difference between the violation of the nesting equality with a distilled sample and the violation of the nesting equality with a coarsened propensity score. A positive value implies that the violation with a distilled sample is larger than the violation with a coarsened propensity score. A negative value implies that the violation with a coarsened propensity score is larger that the violation with a distilled sample.}
    \label{fig:numerical example difference}
\end{figure}

\begin{figure}
    \centering
    \includegraphics{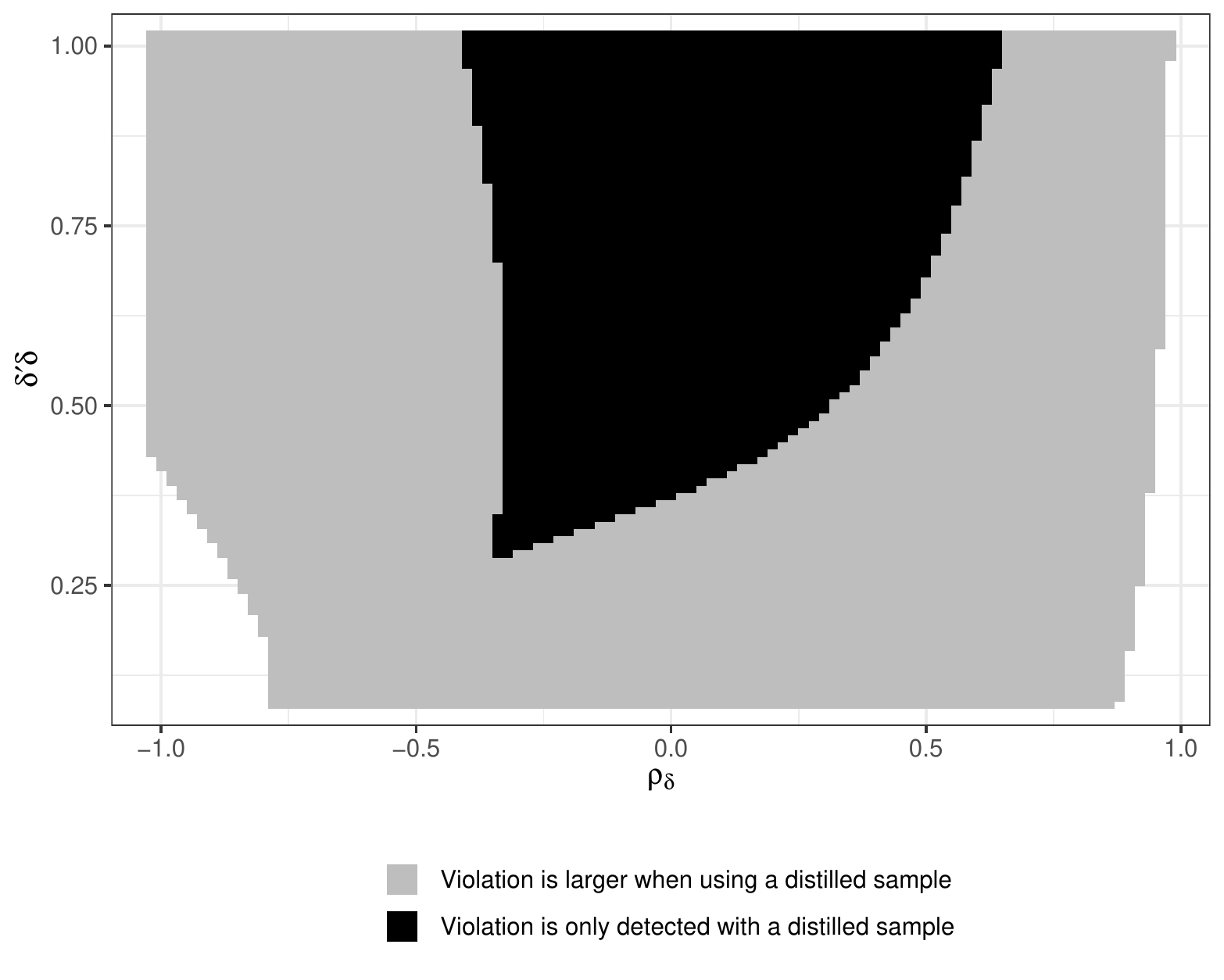}
    \caption{The region of parameter values where the violation of the nesting equality with a distilled sample exceeds the violation of the nesting equality with a coarsened propensity score, and the region where only the nesting equality with a distilled sample is violated. Each point corresponds to a combination of $\rho_{\delta}$ and $\delta^{\prime}\delta$. Let $\mathcal{N}(\rho_{\delta}, \delta^{\prime}\delta)$ denote the violation of the nesting equality with a coarsened propensity score defined in \eqref{eq:numerical example nesting violation definition}, and let $\mathcal{D}(\rho_{\delta}, \delta^{\prime}\delta)$ denote the violation of the nesting inequality with a distilled sample defined in \eqref{eq:numerical example distilled violation definition}. A point is shaded gray if $\mathcal{D}(\rho_{\delta}, \delta^{\prime}\delta) \geq  \mathcal{N}(\rho_{\delta}, \delta^{\prime}\delta) \geq 0$. A point is shaded black if $\mathcal{D}(\rho_{\delta}, \delta^{\prime}\delta) > 0$ and $\mathcal{N}(\rho_{\delta}, \delta^{\prime}\delta) < 0$.}
    \label{fig:numerical example regions}
\end{figure}

\subsection{A Distillation Procedure}\label{sec:distillation procedure}

As in Proposition 5, consider data such that $p \in P^{-} \cup P^{+}$. Let $i$ index observations. Sort observations $i = 1,...,N$ into ascending order by propensity score, so that observation $i = 1$ has the lowest propensity score and observation $i=N$ the highest. Split ties by ordering observations with $Z=0$ before observations with $Z=1$. Let $S_{1,i}$ be the sample inclusion indicator for observation $i$. First-order stochastic dominance holds when
\begin{equation}\label{eq:first order stochastic dominance constraints}
    \frac{\sum_{i=1}^j S_{1,i} Z_i}{\sum_{i=1}^n S_{1,i} Z_i} \leq \frac{\sum_{i=1}^j S_{1,i} (1 - Z_i)}{\sum_{i=1}^n S_{1,i} (1 - Z_i)} \quad \text{for} \: j=1,...,N.
\end{equation}
These constraints restrict the conditional propensity score CDF for $Z=1$ to never exceed the conditional CDF for $Z=0$.

It is immediate from \eqref{eq:first order stochastic dominance constraints} that
\begin{enumerate}
    \item $S_{1,i} = 0$ for all $i$ with $Z_i = 1$ and $p(X_i, Z_i) < \underset{j}{\text{min}}(p(X_j,Z_j):Z_j = 0)$,
    \item $S_{1,i} = 0$ for all $i$ with $Z_i = 0$ and $p(X_i, Z_i) > \underset{j}{\text{max}}(p(X_j,Z_j):Z_j = 1)$.
\end{enumerate}
That is, the observation in the trimmed sample with the lowest propensity score must have $Z=0$ and the observation with the highest propensity score must have $Z=1$. In all the applications considered in Section \ref{sec:applications}, no further trimming is necessary. Assume that we begin with a sample that has been trimmed in this manner, so that $Z_i=0$ for $i = 1$ and $Z_i=1$ for $i = N$. Let $n_0$ be the number of observations with $Z = 0$ and $n_1$ be the number with $Z=1$. Define $\Delta_j$ to be the difference between the conditional propensity score distributions for $Z = 1$ and $Z = 0$ at point $j$ 
\begin{equation*}
    \Delta_j \equiv \frac{\sum_{i=1}^j Z_i}{n_1} - \frac{\sum_{i=1}^j (1 - Z_i)}{n_0}.  
\end{equation*}

The algorithm is as follows
\begin{enumerate}
    \item Set $S_{1i} = 1 \; \forall \; i$
    \item Calculate $ \Delta_j$ for $j = 1,...,N$. If $\text{max}(\Delta_j) \leq 0$, stop. If $\text{max}(\Delta_j) > 0$, continue to steps 3 and 4.
    \item Let $J^{-} = \{j:p_j \in P^{-}\}$
    \begin{enumerate}
        \item Find 
        \begin{align*}
            j^- &= \underset{j \in J^{-}}{\text{argmax}} \;d_{1j},\\
            d_{1} &= \underset{j \in J^{-}}{\text{max}} \; d_{1j},
    \end{align*}
    where
    \begin{equation*}
        d_{1j} = \text{ceil}\left( \frac{n_0 }{n_0 - \sum_{i=1}^j(1-Z_i)} n_1 \Delta_j\right),
    \end{equation*}
    \item Loop over $j = 1,...,j^-$. If $Z_j = 1$
     \begin{enumerate}
         \item Calculate
         \begin{equation*}
              \delta_{1j} = \frac{\sum_{i=1}^j S_{1,i} Z_i}{n_1 -  d_{1}} - \frac{\sum_{i=1}^j (1 - Z_i)}{n_0}.
         \end{equation*}
         \item If $\delta_{1j} > 0$, set $S_{1j} = 0$
     \end{enumerate}
    \end{enumerate}
   \item Let $J^{+} = \{j:p_j \in P^{+}\}$
   \begin{enumerate}
       \item  Find
        \begin{align*}
            j^+ &= \underset{j \in J^{+}}{\text{argmax}} \;   d_{0j}, \\
            d_{0} &= \underset{j \in J^{+}}{\text{max}} \;   d_{0j},
        \end{align*}
        where
        \begin{equation*}
            d_{0j} =  \text{ceil}\left(n_0 \frac{n_1 - d_{1}}{\sum_{i=1}^j Z_i - d_{1}} \left(\Delta_j - \frac{d_{1}}{n_1 - d_{1}} \frac{n_1 - \sum_{i = 1}^j Z_i}{n_1}\right) \right).
        \end{equation*}
     \item Loop over $j = N, ..., j^+ + 1$. If $Z_j = 0$
        \begin{enumerate}
            \item Calculate
            \begin{equation*}
              \delta_{0j} =  \frac{\sum_{i=j}^N S_{1,i}(1 - Z_i)}{n_0 - d_{0}} - \frac{\sum_{i=j}^N  Z_i}{n_1 - d_{1}}.
            \end{equation*}
            \item If $\delta_{0j} > 0$, set $S_{1j} = 0$
        \end{enumerate}
    \end{enumerate}
\end{enumerate}

The algorithm begins by including all possible observations. In step 2, it checks if first-order stochastic dominance holds with all possible observations included. If stochastic dominance would not hold, we proceed to steps 3 and 4. Step 3 considers observations with $p \in P^{-}$. In step 3a, the algorithm calculates for each index $j$ the number of observations with $i \leq j$ and $Z_i = 1$ that need to be trimmed for there to be no violation of stochastic dominance at $p(X_j, Z_j)$. Taking the maximum gives the overall number of observations that must be trimmed, $d_1$, and the index by which this trimming must take place, $j^-$. Step 3b then selects the observations to be trimmed. Looping over $j$ from 1 to $j^-$, an observation is trimmed whenever its inclusion would lead the trimmed propensity score distribution for $Z = 1$ to cross the propensity score distribution for $Z = 0$. Step 4 is analogous to step 3, but instead considers observations with $p \in P^{+}$. Step 4a calculates the number of observations with $i \geq j$ and $Z_i = 1$ that need to be trimmed for there to be no violation of stochastic dominance at $p(X_j, Z_j)$. Taking the maximum of this gives a number of observations to trim, $d_0$, and an index above which this trimming must take place $j^+$. In step 4b, an observation is trimmed whenever its inclusion would lead the trimmed complementary propensity score distribution for $Z = 0$ to cross the trimmed complementary distribution for $Z = 0$.


Figure \ref{fig:subsample simple example} presents an example with artifical data. Panel \subref{fig:subsample simple cdfs} plots the empirical distributions of the propensity scores conditional on $Z$. First-order stochastic dominance is violated as the distribution for $Z = 1$ lies above the distribution for $Z = 0$ for some $p$. Panel \subref{fig:subsample simple d1} plots $d_{1j}$ for each $j \in J^{-}$. In step 3a, $d_1$ and $j^{-}$ are found by taking the maximum. In this case, $d_1 = 6$. Panel \subref{fig:subsample simple P^-} shows step 3b. An observation is trimmed whenever the trimmed propensity score distribution for $Z = 1$ would rise above the distribution for $Z = 0$. Panel \subref{fig:subsample simple d0} plots $d_{0j}$ for each $j \in J^{+}$. In step 4a, $d_0$ and $j^+$ are found by taking the maximum. In this case $d_0 = 11$. Panel \subref{fig:subsample simple P^+} shows step 4b. Here an observation is trimmed whenever the trimmed complementary propensity score distribution for $Z = 0$ would rise above the trimmed complementary distribution for $Z = 1$. Panel \subref{fig:subsample simple trimmed cdfs} shows the resulting trimmed distributions, where first-order stochastic dominance holds as required by distillation.

\begin{figure}
    \centering
    
    \begin{subfigure}[b]{0.49 \textwidth}
    \centering
    \includegraphics[width = \textwidth]{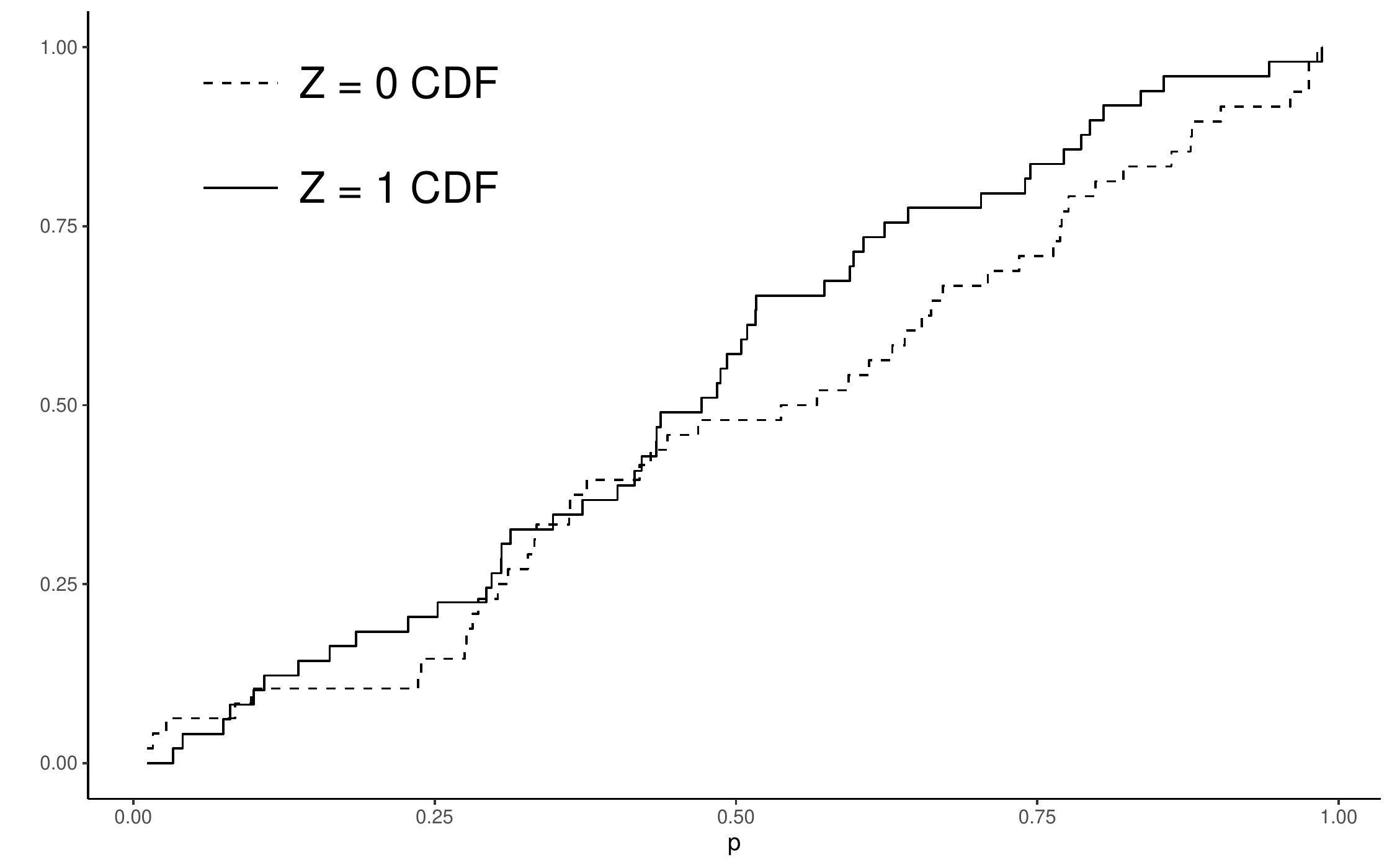}
    \caption{Empirical distributions}\label{fig:subsample simple cdfs}
    \end{subfigure}
    \begin{subfigure}[b]{0.49 \textwidth}
    \centering
    \includegraphics[width = \textwidth]{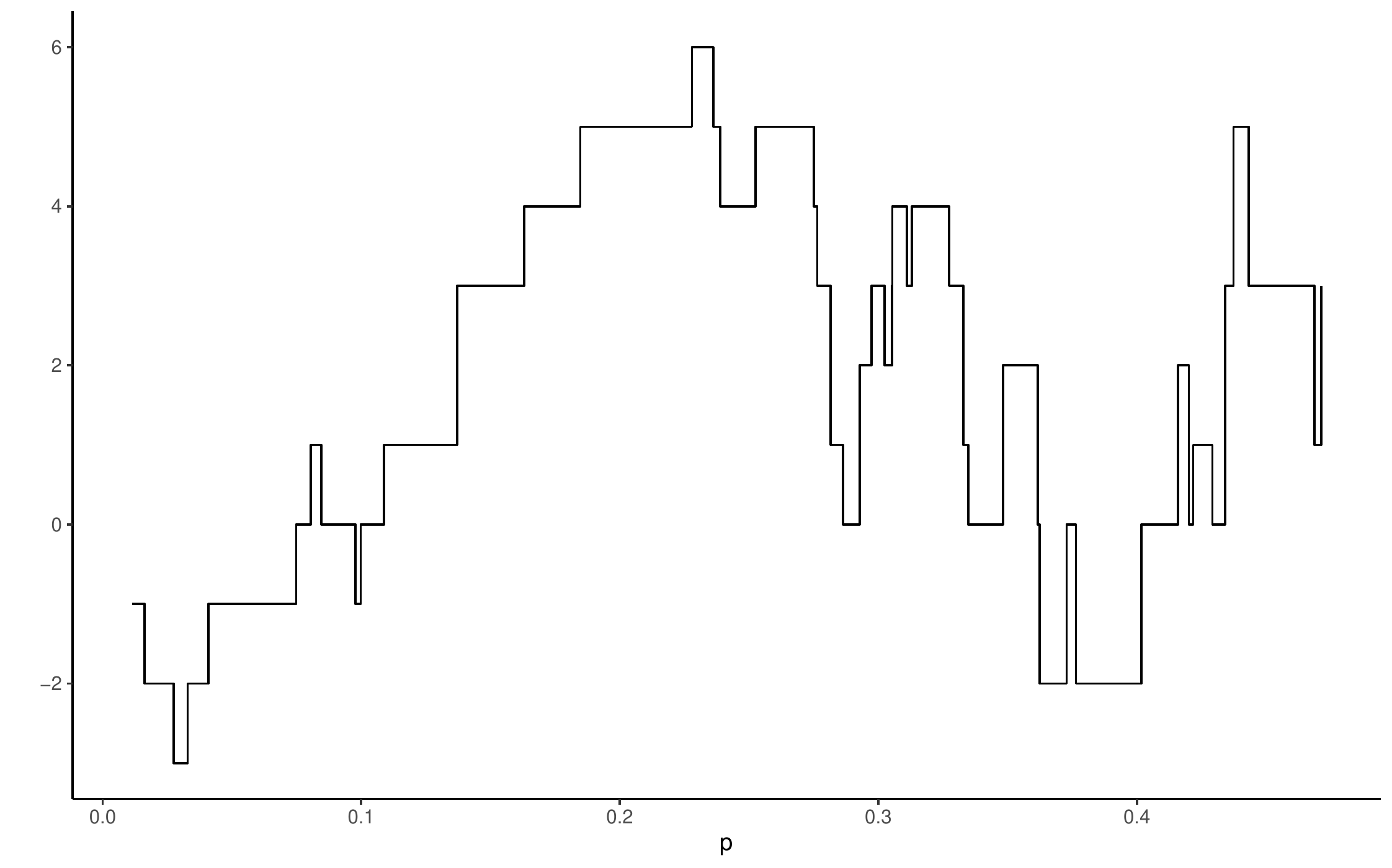}
    \caption{Step 3a}\label{fig:subsample simple d1}
    \end{subfigure}\\
    
    \begin{subfigure}[b]{0.49 \textwidth}
    \centering
    \includegraphics[width = \textwidth]{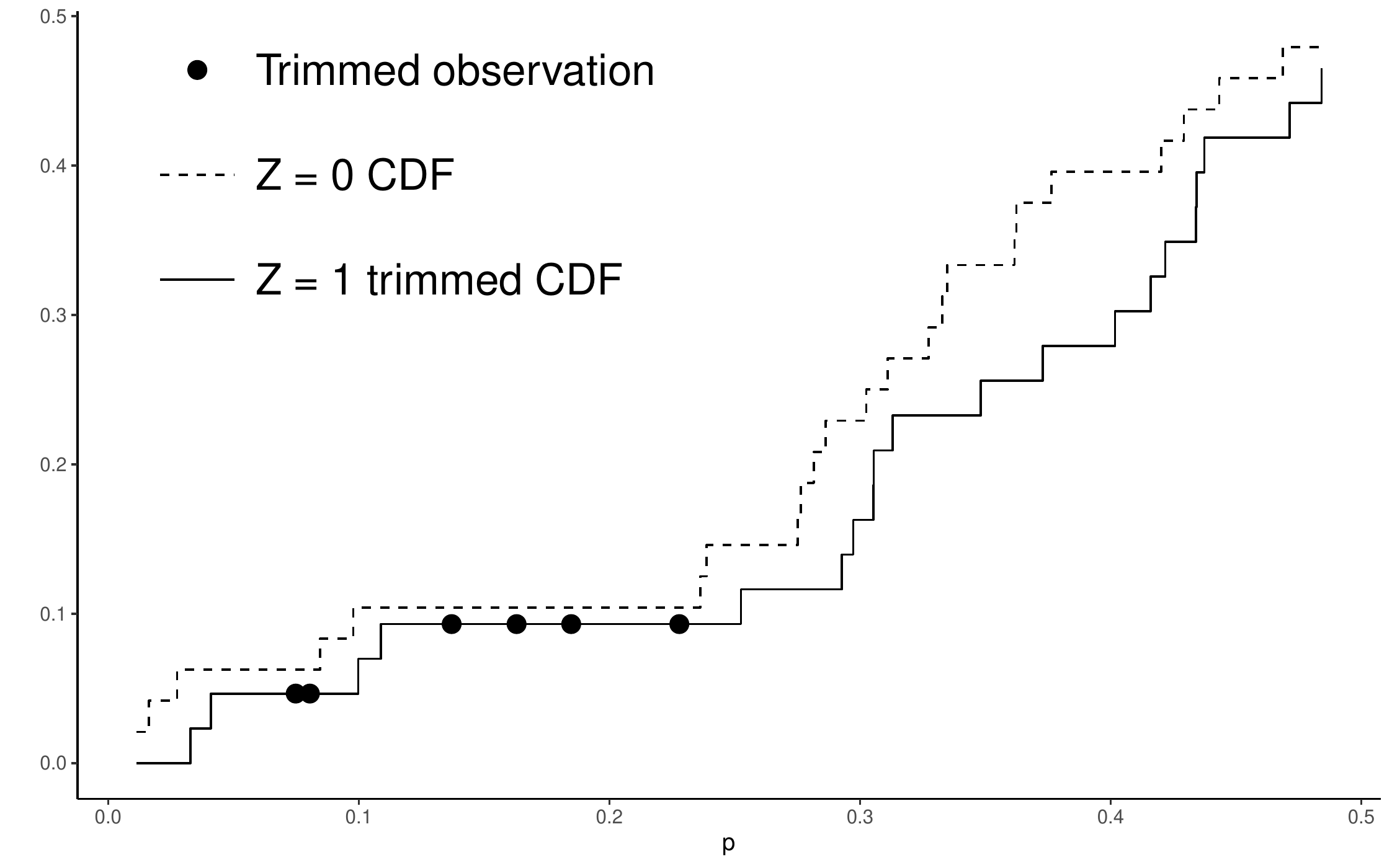}
    \caption{Step 3b}\label{fig:subsample simple P^-}
    \end{subfigure}
    \begin{subfigure}[b]{0.49 \textwidth}
    \centering
    \includegraphics[width = \textwidth]{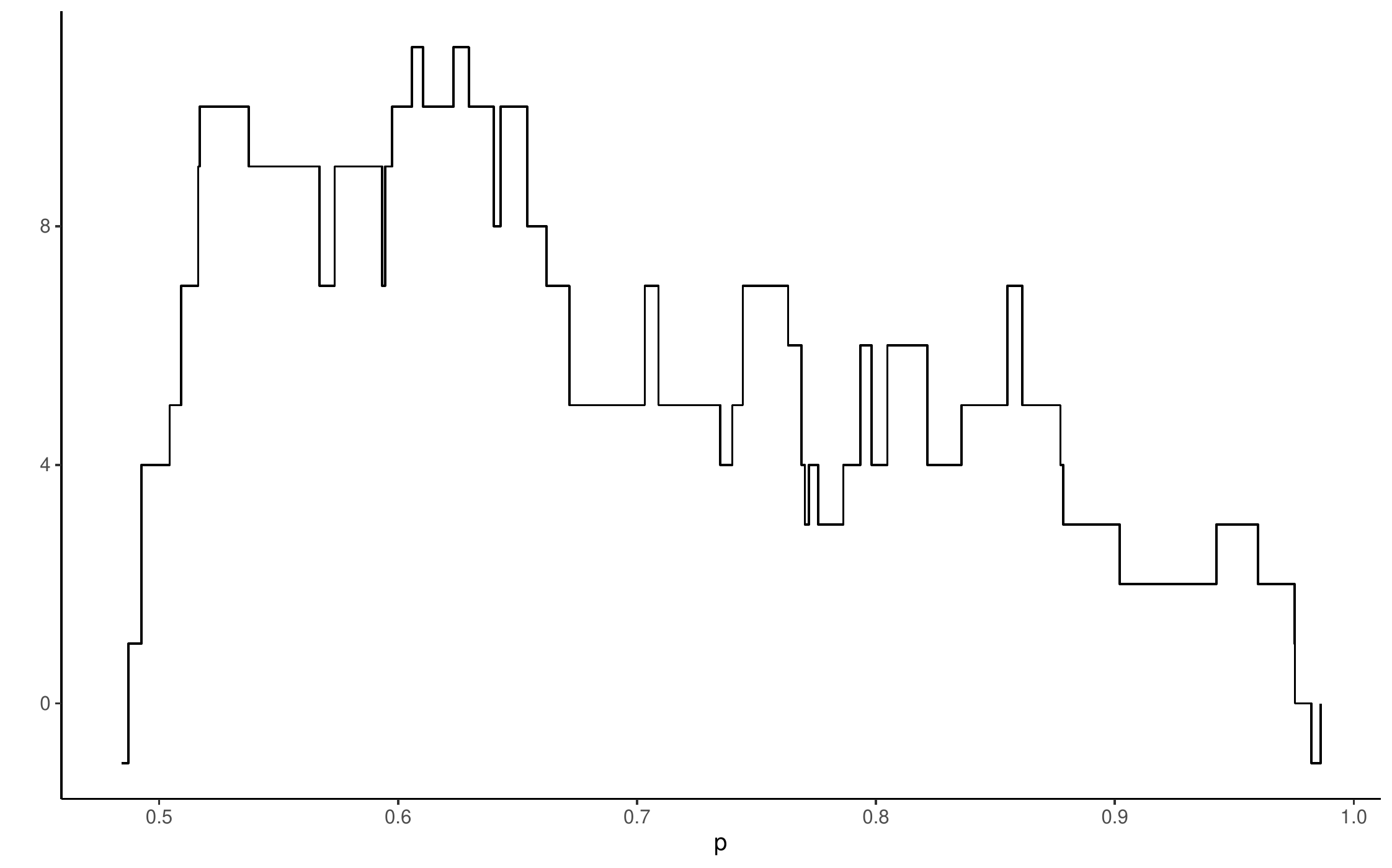}
    \caption{Step 4a}\label{fig:subsample simple d0}
    \end{subfigure}\\
     
    \begin{subfigure}[b]{0.49 \textwidth}
    \centering
    \includegraphics[width = \textwidth]{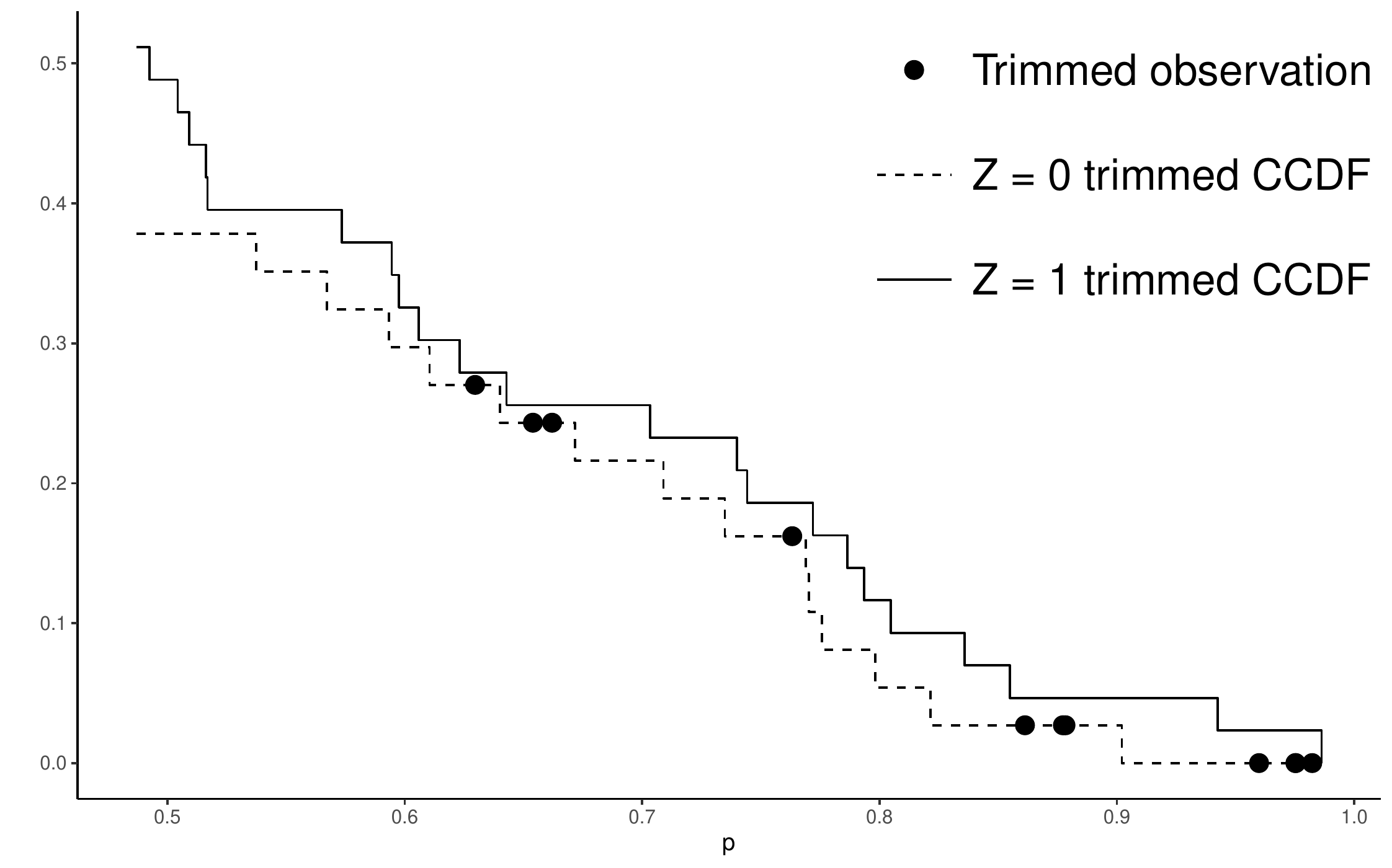}
    \caption{Step 4b}\label{fig:subsample simple P^+}
    \end{subfigure}
      \begin{subfigure}[b]{0.49 \textwidth}
    \centering
    \includegraphics[width = \textwidth]{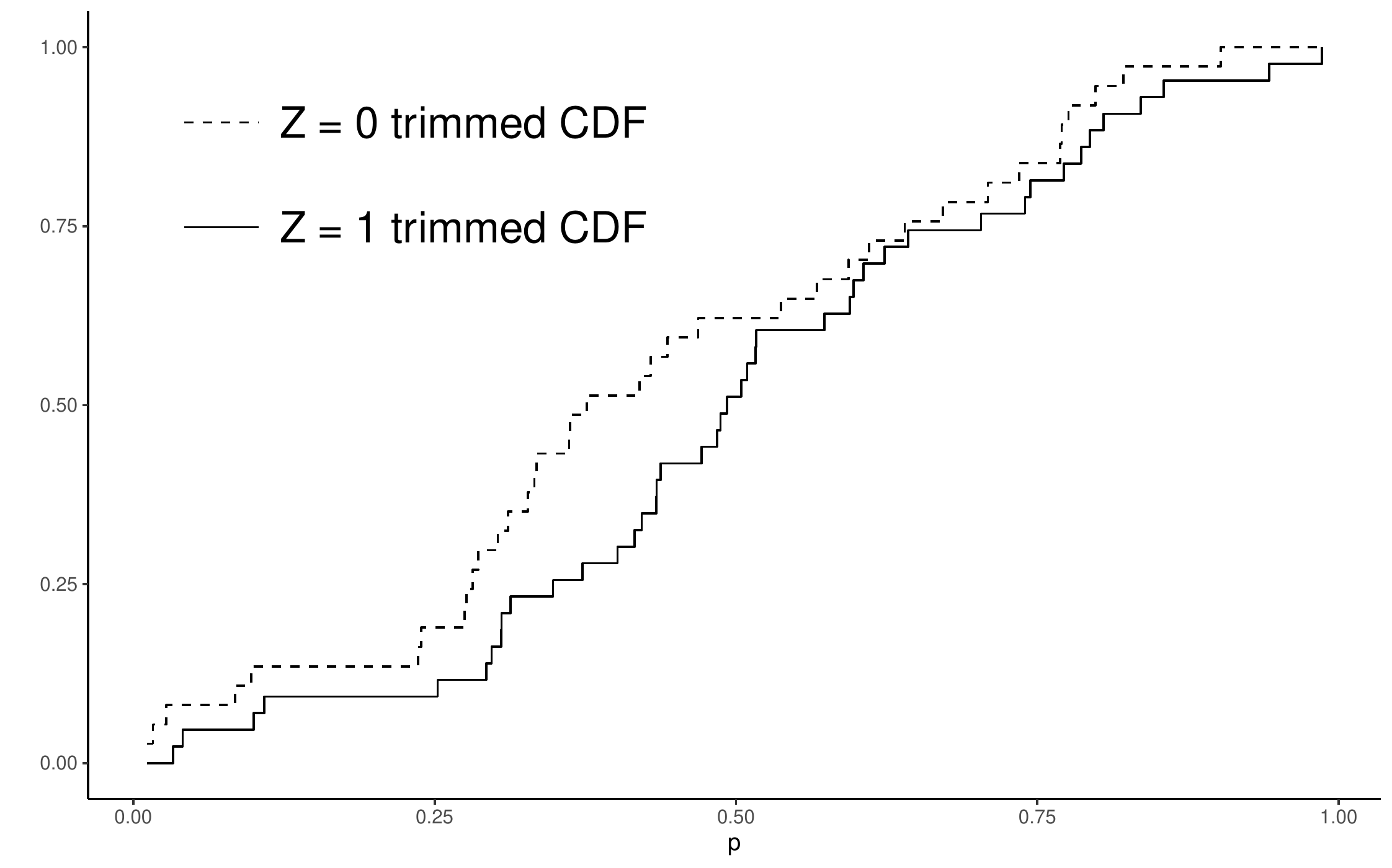}
    \caption{Trimmed distributions}\label{fig:subsample simple trimmed cdfs}
    \end{subfigure} 
    
    \caption{Example of the distillation algorithm. The sample consists of 100 observations. The data generating process sets $\Pr(Z = 1) = 1 / 2$ with propensity score uniformly distributed between 0 and 1 for $Z = 0, 1$. The interval $P^-$ includes all observations with a propensity weakly less than the median, while $P^+$ includes all observations with a propensity score above the median. Panel \subref{fig:subsample simple d1} plots $d_{1j}$ for all $j \in J^{-}$. Panel \subref{fig:subsample simple d0} plots $d_{0j}$ for all $j \in J^+$.} \label{fig:subsample simple example}
\end{figure}


This algorithm satisfies  Definition \ref{def:distilledIV} as first order stochastic dominance is guaranteed to hold for the trimmed sample, and trimming is performed based on propensity score values only. In addition, as all observations with $Z_i = 0$ and $p \in P^{-}$ or $Z_i = 1$ and $p \in P^{+}$ are retained
\begin{equation*}
    \Pr(P \leq p|Z=0, S_1 = 1, p \in P^{-}) = \Pr(P \leq p|Z=0, p \in P^{-})
\end{equation*}
\begin{equation*}
    \Pr(P \leq p|Z=0, p \in P^{+}) = \Pr(P \leq p|Z=0, S_1 = 1, p \in P^{+})
\end{equation*}
which is sufficient to satisfy condition (ii) of Proposition \ref{prop:distillation better}. 

In general there will be many trimmed samples that satisfy Definition \ref{def:distilledIV} and condition (ii) of Proposition \ref{prop:distillation better}. To select a particular sample, in this algorithm we choose the observations that are trimmed in order to make the propensity score distributions for $Z = 0$ and $Z = 1$ as similar as possible. This can be motivated heuristically as follows: variation in subdensities driven by propensity score can mask variation in the subdensities driven by $Z$. Hence, we wish to compare subdensities for similar propensity score distribution.  

As the second step of the algorithm does not take into account any subsequent trimming of observations with $Z_i = 0$, it may trim more observations with $Z_i = 1$ than is necessary to attain stochastic dominance. Appendix \ref{sec:distillation algorithm details} details how the algorithm may be modified to avoid unnecessary trimming. The additional step essentially determines $d_0$ and $d_1$ by iterating using the algorithm above, so satisfies the conditions of Definition \ref{def:distilledIV} and condition (ii) of Proposition \ref{prop:distillation better} in the same manner. 

\section{Test Statistics and Procedure}   

We set the null hypothesis for our test to be the joint restrictions of index sufficiency (\ref{eq:index sufficiency statement}) and the nesting inequalities with a distilled sample (\ref{eq:nesting_inequalities_distilled_IV}) shown in Proposition \ref{Distilled Instrument Testable Implications}. This section presents the construction of our test statistics and an implementation procedure with a bootstrap algorithm for computing p-values. Our exposition of the test here is restricted to a binary instrument $Z \in \{1,0\}$, while the support of covariates $\mathcal{X}$ is unconstrained. Appendix \ref{sec:test procedure multi-valued discrete instrument} presents a test procedure for a single multi-valued discrete instrument. 

\subsection{Test statistics and bootstrap algorithm}\label{sec:test statistic}

Let $\lambda = \Pr(Z=1) \in (0,1)$ and $(Z,S_1) \in \{0,1 \} \times \{0,1 \}$ be a binary instrument and a sample inclusion indicator as defined in Definition \ref{def:distilledIV}. With binary $Z$, we can rewrite the nesting inequalities with a distilled sample (\ref{eq:nesting_inequalities_distilled_IV}) and index sufficiency of Proposition \ref{prop2}
as the following moment inequalities and equalities: for every measurable subset $A \subset \mathcal{Y}$,
\begin{align}
    &E[1\{ U_i \in A \}D_i | Z= 0, S_1 = 1 ] - E[1\{ U_i \in A \}D_i | Z= 1, S_1 = 1 ] \leq  0, \label{eq:moment_ineq_1} \\
    &E[1\{ U_i \in A \}(1-D_i) | Z= 1, S_1 = 1 ] - E[1\{ U_i \in A \}(1-D_i) | Z = 0, S_1 = 1 ] \leq  0, \label{eq:moment_ineq_2} \\
    & E\left[ \frac{1 \{U_i \in A  \} D_i (1- \lambda)}{ \Pr(Z=0|p(X_i,Z_i))} |Z=0, S_2 = 1 \right] - E\left[ \frac{1 \{U_i \in A  \} D_i \lambda}{ \Pr(Z=1|p(X_i,Z_i))} |Z=1, S_2 = 1 \right] = 0,\label{eq:moment_eq_1} \\
    & E\left[ \frac{1 \{U_i \in A  \} (1-D_i) (1- \lambda)}{ \Pr(Z=0|p(X_i,Z_i))} |Z=0, S_2 = 1 \right] - E\left[ \frac{1 \{U_i \in A  \} (1-D_i) \lambda}{ \Pr(Z=1|p(X_i,Z_i))} |Z=1, S_2 = 1 \right] = 0, \label{eq:moment_eq_2}
\end{align}
where $S_2$ is a sample inclusion indicator based on the value of the inverse probability weighting term $\Pr(Z=1| p(X_i, Z_i))$. For example, $S_2 = 1$ if $\Pr(Z=1| p(X_i, Z_i)) \in [ 0.05, 0.95]$ and $S_2 = 0$ otherwise. Trimming in this manner avoids dividing by near zero probabilities 
which regularizes the statistical behaviour of the sample analogues of these expectations.

The inequalities (\ref{eq:moment_ineq_1}) and (\ref{eq:moment_ineq_2}) express the nesting inequalities of Proposition \ref{Distilled Instrument Testable Implications} in terms of the conditional expectations given $(Z,S_1=1)$. 
Using Bayes rule, the moment equality hypotheses (\ref{eq:moment_eq_1}) and (\ref{eq:moment_eq_2}) express the index sufficiency restrictions of Proposition \ref{prop2} in terms of the conditional expectation given $(Z,S_2=1)$. 
Trimming the sample through $S_2$ does not affect validity of the index sufficiency restrictions insofar as $S_2=1$ depends only on the value of the propensity score.

Denote the conditional distributions of $(U,D,X)$ given $Z=1$ and $Z=0$ by $P$ and $Q$, respectively, and the corresponding empirical distributions by $P_{n_1}$ and $Q_{n_0}$, where $n_1$ and $n_0$ are the number of observations with $Z=1$ and $Z=0$, respectively. Following the notational conventions of empirical process theory (e.g., \citet{vandervaart1996}), we denote the expectation of a function of $(U,D,X,S_1)$, $f$, with respect to a generic distribution, $P$, by $Pf$ and the sample analogue (sample average) by $P_{n_1}f$. Similarly, for  a function $g$ of $(U,D,X,S_2)$, the expectation with respect to a generic distribution $P$ and its sample analogue are denoted by $Pg$ and $P_{n_1}g$, respectively. 

Let $\mathcal{C}(\mathcal{Y})$ be the class of closed and connected intervals in $\mathcal{Y}$. Define
\begin{align*}
    f_i^{z}(A,d) &= \frac{1\{U_i \in A, D_i = d \} S_{1i}}{\Pr(S_1 = 1 | Z=z )}, \\
    g_i^z(A,d) &= \frac{1\{U_i \in A, D_i = d \} \lambda^z(1-\lambda)^{1-z} S_{2i}}{ \Pr(Z=z|p(X_i,Z_i))\Pr(S_2 = 1 | Z=z )}.
\end{align*}
Consider the following test statistic,
\begin{align*}
    T \equiv & \max \left\{ \sup_{A \in \mathcal{C}(\mathcal{Y}),d \in \{0 , 1 \}} \frac{T_1(A,d)}{\hat{\sigma}_1(A,d) \vee \xi}, \sup_{A \in \mathcal{C}(\mathcal{Y}),d \in \{0 , 1 \}} \frac{T_2(A,d)}{\hat{\sigma}_2(A,d) \vee \xi}, \sup_{A \in \mathcal{C}(\mathcal{Y}), d \in \{0 , 1 \}} \frac{-T_2(A,d)}{\hat{\sigma}_2(A,d) \vee \xi} \right\}, \ \text{where} \\
    T_1(A,d) = & \sqrt{\frac{n_1 n_0}{n_1+n_0}} \left[ \left( Q_{n_0} f_i^0(A,1) - P_{n_1} f_i^1(A,1) \right) d + (P_{n_1} f_i^1(A,0) - Q_{n_0} f_i^0(A,0) )(1-d)  \right] \\
    T_2(A,d) = & \sqrt{\frac{n_1 n_0}{n_1+n_0}}  \left( Q_{n_0} g_i^0(A,d) - P_{n_1} g_i^1(A,d) \right) \\
    \hat{\sigma}_1^2(A,d) = & \left\{ \lambda[ Q_{n_0} (f_i^0(A,1))^2 - ( Q_{n_0} f_i^0(A,1))^2] + (1-\lambda)[ P_{n_1} (f_i^1(A,1))^2 - ( P_{n_1} f_i^1(A,1))^2] \right\} d \\
    & + \left\{ \lambda[ Q_{n_0} (f_i^0(A,0))^2 - ( Q_{n_0} f_i^0(A,0))^2] + (1-\lambda)[ P_{n_1} (f_i^1(A,0))^2 - ( P_{n_1} f_i^1(A,0))^2] \right\} (1-d), \\
    \hat{\sigma}_2^2(A,d) = & \lambda[ Q_{n_0} (g_i^0(A,d))^2 - ( Q_{n_0} g_i^0(A,d))^2] + (1-\lambda)[ P_{n_1} (g_i^1(A,d))^2 - ( P_{n_1} g_i^1(A,d))^2]. \\
\end{align*}
Statistic $T$ is a variance-weighted Kolmogorov-Smirnov test statistic that jointly tests (\ref{eq:moment_ineq_1}) - (\ref{eq:moment_eq_2}) by searching for a maximal violation of their variance-weighted sample analogues. Here, the denominators of the terms appearing in the max operator are consistent estimators for the asymptotic standard deviations of the numerators. $\xi$ is a user-specified trimming constant which serves to bound the denominator away from zero. Similar to the test of \citet{Kitagawa2015}, searching over the closed and connected intervals $\mathcal{C}(\mathcal{Y})$ suffices for detecting violations in any measurable set in $\mathcal{Y}$. 


Keeping the standard deviation estimators fixed, we want to resample $T_1(\cdot, \cdot)$ and $T_2(\cdot, \cdot)$ from a distribution such that the null hypothesis holds and their variance-covariance structure is preserved. Specifically, we consider a multiplier bootstrap: let $(\hat{M}_1, \dots, \hat{M}_{n})$ be iid random bootstrap multipliers such that they are independent of the original sample and satisfy $E(\hat{M}_i) = 0$ and $Var(\hat{M}_i) = 1$. A bootstrap analogue of $T$ under a least favorable null can be constructed as
\begin{align*}
    \hat{T} \equiv & \max \left\{ \sup_{A  \in \mathcal{C}(\mathcal{Y}),d \in \{0 , 1 \}} \frac{\hat{T}_1(A,d)}{\hat{\sigma}_1(A,d) \vee \xi}, \sup_{A \in \mathcal{C}(\mathcal{Y}),d \in \{0 , 1 \}} \frac{\hat{T}_2(A,d)}{\hat{\sigma}_2(A,d) \vee \xi}, \sup_{A \in \mathcal{C}(\mathcal{Y}),d \in \{0 , 1 \}} \frac{-\hat{T}_2(A,d)}{\hat{\sigma}_2(A,d) \vee \xi}\right\}, \ \text{where} \\
    \hat{T}_1(A,d) = & \sqrt{\frac{n_1 n_0}{n_1+n_0}} \left( (\hat{Q}_{n_0} - Q_{n_0}) f_i^0(A,1) - (\hat{P}_{n_1}- P_{n_1}) f_i^1(A,1) \right) d \\ 
    & + \sqrt{\frac{n_1 n_0}{n_1 + n_0}} \left( (\hat{P}_{n_1} - P_{n_1}) f_i^1(A,0) - (\hat{Q}_{n_0} - Q_{n_0} ) f_i^0(A,0) \right) (1-d) \\
    \hat{T}_2(A,d) = & \sqrt{\frac{n_1 n_0}{n_1+n_0}}  \left[ ( \hat{Q}_{n_0} - Q_{n_0}) g_i^0(A,d) - (\hat{P}_{n_1} - P_{n_1}) g_i^1(A,d) \right],
    \end{align*}
where, for the random variable $(a_i: i=1, \dots, n)$, we define
\begin{align*}
(\hat{Q}_{n_0} - Q_{n_0}) a_i & = \frac{1}{n_0} \sum_{i: Z_i=0} \hat{M}_i a_i, \mspace{10mu}
(\hat{P}_{n_1} - P_{n_1}) a_i = \frac{1}{n_1} \sum_{i: Z_i=1} \hat{M}_i a_i.
\end{align*}

\subsubsection{Test Procedure}\label{sec:test procedure}


An implementation of the test is as follows.

\begin{enumerate}
\item Estimate the propensity scores $\Pr (D=1|Z,X)$ using either
parametric or nonparametric methods, and obtain the fitted values $\hat{p}\left( x_{i},z_{i}\right) $.

\item Estimate the partially linear regression model (\ref{reg y on p and x})
and obtain estimates of $\theta _{1}$ and $\theta _{0}$. \ Specifically, let $%
\hat{\mu}_{x,i}=\hat{E}\left( x_{i}|\hat{p}\left( x_{i},z_{i}\right) \right) 
$ be nonparametric regression estimates (e.g. local linear predicted values) of the
covariate vector $x_{i}$ onto the estimated propensity scores, $\hat{p}%
_{i}\equiv \hat{p}\left( x_{i},z_{i}\right) ,$ and let $\hat{\mu}_{y,i}=\hat{E}%
\left( y_{i}|\hat{p}\left( x_{i},z_{i}\right) \right) $ be a nonparametric
regression estimate of the outcome onto the estimated propensity scores. \ $\left( \theta _{0},\theta _{1}\right) $ can be estimated by running OLS on%
\begin{equation*}
y_{i}-\hat{\mu}_{y,i}=\hat{p}_{i}\left( x_{i}-\hat{\mu}_{x,i}\right) \theta
_{1}+\left( 1-\hat{p}_{i}\right) \left( x_{i}-\hat{\mu}_{x,i}\right) \theta
_{0}+\epsilon _{i}.
\end{equation*}

\item Construct the residual observations by $\hat{u}_{i}=D_{i}\left(
y_{i}-x_{i}^{\prime }\hat{\theta}_{1}\right) +(1-D_{i})\left(
y_{i}-x_{i}^{\prime }\hat{\theta}_{0}\right) $.
\item Construct the inclusion indicators $S_1$. This can be done using the algorithm of Section 2.5.
\item Estimate $\Pr(Z=1|p(X,Z))$ and generate the inclusion indicators $S_2$
\item Using $\left( \hat{u}_{i},D_{i}, z, \hat{\Pr}(Z=1|p(X,Z), S_{1,i}, S_{2,i}\right) $ as data, calculate the test statistic of Section \ref{sec:test statistic}
\end{enumerate}


The test procedure described above involves semiparametric estimation of a partially linear model where each conditioning covariate and the outcome are non-parametrically regressed on the propensity score. While non-parametric estimation can be computationally taxing, this step is only performed once and the main determinant of the computational burden is the sample size rather than the number of conditioning covariates. In cases where this non-parametric estimation step is infeasible, the test procedure can still be implemented by specifying a parametric functional form for $\phi(p)$, such as a quadratic polynomial in the propensity score. 

\subsection{Practical Considerations}

\subsubsection{Testing Causal Intepretability of 2SLS}

The test proposed here does not directly test instrument validity in the context of linear 2SLS. This is because the set of identifying assumptions considered in Proposition \ref{prop2} is distinct from the set of conditions that allows us to interpret the linear two-stage least square (2SLS) estimand as a weighted average of LATEs/MTEs with positive weights. 
Specifically, causal interpretability of linear 2SLS requires only the weaker exogeneity condition (A2), not the strong exogeneity of (A5), and does not require the functional form specification for the potential outcome equations of (A4). On the other hand, as shown in \citet{Abadie_2003}, \citet{Kolesar_2013}, \citet{sloczynski2021}, and \citet{Blandhol_etal_2022}, causal interpretability of 2SLS relies crucially on the linearity of $E(Z|X)$ in the covariate vector $X$, whereas the identifying assumptions here considered place no restrictions on $E(Z|X)$.   

However, we believe our test is useful as a specification check when one reports linear 2SLS estimates for the following reasons. 
First, non-rejection in our test means the data do not reject strong exogeneity (A5), implying that the data also do not contradict weak exogeneity (A2). With credible evidence for the linearity of $E(Z|X)$, non-rejection of our test can be used to support causal interpretability of linear 2SLS.
Second, with specifications of the potential outcome equations and the propensity score, our test is useful for checking which variables should be included as controlling covariates. Hence, p-values of our test can also suggest which set of covariates should be included in the linear 2SLS estimation.

\subsubsection{Continuous Instruments}

The test procedure can accommodate continuous instruments as follows. The initial estimation of a partially linear model can be performed using the continuous instrument. With estimated partial residuals in hand, the instrument can then be discretized, and the test procedure for a multi-valued discrete instrument of Appendix \ref{sec:test procedure multi-valued discrete instrument} applied. 


\subsubsection{Multiple Instruments}\label{subsec:multipleIV}

The discussion to this point has considered only a scalar $Z$, here we consider the case where $Z$ is a vector of multiple instruments. Propositions \ref{prop: joint-ineq} and \ref{prop2} hold irrespective of the number of instruments. Hence, both testable implications remain available. However, a test of the nesting inequalities using a coarsened propensity score would be subject to the same concerns regarding power as the single instrument case.  

Consider the case of two binary instruments $Z = (Z_1, Z_2)$, $Z_1 \in \{0,1\}$, $Z_2 \in \{0,1\}$. When testing nesting inequalities and index sufficiency, $Z$ can be treated as a single instrument taking four values. That is, the estimation step is performed taking $Z$ to be vector, but when testing we map from $(0,0), (0,1), (1,0), (1,1)$ to four values of a single categorical variable and implement the test for a multi-valued instrument described in Appendix \ref{sec:test procedure multi-valued discrete instrument}. In cases where this is not viable, one approach is to aggregate the multiple instruments into a single index, and then treat this index as single multivalued instrument. Suppose the propensity score has a generalised linear form with a probit link
\begin{equation}
    p(X,Z) = \Phi(\rho X + \varphi(Z)),
\end{equation}
where $\varphi(\cdot)$ reflects the researchers assumptions about the relationship between the instruments and $Z$ and $\rho$ is a vector of coefficients. In this case, we view $\varphi(Z)$ as an aggregated index of instruments, and perform the test as if it was a single instrument.


Our testing approach can be extended to the setting of \citet{Mogstadetal2019, Mogstadetal2020} where the montonicity assumption \citet{imbensangrist} is replaced by partial monotonicity. \citet{Mogstadetal2019, Mogstadetal2020} show that inference using a single instrument remains valid as long as the remaining instruments are controlled for. In the context of our test procedure, the estimation of the propensity score and residuals would be performed using the entire set of instruments, but the final step would be performed separately for each instrument. That is, index sufficiency and the nesting inequalities are tested for every instrument, with a distilled sample used in the test of nesting inequalities. A Bonferroni correction can be applied to P-values to correct for multiple tests. 

\section{Monte Carlo}\label{sec:MC}

We perform two sets of Monte Carlo exercises. The first evaluates the size and power of the joint test of nesting inequalities with a distilled sample and index sufficiency, and compares it to the \citet{Kitagawa2015} test which does not control for covariates and the test of the nesting inequalities using a coarsened propensity score. The second looks separately at the power of the tests of index sufficiency and nesting inequalities with a distilled sample.  

All exercises use
\begin{align}
    &X \sim \text{N}_3(0_3, I_3) \label{eq:monte carlo X} \\
    &Z = \mathbbm{1}\{X^{\prime} \gamma + U_Z \geq 0\} \label{eq:MC Z function}\\
    &D = \mathbbm{1}\{\alpha_0 (1 - Z) + \alpha_1 Z + X^{\prime} \delta + U_D \geq 0\} \label{eq:monte carlo selection}\\
    &Y = D Y_1 + (1 - D)Y_0 \\
    &Y_0 = X^{\prime} \theta + \tilde{U}_0 \\
    &U_Z \sim N(0,1)\\
    &(\tilde{U}_0, U_D)\sim N(0, \Sigma), \label{eq:monte carlo UZ UD}
\end{align}
with diagonal entries of $\Sigma$ set to 1 and off-diagonal entries to 0.3. Elements of $\theta$ are drawn from the uniform distribution with bounds $(-1,1)$. The distributions for $\gamma$ and $\delta$ are described below. $\alpha_0$, $\alpha_1$, and the specification for $Y_1$ differ for the size and power processes. When checking size we consider a valid a but irrelevant instrument
\begin{align}
    &\alpha_0 = \alpha_1 = 0, \\
    &Y_1 = X^{\prime} \theta + 1 + \tilde{U}_1, \ \text{with $\tilde{U}_1 = \tilde{U}_0$.}
\end{align}
Here, conditional on the propensity score, the joint distribution of $(U, D)$ conditional on $Z$ is identical for $Z=0$ and $Z=1$. Four specifications are used to check power. They share
\begin{align}
    &\alpha_0 = \Phi^{-1}(0.45), \\
    &\alpha_1 = \Phi^{-1}(0.55), \\
    &Y_1 = X^{\prime} \theta + \tilde{U}_1,
\end{align}
where $\Phi^{-1}()$ is the standard normal quantile function. The specifications for $\tilde{U}_1$ are
\begin{alignat*}{2}
    &\text{DGP1:}  & \quad &  \tilde{U}_1 = Z \times \tilde{U}_0 + (1 - Z) \times (-0.7 + \tilde{U}_0), \\
    &\text{DGP2:}  & \quad &  \tilde{U}_1 =  Z \times \tilde{U}_0 + (1 - Z) \times 1.675 \times \tilde{U}_0, \\
    &\text{DGP3:}  & \quad &  \tilde{U}_1 = Z \times \tilde{U}_0 + (1 - Z) \times  0.515 \times \tilde{U}_0, \\
    &\text{DGP4:}  & \quad &  \tilde{U}_1 = Z \times \tilde{U}_0 + (1 - Z) \times (\mu_l + 0.125 \times \tilde{U}_0),
\end{alignat*}
where $\mu_l$ takes values $(-1,-0.5,0,0.5,1)$ with probabilities $(0.15, 0.2, 0.3, 0.2, 0.15)$. Figure \ref{fig:Monte Carlo true residuals} plots the densities of $(\tilde{U}_1,D = 1|Z)$ for each of the four processes for one draw of parameters. In the case of DGP1 the distribution for $Z=0$ is shifted, with DGP2 it has thicker tails, with DGP3 it is more concentrated, and with DGP4 there are violations at each value of $\mu_l$. In all cases, for a given propensity score, the distribution of $(\tilde{U}_1, D=1)$ conditional on $Z$ does not coincide for $Z=0$ and $Z=1$. 

We consider three values of the sample size $N$: 200, 500, and 1000. For each process and sample size, we generate 1,000 samples. When comparing test procedures, the same simulated samples are used for each test procedure. The proposed test procedure requires estimating propensity scores, a partially linear model for outcomes, and $\Pr(Z=1|p(X,Z))$. Propensity scores are estimated by probit with the correct specification supplied. For the nonparametric regressions involved in the estimation of the partially linear model we use local linear regressions. For the estimation of $\Pr(Z=1|p(X,Z))$, we use local constant regressions. In both cases bandwidths are chosen by least squares cross validation. The partially linear model is estimated under the assumption that the instrument is valid, so the specification is
\begin{equation*}
    Y = X^{\prime}\theta + \phi(p) + \epsilon.
\end{equation*}
In the case DGP1-DGP4, this specification is incorrect, so estimates will be biased.\footnote{For the size process, this specification is correct. However, as the instrument is irrelevant, the model is not identified. In particular, it can be shown that asymptotically there is perfect multicollinearity among the columns of $X - E[X|p]$. In the finite samples used in these Monte Carlo exercises it is still possible to calculate an estimate of $\theta$. This estimate reflects finite sample noise, so the associated partial residuals can be used to check size.} We set the sample inclusion indicator for the index sufficiency test to $S_2 = \mathbbm{1} \{ \Pr(Z=1|p(X,Z)) \in [0.05, 0.95]\}$. Results are reported for two values of the trimming parameter $\xi$: $\sqrt{0.05 \times 0.95} \approx 0.21$ and 0.3. Results for additional values of $\xi$ are available upon request. P-values are calculated using 500 bootstrap samples. 

\subsection{Size and Power}\label{sec:monte carlo exercise 1}

Here we evaluate the size and power of the proposed test procedure, and compare it to the \citet{Kitagawa2015} test that does not control for covariates and the test of nesting inequalities with a binarised propensity score. 

For this exercise, the elements of $\delta$ are drawn from the uniform distribution with bounds $(-1,1)$. For $\gamma$, we consider two cases. The first sets all elements of $\gamma$ to 0. $X$ and $Z$ are then independent, so an instrument that is valid conditional on the covariates will also be valid without conditioning on covariates and failure to control appropriately for $X$ should not affect size. In the second case each element of $\gamma$ is drawn from the uniform distribution with bounds $(-1,1)$. Here $X$ and $Z$ will not, in general, be independent, and the instrument can only be valid conditional on $X$.

Tables \ref{tab:size} and \ref{tab:power} present rejection rates. Turning first to the proposed test procedure, rejection rates for the size process are below nominal size irrespective of whether $X$ and $Z$ are independent. For the four power DGPs, with the exception of the smaller sample sizes for DGP4, all rejection rates are well above the nominal size. Similar rejection rates are obtained irrespective of whether $X$ and $Z$ are independent.   

Next consider the test with a binarised propensity score. This test performs well in terms of size: rejection rates for the size processes are below nominal size. However, consistent with the results of Proposition \ref{prop:distillation better},  there is a considerable loss in power compared to the proposed test procedure. Rejection rates above nominal size are obtained only for DGP1. For all other DGPs, rejection rates are low and change little as the sample size increases.

Figures \ref{fig:Monte Carlo Distilled} and \ref{fig:Monte Carlo binarised} illustrate why the test with a binarized propensity score lacks power. To construct these figures, for each of the four power DGPs, we generate a single draw for $\theta$ and $\delta$. To simplify, $\gamma$ is set to the zero vector. We approximate the asymptotic estimates of $\theta$,  $\hat{\theta}^{\star}$, by generating a sample of a million observations and regressing $Y - E[Y|p]$ on $X - E[X|p]$. Partial residuals are then calculated as $U = Y - X^{\prime}\hat{\theta}^{\star}$. Figure \ref{fig:Monte Carlo Distilled} plots the densities of  $(U, D = 1)$ conditional on $Z = 0$ and $Z = 1$, whereas Figure \ref{fig:Monte Carlo binarised} plots the subdensities conditional on the binarised propensity score $Z^{*}$. In Figure \ref{fig:Monte Carlo Distilled}, when we condition on $Z$, violations of the null are apparent in all cases, but in Figure \ref{fig:Monte Carlo binarised} violations are either small or have vanished entirely. Each value of $Z^{*}$ mixes observations with $Z=0$ and $Z=1$. The subdensities in Figure \ref{fig:Monte Carlo binarised} are thus a mixture of those in Figure \ref{fig:Monte Carlo Distilled}, which masks differences in shape. Furthermore, by construction, all observations with $Z^{*}=0$ have a lower propensity score than observations with $Z^{*}=1$, so the $D=1$ subdensity for $Z^{*}=0$ is guaranteed to have less total mass than the subdensity for $Z^{*}=1$.      

The \citet{Kitagawa2015} test, which does not control for covariates, performs reasonably well when the instrument is independent of the covariates. The size of the test is controlled, but power is lower than the proposed test procedure. This is for two reasons. First, the \citet{Kitagawa2015} test essentially tests the nesting inequalities only, and some violations are more easily detected by index sufficiency. Second, failure to control for covariates adds noise to the outcome and which makes violations of the nesting inequalities more difficult to detect. For a single draw of $\theta$ and $\delta$, and with $\gamma$ set to zero, Figure \ref{fig:Monte Carlo Y} plots the densities of $(Y, D=1)$  conditional on $Z = 0$ and $Z = 1$ for each DGP. These differ from the densities of Figure \ref{fig:Monte Carlo true residuals} as the effect of the covariates on the outcome has not been partialed out. Compared to the densities in Figure \ref{fig:Monte Carlo true residuals}, violations are either reduced or have vanished entirely. As expected, when $Z$ and $X$ are not independent, this test procedure has rejection rates above nominal size regardless of whether the instrument is valid conditional on covariates. 

\begin{table}[ht]
\resizebox{\textwidth}{!}{\begin{tabular}{@{\extracolsep{0.5pt}}c|cccccc|cccccc}
  \hline
  \hline
  \addlinespace[0.1cm]  
  \multicolumn{1}{r}{} & \multicolumn{6}{c}{$\gamma_k = 0$} & \multicolumn{6}{c}{$\gamma_k \sim U(-1,1)$} \\
   \multicolumn{1}{r}{Trimming Constant:} & \multicolumn{3}{c}{$\xi \approx 0.21$} & \multicolumn{3}{c}{$\xi = 0.30$} & \multicolumn{3}{c}{$\xi \approx 0.21$} & \multicolumn{3}{c}{$\xi = 0.30$} \\
  \cline{2-4} \cline{5-7} \cline{8-10} \cline{11-13}
  \addlinespace[0.1cm]  
   \multicolumn{1}{r}{Nominal Size:} & 0.1 & 0.05 & 0.01 & 0.1 & 0.05 & 0.01 & 0.1 & 0.05 & 0.01 & 0.1 & 0.05 & 0.01 \\ 
  \addlinespace[0.1cm]  
  \hline
  \multicolumn{13}{l}{Proposed test procedure} \\   
200 & 0.027 & 0.004 & 0.000 & 0.018 & 0.005 & 0.000 & 0.085 & 0.040 & 0.012 & 0.048 & 0.022 & 0.008 \\  
  500 & 0.024 & 0.005 & 0.000 & 0.018 & 0.004 & 0.000 & 0.068 & 0.033 & 0.009 & 0.045 & 0.022 & 0.007 \\ 
   1000 & 0.028 & 0.009 & 0.002 & 0.017 & 0.006 & 0.000 & 0.060 & 0.029 & 0.010 & 0.040 & 0.016 & 0.003 \\  
  \multicolumn{13}{l}{Binarised $p(X,Z)$}\\
 200 & 0.043 & 0.024 & 0.013 & 0.031 & 0.023 & 0.013 & 0.046 & 0.033 & 0.017 & 0.044 & 0.030 & 0.017 \\ 
  500 & 0.067 & 0.052 & 0.035 & 0.054 & 0.043 & 0.029 & 0.056 & 0.042 & 0.033 & 0.045 & 0.039 & 0.029 \\ 
  1000 & 0.075 & 0.069 & 0.051 & 0.062 & 0.053 & 0.040 & 0.096 & 0.082 & 0.060 & 0.081 & 0.068 & 0.048 \\ 
  \multicolumn{13}{l}{Kitagawa (2015)} \\   
 200 & 0.131 & 0.070 & 0.016 & 0.134 & 0.067 & 0.017 & 0.454 & 0.384 & 0.262 & 0.453 & 0.403 & 0.268 \\  
  500 & 0.135 & 0.066 & 0.015 & 0.116 & 0.063 & 0.012 & 0.524 & 0.483 & 0.403 & 0.530 & 0.485 & 0.404 \\  
  1000 & 0.115 & 0.066 & 0.012 & 0.124 & 0.054 & 0.012  & 0.645 & 0.599 & 0.540 & 0.642 & 0.608 & 0.539 \\ 
  \hline
  \end{tabular}}
  \caption{Size}\label{tab:size}
\end{table}

\begin{table}
\vspace{-1cm}
\resizebox{\textwidth}{!}{\begin{tabular}{@{\extracolsep{0.5pt}}c|cccccc|cccccc}
  \hline
  \hline
  \addlinespace[0.1cm] 
  \multicolumn{1}{r}{} & \multicolumn{6}{c}{$\gamma_k = 0$} & \multicolumn{6}{c}{$\gamma_k \sim U(-1,1)$} \\
  \multicolumn{1}{r}{Trimming Constant:} & \multicolumn{3}{c}{$\xi \approx 0.21$} & \multicolumn{3}{c}{$\xi = 0.30$} & \multicolumn{3}{c}{$\xi \approx 0.21$} & \multicolumn{3}{c}{$\xi = 0.30$}\\
  \cline{2-4} \cline{5-7} \cline{8-10} \cline{11-13}
  \addlinespace[0.1cm]  
  \multicolumn{1}{r}{Nominal Size:} & 0.1 & 0.05 & 0.01 & 0.1 & 0.05 & 0.01 & 0.1 & 0.05 & 0.01 & 0.1 & 0.05 & 0.01 \\ 
  \addlinespace[0.1cm]  
  \hline
  
 \multicolumn{13}{l}{\textbf{Power DGP1}} \\
  \multicolumn{13}{l}{Proposed test procedure} \\  
 200 & 0.164 & 0.085 & 0.020 & 0.158 & 0.078 & 0.018 & 0.228 & 0.144 & 0.059 & 0.189 & 0.113 & 0.042 \\ 
  500 & 0.446 & 0.311 & 0.104 & 0.440 & 0.309 & 0.110 & 0.454 & 0.328 & 0.154 & 0.413 & 0.289 & 0.124 \\ 
  1000 & 0.725 & 0.587 & 0.346 & 0.718 & 0.586 & 0.327 & 0.706 & 0.569 & 0.327 & 0.676 & 0.548 & 0.300 \\  
  \multicolumn{13}{l}{Binarised Propensity Score}\\
 200 & 0.053 & 0.040 & 0.012 & 0.050 & 0.040 & 0.011 & 0.049 & 0.035 & 0.013 & 0.044 & 0.028 & 0.013 \\  
  500 & 0.153 & 0.128 & 0.074 & 0.131 & 0.104 & 0.064 & 0.127 & 0.096 & 0.051 & 0.096 & 0.079 & 0.045 \\ 
  1000 & 0.313 & 0.263 & 0.192 & 0.251 & 0.214 & 0.150 & 0.266 & 0.232 & 0.158 & 0.225 & 0.183 & 0.114 \\  
  \multicolumn{13}{l}{Kitagawa (2015)} \\   
 200 & 0.205 & 0.134 & 0.039 & 0.198 & 0.123 & 0.033 & 0.419 & 0.354 & 0.249 & 0.417 & 0.358 & 0.251 \\ 
  500 & 0.359 & 0.257 & 0.122 & 0.354 & 0.255 & 0.118 & 0.583 & 0.535 & 0.450 & 0.577 & 0.530 & 0.448 \\ 
  1000 & 0.637 & 0.545 & 0.362 & 0.652 & 0.548 & 0.342 & 0.699 & 0.668 & 0.604 & 0.690 & 0.656 & 0.591 \\  
 \addlinespace[0.1cm]  
 \hline 
 \addlinespace[0.1cm]  
 \multicolumn{13}{l}{\textbf{Power DGP2}} \\
  \multicolumn{13}{l}{Proposed test procedure} \\   
200 & 0.201 & 0.116 & 0.026 & 0.206 & 0.118 & 0.031 & 0.232 & 0.146 & 0.039 & 0.200 & 0.124 & 0.033 \\ 
  500 & 0.638 & 0.493 & 0.194 & 0.648 & 0.505 & 0.195 & 0.585 & 0.432 & 0.167 & 0.529 & 0.387 & 0.145 \\  
  1000 & 0.954 & 0.908 & 0.697 & 0.955 & 0.902 & 0.684 & 0.921 & 0.858 & 0.640 & 0.885 & 0.791 & 0.532 \\
  \multicolumn{13}{l}{Binarised Propensity Score}\\
 200 & 0.026 & 0.018 & 0.008 & 0.023 & 0.019 & 0.009 & 0.039 & 0.029 & 0.014 & 0.041 & 0.026 & 0.016 \\ 
  500 & 0.036 & 0.027 & 0.017 & 0.025 & 0.021 & 0.012 & 0.043 & 0.033 & 0.020 & 0.035 & 0.029 & 0.018 \\  
  1000 & 0.040 & 0.032 & 0.018 & 0.029 & 0.021 & 0.012  & 0.047 & 0.038 & 0.025 & 0.039 & 0.032 & 0.015 \\ 
  \multicolumn{13}{l}{Kitagawa (2015)} \\   
 200 & 0.105 & 0.054 & 0.011 & 0.079 & 0.033 & 0.009 & 0.369 & 0.313 & 0.196 & 0.372 & 0.303 & 0.190 \\ 
  500 & 0.214 & 0.141 & 0.054 & 0.118 & 0.070 & 0.027 & 0.582 & 0.526 & 0.404 & 0.533 & 0.477 & 0.381 \\ 
  1000 & 0.470 & 0.366 & 0.215 & 0.290 & 0.200 & 0.082 & 0.720 & 0.680 & 0.592 & 0.650 & 0.602 & 0.518 \\ 
 
  \addlinespace[0.1cm]  
  \hline
  \addlinespace[0.1cm]
  
  \multicolumn{13}{l}{\textbf{Power DGP3}} \\
  \multicolumn{13}{l}{Proposed test procedure} \\   
200 & 0.209 & 0.116 & 0.024 & 0.185 & 0.098 & 0.020 & 0.295 & 0.189 & 0.059 & 0.223 & 0.142 & 0.047 \\  
  500 & 0.777 & 0.640 & 0.303 & 0.649 & 0.500 & 0.185 & 0.772 & 0.660 & 0.364 & 0.669 & 0.535 & 0.268 \\ 
  1000 & 0.991 & 0.980 & 0.925 & 0.976 & 0.954 & 0.827  & 0.987 & 0.977 & 0.901 & 0.967 & 0.944 & 0.784 \\  
  \multicolumn{13}{l}{Binarised Propensity Score}\\
 200 & 0.028 & 0.020 & 0.008 & 0.027 & 0.018 & 0.007 & 0.027 & 0.016 & 0.007 & 0.025 & 0.015 & 0.007 \\  
  500 & 0.018 & 0.016 & 0.012 & 0.018 & 0.015 & 0.010 & 0.025 & 0.016 & 0.010 & 0.019 & 0.016 & 0.009 \\ 
  1000 & 0.006 & 0.004 & 0.003 & 0.005 & 0.005 & 0.002  & 0.013 & 0.010 & 0.006 & 0.012 & 0.009 & 0.003 \\ 
  \multicolumn{13}{l}{Kitagawa (2015)} \\   
200 & 0.115 & 0.061 & 0.013 & 0.105 & 0.058 & 0.014 & 0.359 & 0.301 & 0.200 & 0.352 & 0.294 & 0.204 \\ 
  500 & 0.118 & 0.073 & 0.022 & 0.126 & 0.076 & 0.027 & 0.497 & 0.437 & 0.338 & 0.508 & 0.456 & 0.349 \\  
  1000 & 0.163 & 0.108 & 0.039 & 0.180 & 0.125 & 0.042 & 0.582 & 0.547 & 0.475 & 0.582 & 0.545 & 0.478 \\ 

  \addlinespace[0.1cm]  
  \hline
  \addlinespace[0.1cm]
  
   \multicolumn{13}{l}{\textbf{Power DGP4}} \\
  \multicolumn{13}{l}{Proposed test procedure} \\   
 200 & 0.052 & 0.018 & 0.002 & 0.031 & 0.011 & 0.001 & 0.130 & 0.067 & 0.013 & 0.078 & 0.038 & 0.008 \\ 
  500 & 0.302 & 0.177 & 0.036 & 0.125 & 0.056 & 0.007 & 0.326 & 0.212 & 0.061 & 0.194 & 0.110 & 0.027 \\ 
  1000 & 0.732 & 0.616 & 0.364 & 0.412 & 0.278 & 0.068  & 0.710 & 0.609 & 0.340 & 0.471 & 0.312 & 0.098 \\ 
  \multicolumn{13}{l}{Binarised Propensity Score}\\
  200 & 0.022 & 0.015 & 0.007 & 0.017 & 0.011 & 0.007 & 0.023 & 0.017 & 0.009 & 0.021 & 0.014 & 0.009 \\ 
  500 & 0.021 & 0.018 & 0.012 & 0.018 & 0.015 & 0.011 & 0.040 & 0.028 & 0.015 & 0.022 & 0.019 & 0.012 \\
  1000 & 0.016 & 0.012 & 0.007 & 0.011 & 0.009 & 0.005 & 0.028 & 0.021 & 0.016 & 0.019 & 0.015 & 0.012 \\ 
  \multicolumn{13}{l}{Kitagawa (2015)} \\   
 200 & 0.081 & 0.040 & 0.003 & 0.066 & 0.029 & 0.004 & 0.346 & 0.288 & 0.176 & 0.358 & 0.291 & 0.186 \\ 
  500 & 0.080 & 0.035 & 0.007 & 0.071 & 0.042 & 0.005 & 0.451 & 0.399 & 0.318 & 0.453 & 0.406 & 0.328 \\ 
  1000 & 0.076 & 0.044 & 0.010 & 0.078 & 0.046 & 0.012 & 0.527 & 0.485 & 0.418 & 0.532 & 0.494 & 0.424 \\ 
  \hline
  \end{tabular}}
  \caption{Power}\label{tab:power}
\end{table}

\subsection{Comparison of Index Sufficiency and Nesting Inequalities}\label{sec:monte carlo exercise 2}

Our second exercise compares the nesting inequality and index sufficiency tests. Section 2 argued that the relative power of testing nesting inequalities and index sufficiency will depend on similarity between propensity score distributions conditional on $Z$. To examine this more formally, we compare rejection rates foe the nesting inequalities and index sufficiency in two cases. Both set all elements of $\gamma$ to 0, so that the instrument and covariates are independent. In the first, the elements of $\delta$ are drawn from the uniform distribution with bounds $(-1,1)$, as in section \ref{sec:monte carlo exercise 1}. In the second, the elements of $\delta$ are drawn from the uniform distribution with bounds $(-0.1,0.1)$. This second case limits the influence of covariates on the propensity score, which in turn limits the overlap of the conditional propensity score distributions for $Z = 0$ and $Z = 1$. Figure \ref{fig:Monte Carlo pscore dens comp} plots the density of propensity scores conditional on $Z = 0$ and $Z = 1$ for different values of $\delta$. As the elements of $\delta$ grow larger in magnitude, the densities overlap more. Figure \ref{fig:Monte Carlo PZ1 comp} plots the $\Pr(Z=1|p(X,Z))$ function associated with each value of $\delta$. For the smallest $\delta$ this function is close to a step function but, as the elements of $\delta$ grow in magnitude, it becomes flatter.   

Table \ref{tab:nesting index comparison} presents rejection rates. When the elements of $\delta$ are drawn from the uniform distribution with bounds $(-1,1)$, index sufficiency has higher rejection rates. When the elements of $\delta$ are drawn from the uniform distribution with bounds $(-0.1,0.1)$, index sufficiency loses power in all cases. Nesting inequalities now have considerably higher power for DGP1, slightly higher power for DGP2 and DGP3, while index sufficiency still exhibits more power for DGP4. In all cases, the overall rejection rate is close to the maximum of the rejection rates for the nesting inequalities and index sufficiency.  This suggests that testing both testable implications jointly dominates testing only one of them.   

\begin{table}
\vspace{-2cm}
\centering
\resizebox{\textwidth}{!}{\begin{tabular}{c|cccccc|cccccc}
  \hline
  \hline
  \addlinespace[0.1cm]  
   \multicolumn{1}{r}{} & \multicolumn{6}{c}{$\delta_k \sim U(-1,1)$} & \multicolumn{6}{c}{$\delta_k \sim U(-0.1,0.1)$} \\
   \multicolumn{1}{r}{Trimming Constant:} & \multicolumn{3}{c}{$\xi \approx 0.21$} & \multicolumn{3}{c}{$\xi = 0.30$}  & \multicolumn{3}{c}{$\xi \approx 0.21$} & \multicolumn{3}{c}{$\xi = 0.30$} \\
  \cline{2-4} \cline{5-7} \cline{8-10} \cline{11-13}
 \multicolumn{1}{r}{Nominal Size:} & 0.1 & 0.05 & 0.01 & 0.1 & 0.05 & 0.01 & 0.1 & 0.05 & 0.01 & 0.1 & 0.05 & 0.01 \\ 
  \hline
   \multicolumn{13}{l}{\textbf{Power DGP 1}}\\
  \multicolumn{13}{l}{Nesting:}\\
200 & 0.075 & 0.036 & 0.004 & 0.067 & 0.033 & 0.003 & 0.092 & 0.044 & 0.006 & 0.081 & 0.039 & 0.003 \\ 
  500 & 0.256 & 0.160 & 0.040 & 0.251 & 0.157 & 0.034 & 0.401 & 0.275 & 0.087 & 0.386 & 0.261 & 0.076 \\ 
  1000 & 0.558 & 0.434 & 0.236 & 0.529 & 0.409 & 0.213 & 0.859 & 0.762 & 0.495 & 0.840 & 0.756 & 0.449 \\   
  \multicolumn{13}{l}{Index Sufficiency:}\\
 200 & 0.174 & 0.090 & 0.023 & 0.173 & 0.088 & 0.022 & 0.103 & 0.050 & 0.013 & 0.085 & 0.043 & 0.008 \\  
  500 & 0.460 & 0.318 & 0.111 & 0.445 & 0.310 & 0.115 & 0.196 & 0.104 & 0.015 & 0.167 & 0.086 & 0.016 \\ 
  1000 & 0.708 & 0.574 & 0.349 & 0.710 & 0.573 & 0.328 & 0.316 & 0.216 & 0.086 & 0.330 & 0.214 & 0.074 \\  
  \multicolumn{13}{l}{Overall:}\\
 200 & 0.164 & 0.085 & 0.020 & 0.158 & 0.078 & 0.018 & 0.101 & 0.057 & 0.007 & 0.084 & 0.046 & 0.003 \\ 
  500 & 0.446 & 0.311 & 0.104 & 0.440 & 0.309 & 0.110 & 0.363 & 0.211 & 0.066 & 0.321 & 0.205 & 0.056 \\ 
  1000 & 0.725 & 0.587 & 0.346 & 0.718 & 0.586 & 0.327 & 0.808 & 0.718 & 0.453 & 0.786 & 0.692 & 0.403 \\  
   \addlinespace[0.1cm]  
  \hline
  \addlinespace[0.1cm]
  
 \multicolumn{13}{l}{\textbf{Power DGP 2}}\\
 \multicolumn{13}{l}{Nesting:}\\
200 & 0.048 & 0.019 & 0.000 & 0.025 & 0.007 & 0.000 & 0.063 & 0.024 & 0.000 & 0.027 & 0.006 & 0.000 \\ 
  500 & 0.315 & 0.186 & 0.053 & 0.137 & 0.074 & 0.013 & 0.380 & 0.227 & 0.059 & 0.176 & 0.085 & 0.019 \\  
  1000 & 0.816 & 0.675 & 0.375 & 0.474 & 0.305 & 0.082 & 0.886 & 0.791 & 0.505 & 0.619 & 0.442 & 0.180 \\ 
  \multicolumn{13}{l}{Index Sufficiency:}\\
 200 & 0.224 & 0.126 & 0.033 & 0.227 & 0.130 & 0.035 & 0.148 & 0.075 & 0.011 & 0.158 & 0.074 & 0.010 \\ 
  500 & 0.656 & 0.510 & 0.214 & 0.661 & 0.519 & 0.209  & 0.406 & 0.266 & 0.095 & 0.402 & 0.272 & 0.097 \\ 
  1000 & 0.951 & 0.911 & 0.705 & 0.961 & 0.912 & 0.702 & 0.769 & 0.673 & 0.403 & 0.748 & 0.634 & 0.372 \\  
   \multicolumn{13}{l}{Overall:}\\
  200 & 0.201 & 0.116 & 0.026 & 0.206 & 0.118 & 0.031 & 0.126 & 0.049 & 0.006 & 0.114 & 0.055 & 0.005 \\ 
  500 & 0.638 & 0.493 & 0.194 & 0.648 & 0.505 & 0.195 & 0.426 & 0.288 & 0.093 & 0.353 & 0.223 & 0.071 \\ 
  1000 & 0.954 & 0.908 & 0.697 & 0.955 & 0.902 & 0.684 & 0.898 & 0.805 & 0.527 & 0.759 & 0.612 & 0.337 \\   
  \addlinespace[0.1cm]  
  \hline
  \addlinespace[0.1cm]
  
      \multicolumn{13}{l}{\textbf{Power DGP 3}}\\
  \multicolumn{13}{l}{Nesting:}\\
 200 & 0.122 & 0.063 & 0.010 & 0.143 & 0.074 & 0.010 & 0.116 & 0.056 & 0.012 & 0.124 & 0.063 & 0.010 \\ 
  500 & 0.385 & 0.247 & 0.077 & 0.435 & 0.278 & 0.083 & 0.382 & 0.275 & 0.097 & 0.425 & 0.302 & 0.108 \\ 
  1000 & 0.816 & 0.710 & 0.451 & 0.846 & 0.753 & 0.475 & 0.765 & 0.657 & 0.363 & 0.818 & 0.691 & 0.398 \\ 
  \multicolumn{13}{l}{Index Sufficiency:}\\
  200 & 0.232 & 0.132 & 0.027 & 0.199 & 0.110 & 0.022 & 0.157 & 0.070 & 0.014 & 0.136 & 0.069 & 0.015 \\ 
  500 & 0.792 & 0.659 & 0.326 & 0.672 & 0.522 & 0.199 & 0.462 & 0.363 & 0.158 & 0.414 & 0.296 & 0.110 \\ 
  1000 & 0.991 & 0.981 & 0.928 & 0.977 & 0.958 & 0.836 & 0.667 & 0.613 & 0.452 & 0.667 & 0.585 & 0.382 \\   
   \multicolumn{13}{l}{Overall:}\\
 200 & 0.209 & 0.116 & 0.024 & 0.185 & 0.098 & 0.020 & 0.131 & 0.061 & 0.008 & 0.114 & 0.055 & 0.009 \\ 
  500 & 0.777 & 0.640 & 0.303 & 0.649 & 0.500 & 0.185 & 0.475 & 0.367 & 0.147 & 0.434 & 0.320 & 0.111 \\ 
  1000 & 0.991 & 0.980 & 0.925 & 0.976 & 0.954 & 0.827 & 0.832 & 0.740 & 0.488 & 0.829 & 0.731 & 0.443 \\ 
   \addlinespace[0.1cm]  
  \hline
  \addlinespace[0.1cm]
  
   \multicolumn{13}{l}{\textbf{Power DGP 4}}\\
  \multicolumn{13}{l}{Nesting:}\\
 200 & 0.028 & 0.008 & 0.001 & 0.022 & 0.007 & 0.001 & 0.039 & 0.017 & 0.002 & 0.027 & 0.012 & 0.002 \\ 
  500 & 0.064 & 0.034 & 0.005 & 0.050 & 0.029 & 0.004 & 0.094 & 0.038 & 0.005 & 0.063 & 0.031 & 0.002 \\ 
  1000 & 0.170 & 0.102 & 0.030 & 0.142 & 0.077 & 0.020 & 0.231 & 0.122 & 0.031 & 0.161 & 0.078 & 0.013 \\  
  \multicolumn{13}{l}{Index Sufficiency:}\\
 200 & 0.064 & 0.021 & 0.003 & 0.040 & 0.015 & 0.002 & 0.054 & 0.026 & 0.001 & 0.050 & 0.024 & 0.002 \\  
  500 & 0.321 & 0.186 & 0.041 & 0.141 & 0.061 & 0.009 & 0.174 & 0.113 & 0.032 & 0.106 & 0.048 & 0.013 \\ 
  1000 & 0.748 & 0.630 & 0.376 & 0.431 & 0.292 & 0.082 & 0.458 & 0.372 & 0.196 & 0.284 & 0.188 & 0.065 \\ 
  \multicolumn{13}{l}{Overall:}\\
 200 & 0.052 & 0.018 & 0.002 & 0.031 & 0.011 & 0.001 & 0.055 & 0.020 & 0.002 & 0.034 & 0.018 & 0.001 \\ 
  500 & 0.302 & 0.177 & 0.036 & 0.125 & 0.056 & 0.007 & 0.161 & 0.088 & 0.023 & 0.082 & 0.035 & 0.006 \\ 
  1000 & 0.732 & 0.616 & 0.364 & 0.412 & 0.278 & 0.068 & 0.451 & 0.342 & 0.154 & 0.250 & 0.149 & 0.042 \\  
   \addlinespace[0.1cm]  
   \hline
\end{tabular}} \caption{Comparison of index sufficiency and nesting inequalities}\label{tab:nesting index comparison}
\end{table}

\begin{figure}
    \centering
    \includegraphics[trim = 50 15 50 25, width = \linewidth, keepaspectratio]{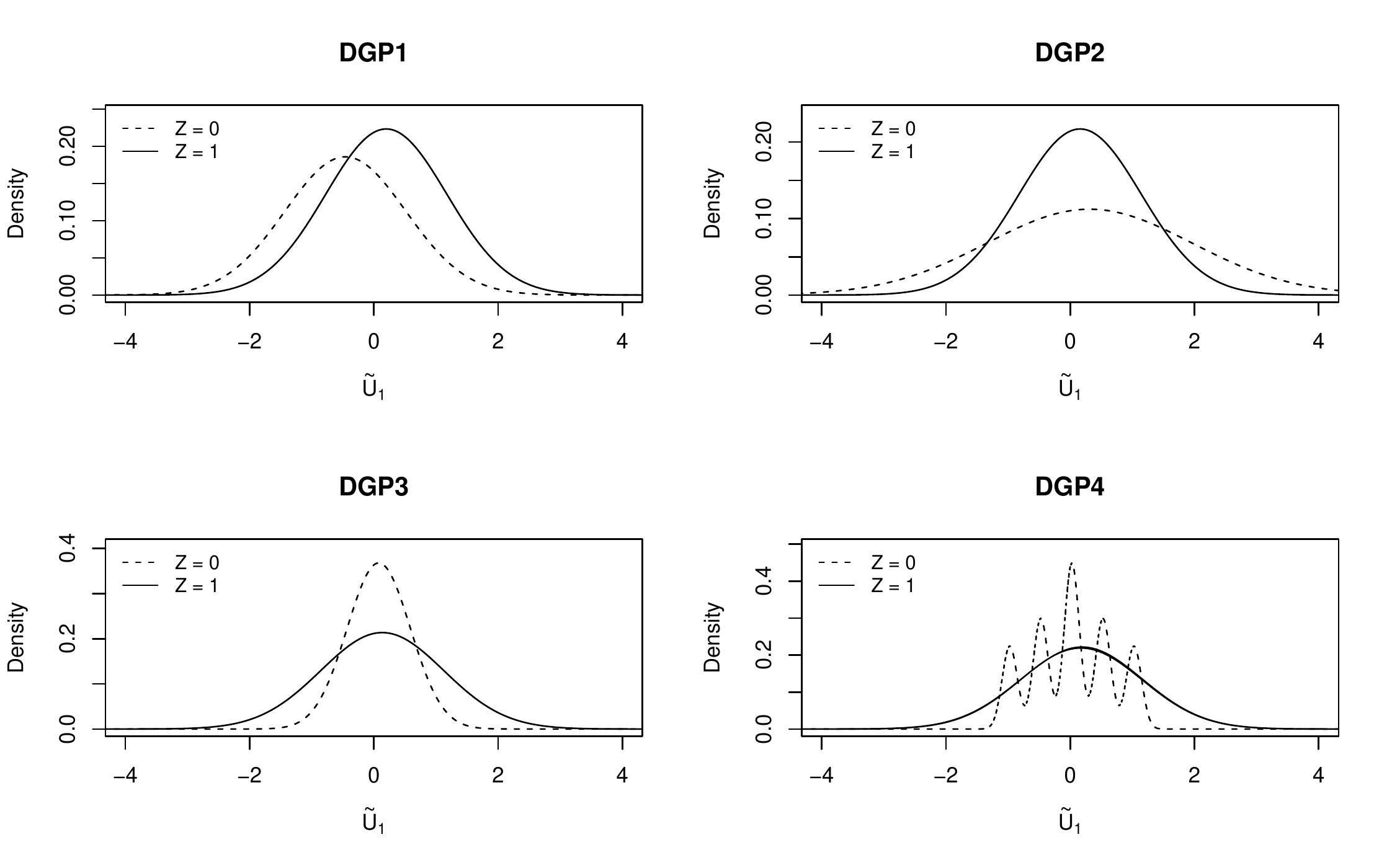}
    \caption{Conditional subdensities of $(\tilde{U}_1, D = 1)$ for DGPs 1 - 4 given values of $\theta$ and $\delta$ drawn from the $U(-1,1)$ distribution. The elements of $\gamma$ are set to 0.
    }
    \label{fig:Monte Carlo true residuals}
\end{figure}

\begin{figure}
    \centering
    \includegraphics[trim = 50 15 50 25, width = \linewidth, keepaspectratio]{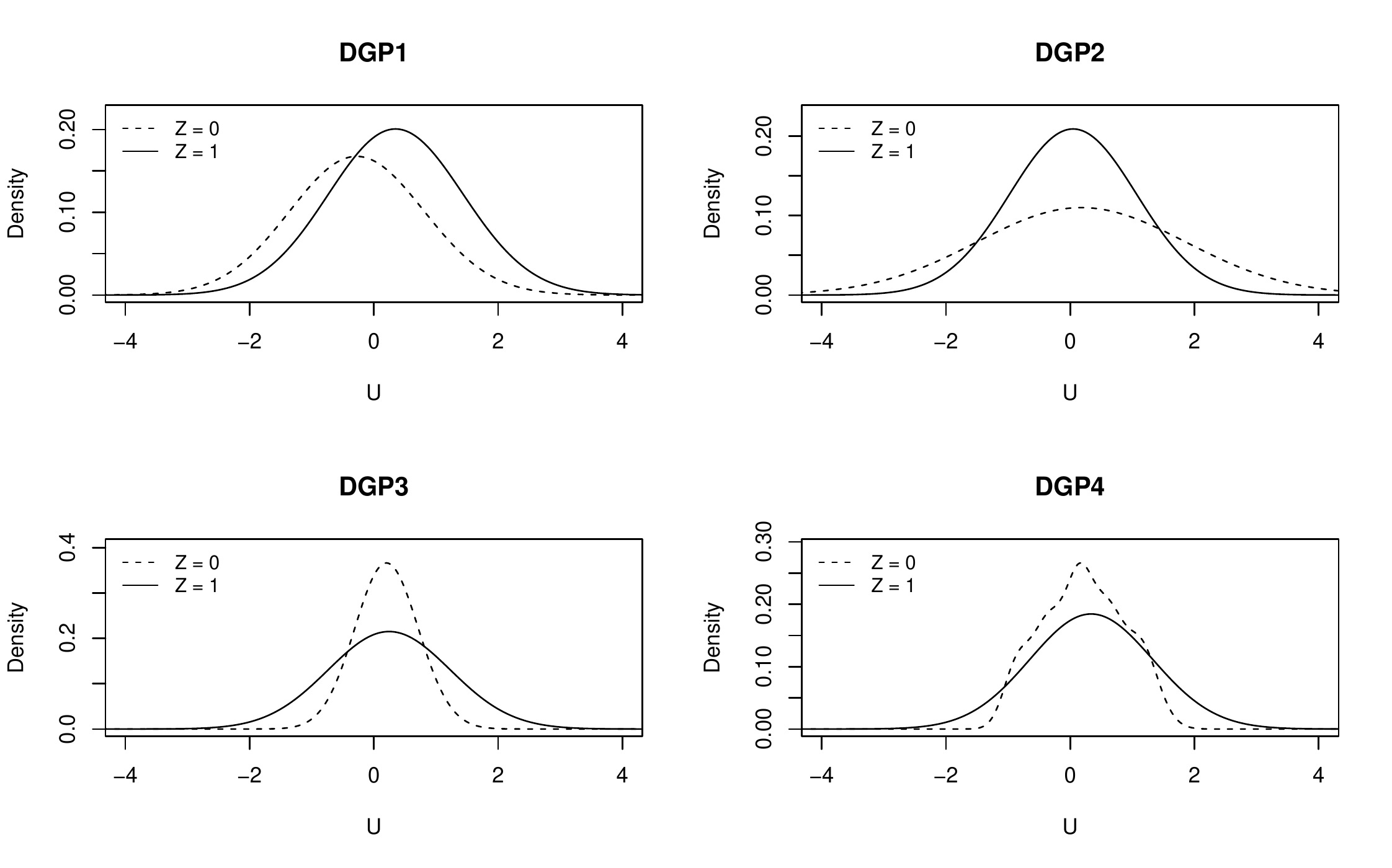}
    \caption{Subdensities of $(U, D = 1)$ for DGPs 1 - 4 conditional on the instrument $Z$, where $U = Y - X^{\prime}\beta^{*}$ with $\beta^{*} = \text{plim} \; \hat{\beta}$s of $\theta$ and $\delta$. The elements of $\beta$ are drawn from the $U(0.5, 1)$ distribution with random flips of their signs and the elements of $\delta$ are drawn from the $U(-1,1)$ distribution. All elements of $\gamma$ are set to 0. 
    }
  \label{fig:Monte Carlo Distilled}
\end{figure}

\begin{figure}
    \centering
    \includegraphics[trim = 50 15 50 25, width = \linewidth, keepaspectratio]{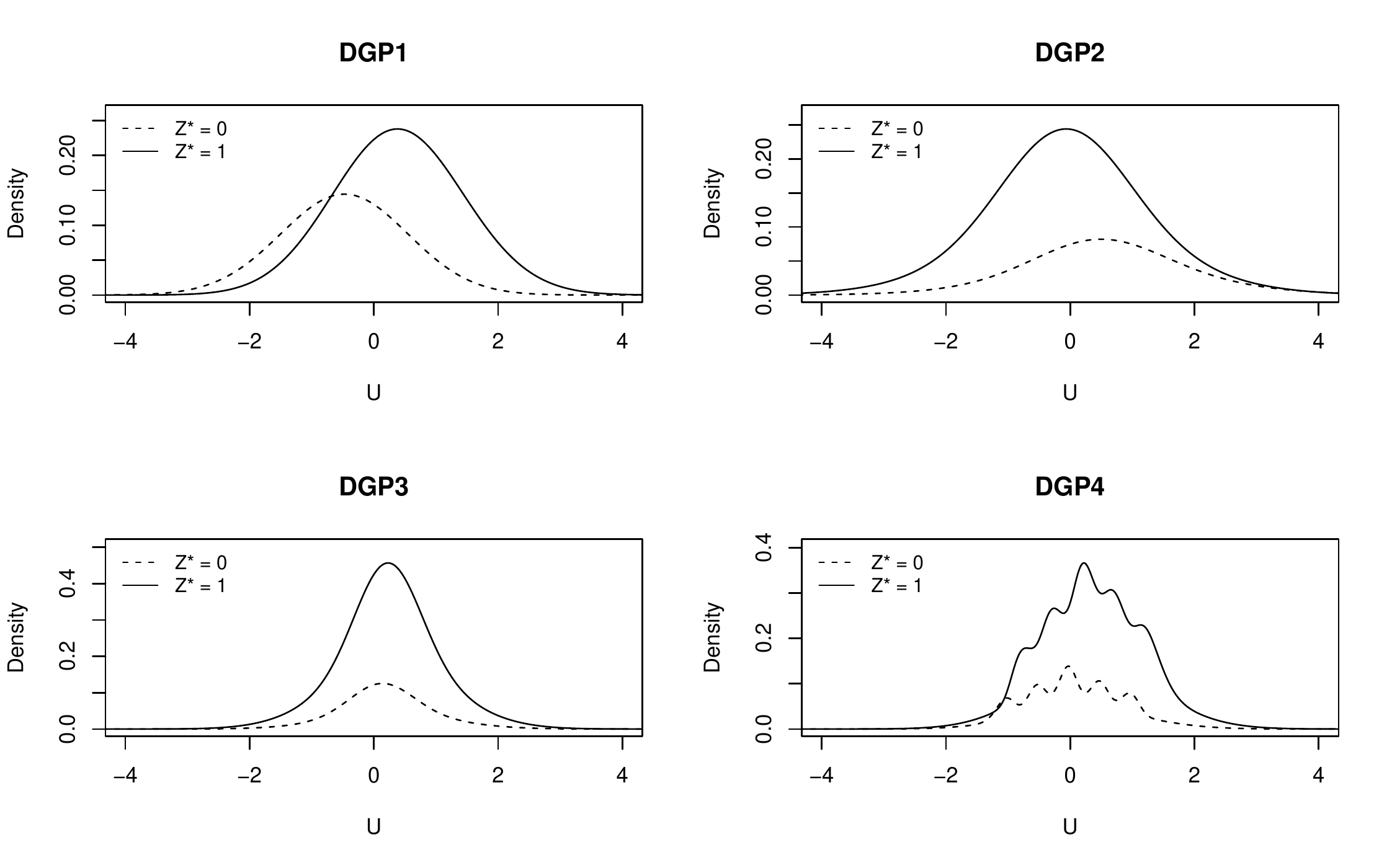}
    \caption{Subdensities of $(U, D = 1)$ for DGPs 1 - 4 conditional on the binarised propensity score, where $U = Y - X^{\prime}\theta^{*}$ with $\theta^{*} = \text{plim} \; \hat{\theta}$, and given the values of $\theta$ and $\delta$. The elements of $\theta$ are drawn from the $U(0.5, 1)$ distribution with random flips of their signs and the elements of $\delta$ are drawn from the $U(-1,1)$ distribution. All elements of $\gamma$ are set to 0. 
    }
  \label{fig:Monte Carlo binarised}
\end{figure}

\begin{figure}
    \centering
    \includegraphics[trim = 50 15 50 25, width = \linewidth, keepaspectratio]{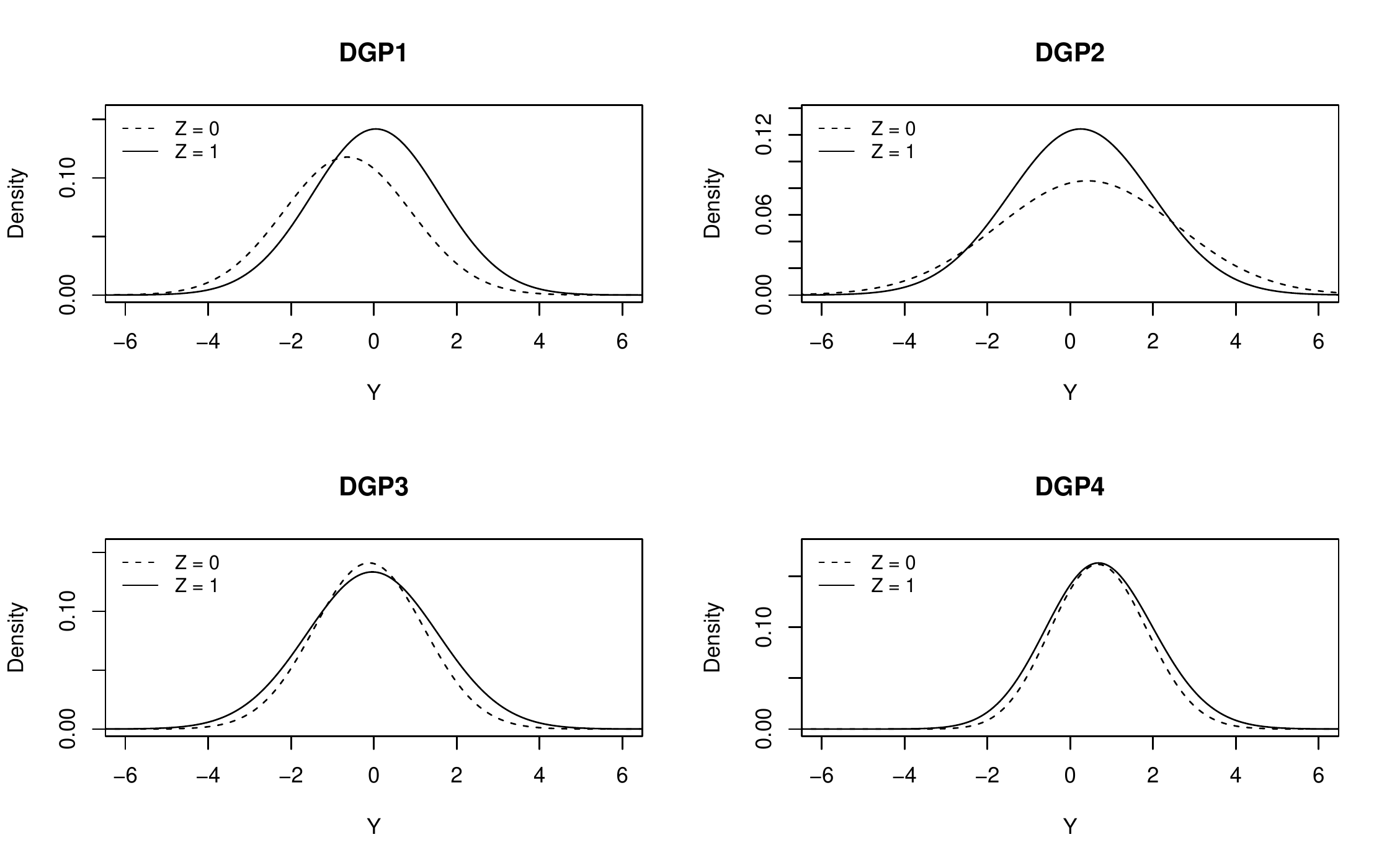}
    \caption{Subdensities of $(Y, D = 1)$ for DGPs 1 - 4 conditional on the instrument $Z$ given $\beta$ and $\delta$.he elements of $\beta$ are drawn from the $U(0.5, 1)$ distribution with random flips of their signs and the elements of $\delta$ are drawn from the $U(-1,1)$ distribution.     
    }
    \label{fig:Monte Carlo Y}
\end{figure}


\begin{figure}
    \centering
    \includegraphics[width = \linewidth, keepaspectratio]{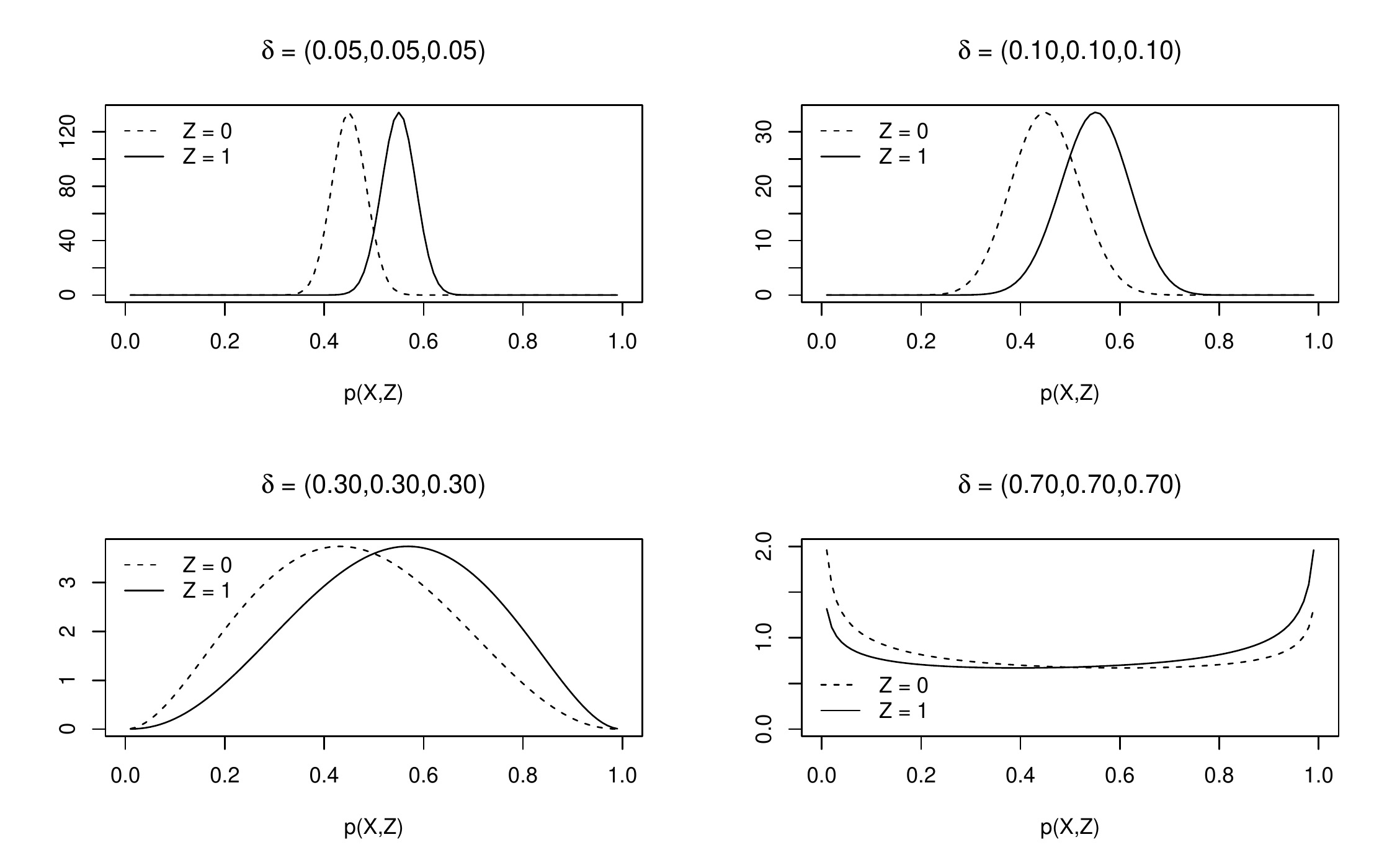}
    \caption{Density of $p(X,Z)$ conditional on $Z = 0$ and $Z = 1$ for four values of $\delta$ with $\gamma$ set to 0.} 
    \label{fig:Monte Carlo pscore dens comp}
\end{figure}

\begin{figure}
    \centering
    \includegraphics[width = \linewidth, keepaspectratio]{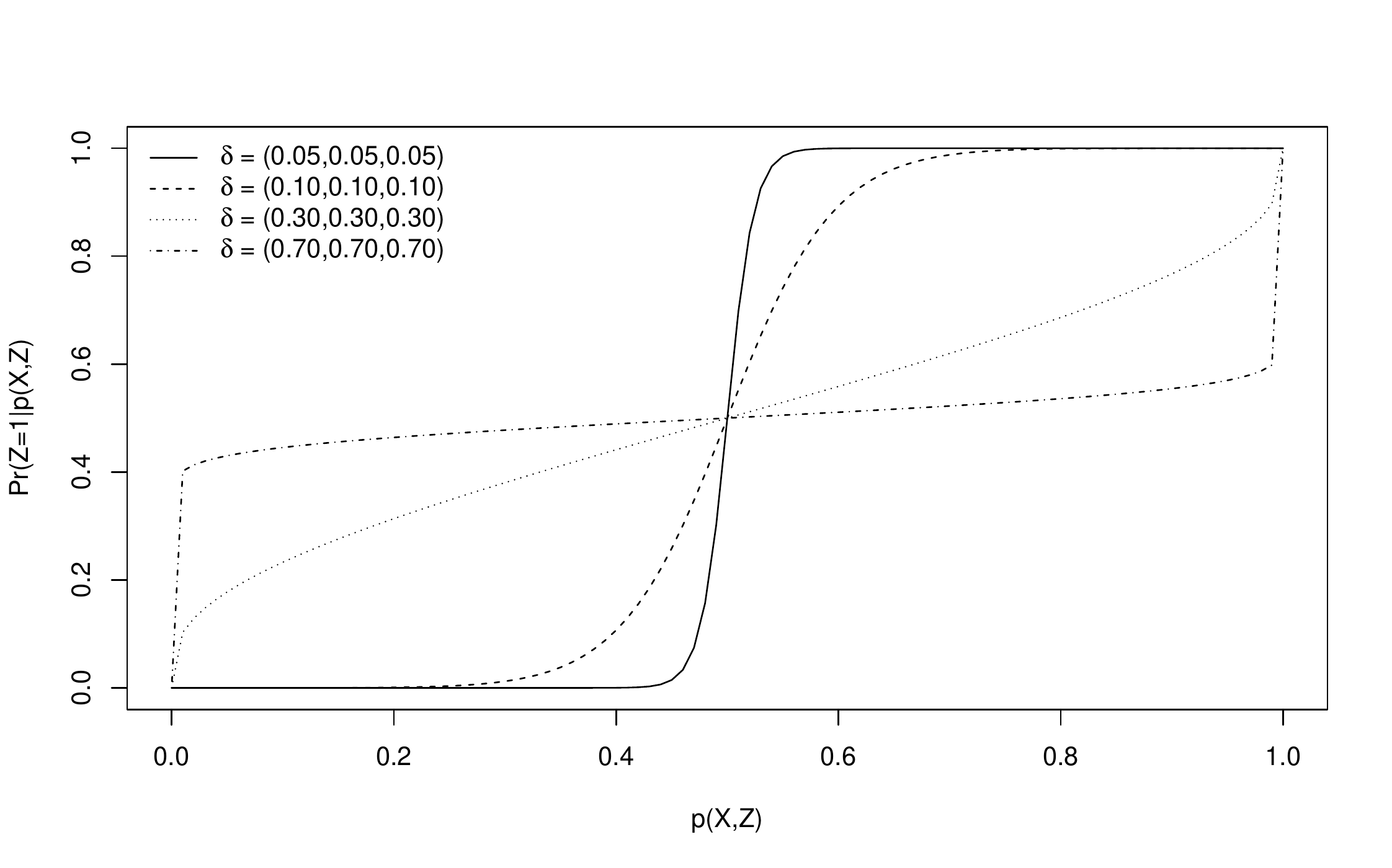}
    \caption{$Pr(Z=1|p(X,Z))$ for four values of $\delta$ with $\gamma$ set to 0.}
    \label{fig:Monte Carlo PZ1 comp}
\end{figure}

\section{Applications}\label{sec:applications}

In this section we apply the proposed test procedure to the instrumental variables of \citet{Card1993}, \citet{angristevans1998}, and \citet{oreopoulos2006}. In all cases we maintain the outcome, treatment, instrument and conditioning covariates of the original design, and estimate a potential outcome model of the form described in Section 2. For the propensity score, we estimate a probit with the instrument, conditioning covariates and interactions between the two as controls. For the outcome, we allow coefficients on the covariates to depend on the treatment. In the case of \citet{Card1993} and \citet{angristevans1998}, the partially linear model is estimated as described in step 2 of Section \ref{sec:test procedure}, with local linear regressions used to estimate $E[X|p]$ and $E[Y|p]$. $\Pr(Z = 1|p)$ is estimated by local constant regressions. Bandwidths are selected by least squares cross validation. For \citet{oreopoulos2006}, the partially linear model is estimated by assuming that the nonparametric component $\phi(p)$ is a cubic polynomial in $p$, and $\Pr(Z = 1|p)$ is estimated by probit. 

Table \ref{tab:pvalues} presents p-values. A p-value below 0.05 corresponds to a rejection of the null hypothesis at the 5\% level of significance. p-values are reported for a range of values of the trimming parameter $\xi$. For each application, Table \ref{tab:sample sizes for each test} shows the number of observations in the data set, and the number of observations included in the trimmed samples used in the tests of nesting inequalities with a distilled sample and index sufficiency.  


\begin{table}[h]
    \centering
    \begin{tabular}{l c c c c}
    \hline \hline
     & $\xi \approx 0.07$ & $\xi \approx 0.21$ & $\xi = 0.30$ & $\xi = 1$ \\
    \hline
    \addlinespace[0.1cm]
    
    \multicolumn{5}{l}{\textbf{\citet{Card1993}}}\\
    Nesting inequalities & 0.360 & 0.996 & 0.998 & 0.998 \\ 
    Index sufficiency & 0.190 & 0.328 & 0.248 & 0.180 \\ 
    Overall & 0.210 & 0.354 & 0.268 & 0.198 \\
    \addlinespace[0.1cm]
    
    \multicolumn{5}{l}{\textbf{\citet{angristevans1998}}} \\
    Nesting inequalities &  0.636 & 0.986 & 0.996 & 0.998 \\
    Index sufficiency &  0.908 & 0.766 & 0.666 & 0.490 \\ 
    Overall &  0.876 & 0.842 & 0.734 & 0.490 \\ 
    \addlinespace[0.1cm]
 
    \multicolumn{5}{l}{\textbf{\citet{oreopoulos2006}}}\\
    Nesting inequalities & 0.000 & 0.000 & 0.000 & 0.006 \\ 
      \addlinespace[0.1cm]
    \addlinespace[0.05cm]
    \hline
    \end{tabular}
    \caption{P-values for the test of nesting inequalities, the test of index sufficiency, and the combined test of both. A P-value below 0.05 indicates that the null hypothesis would be rejected at the 5\% level of significance. $\xi$ is a user specified trimming parameter as described in \citet{Kitagawa2015}. For \citet{Card1993} conditioning covariates are residence in the South, residence in a standard metropolitan statistical area in 1966, residence in a standard metropolitan statistical area in 1976, race, a quadratic in potential experience, living in a single mother household at age 14, living with both parents at age 14, mother's and father's years of education, dummies for mother's and father's years of education being missing, a nine value categorical variable constructed by interacting father's and mother's education class, and a categorical variable for region of residence in 1966. Covariates for \citet{angristevans1998} are mother's current age, mother's age at the time of the birth of their first child, the sex of the first child, and two dummies for race. Conditioning covariates for \citet{oreopoulos2006} are a quartic in age, survey year fixed effects, year aged 14 fixed effects, sex and a dummy for residence in Northern Ireland. }\label{tab:pvalues}  

\bigskip


    \begin{tabular}{l c c c }
    \hline \hline
     &  Observations & Nesting inequalities & Index sufficiency \\
    \hline
    \addlinespace[0.1cm]
    
   \citet{Card1993} & 3010 & 3010 & 3010 \\   
  \citet{angristevans1998} & 12732 & 12731 & 752 \\
   \citet{oreopoulos2006} & 82908 & 82908 & 0 \\
    \addlinespace[0.1cm]
    \hline
    \end{tabular}
    \caption{The number of observations in the original sample, the number of observations used in the test of nesting inequalities, and the number of observations used in the test of index sufficiency. Observations are removed for the test of nesting inequalities in order to attain stochastic dominance in the trimmed sample. Observations are removed for the test of index sufficiency if $\Pr(Z=1|p(X,Z)) \leq 0.05$ or $\Pr(Z=1|p(X,Z)) \geq 0.95$}\label{tab:sample sizes for each test}  
\end{table}
 
\subsection{\citet{Card1993}}

\citet{Card1993} estimates returns to education. The outcome $Y$ is the logarithm of weekly earnings, $D$ is years of education, and $Z$ is growing up in a local labour market with an accredited four year college. The headline specification uses four ordered discrete, three categorical and six binary conditioning covariates. These include potential experience, parental education, and residence in a standard metropolitan statistical area (SMSA). The sample consists of 3,010 observations.

The validity of this instrument has previously been tested by \citet{hubermellace15Testing}, \citet{Kitagawa2015}, \citet{mourifiewan} and \citet{sun2020}. However, for computational reasons, at most only a subset of the conditioning covariates have been used. \citet{sun2020} does not control for covariates,  \citet{hubermellace15Testing} conditions on two covariates, \citet{mourifiewan} three, and \citet{Kitagawa2015} five. In addition, all of the conditioning covariates that have been used are binary. The proposed test procedure allows us to include all conditioning covariates.

As the treatment is not a binary variable, we binarize it. The binarized treatment is defined to be 1 if years of education is at least 16. \citet{angristimbens1995} and \citet{marshall_2016} show that coarsening the treatment variable in this way can lead the instrument to violate the exclusion restriction through coarsening bias.\footnote{Section 5 of \citet{sun2020} discusses coarsening bias in the context of testing for instrument validity.} This occurs when the instrument induces changes in the underlying continuous treatment which are not captured by the discretised treatment. We verify that, conditional on the binarised treatment, the distribution of the continuous treatment is similar for both values of the instrument. Therefore, the binarized treatment does appear to capture changes in the underlying continuous treatment.

Panel \subref{fig:Card outcome} of Figure \ref{fig:Card subdensities} plots joint $(Y,D)$ subdensities for $Z=0$ and $Z=1$. \citet{Kitagawa2015} finds that, without conditioning on covariates, the null of instrument validity can be rejected. This violation can be seen in the $Y(0)$ subdensities. Figure \ref{fig:Card pscore} plots the conditional distributions and smoothed conditional densities for the estimated propensity score. The distribution for $Z=1$ lies below the distribution for $Z=0$, which indicates that first-order stochastic dominance holds, and there is considerable overlap between the densities. As show in Table \ref{tab:sample sizes for each test}, no observations are removed for either the test of nesting inequalities with a distilled sample or the test of index sufficiency. Panel \subref{fig:Card nesting} of Figure \ref{fig:Card subdensities} plots the subdensities of estimated partial residuals used in the test of nesting inequalities. After controlling for covariates, the violation in the $D=0$ subdensities is no longer evident. Panel \subref{fig:Card index} plots the reweighted densities used to test index sufficiency. These are similar to those in panel \subref{fig:Card nesting}. p-values are presented in the top panel of \ref{tab:pvalues}. p-values are lower for the index sufficiency test than the test of nesting inequalities with a distilled sample, but the null is not rejected at any conventional level of significance. 
      

\begin{figure}
    \centering
    \begin{subfigure}{\textwidth}
    \centering
    \includegraphics[trim = 50 240 50 25, width = \linewidth]{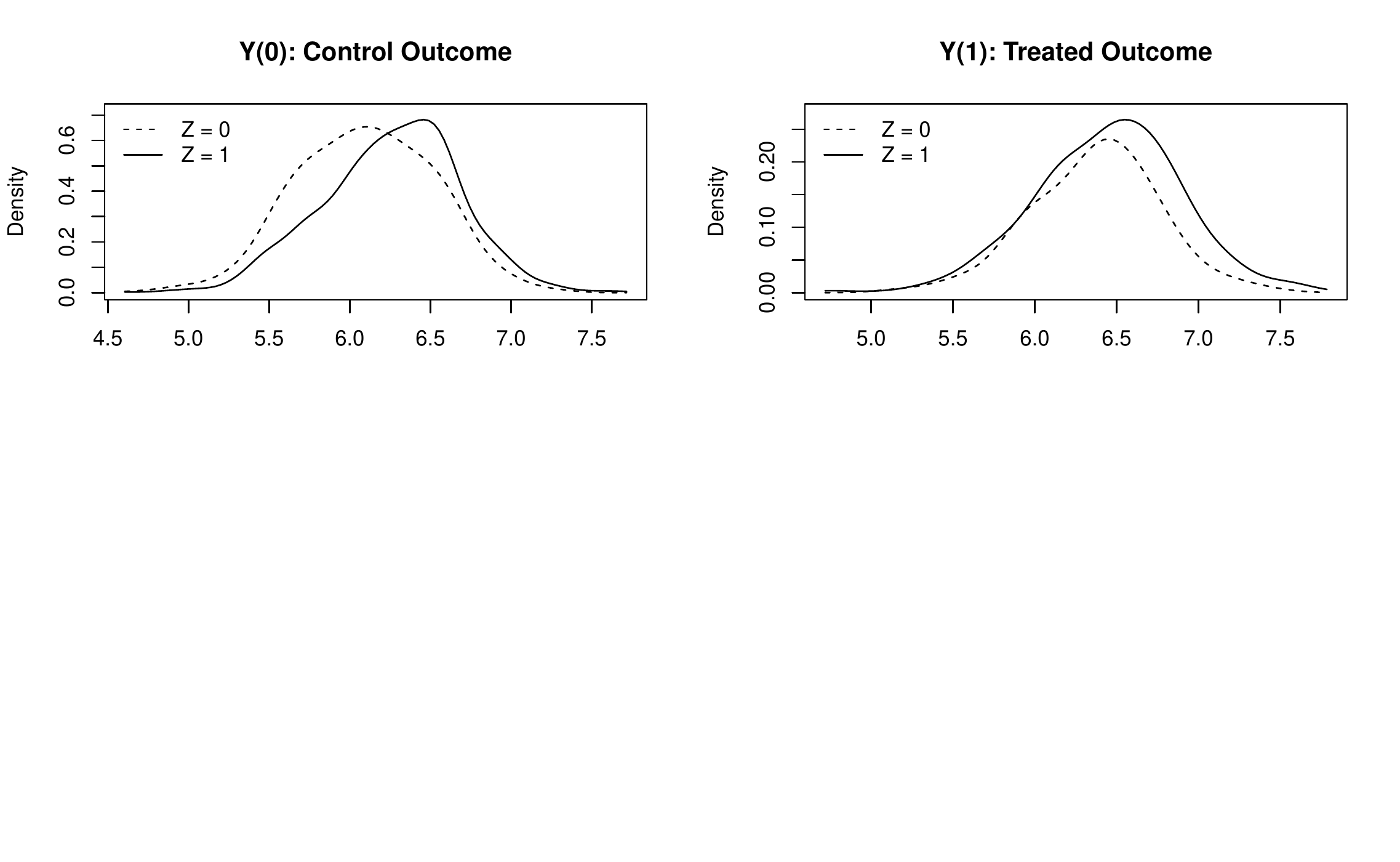}
    \caption{Outcome}\label{fig:Card outcome}
    \end{subfigure}\\
    
    \begin{subfigure}{\textwidth}
    \includegraphics[trim = 50 240 50 0, width = \linewidth, keepaspectratio]{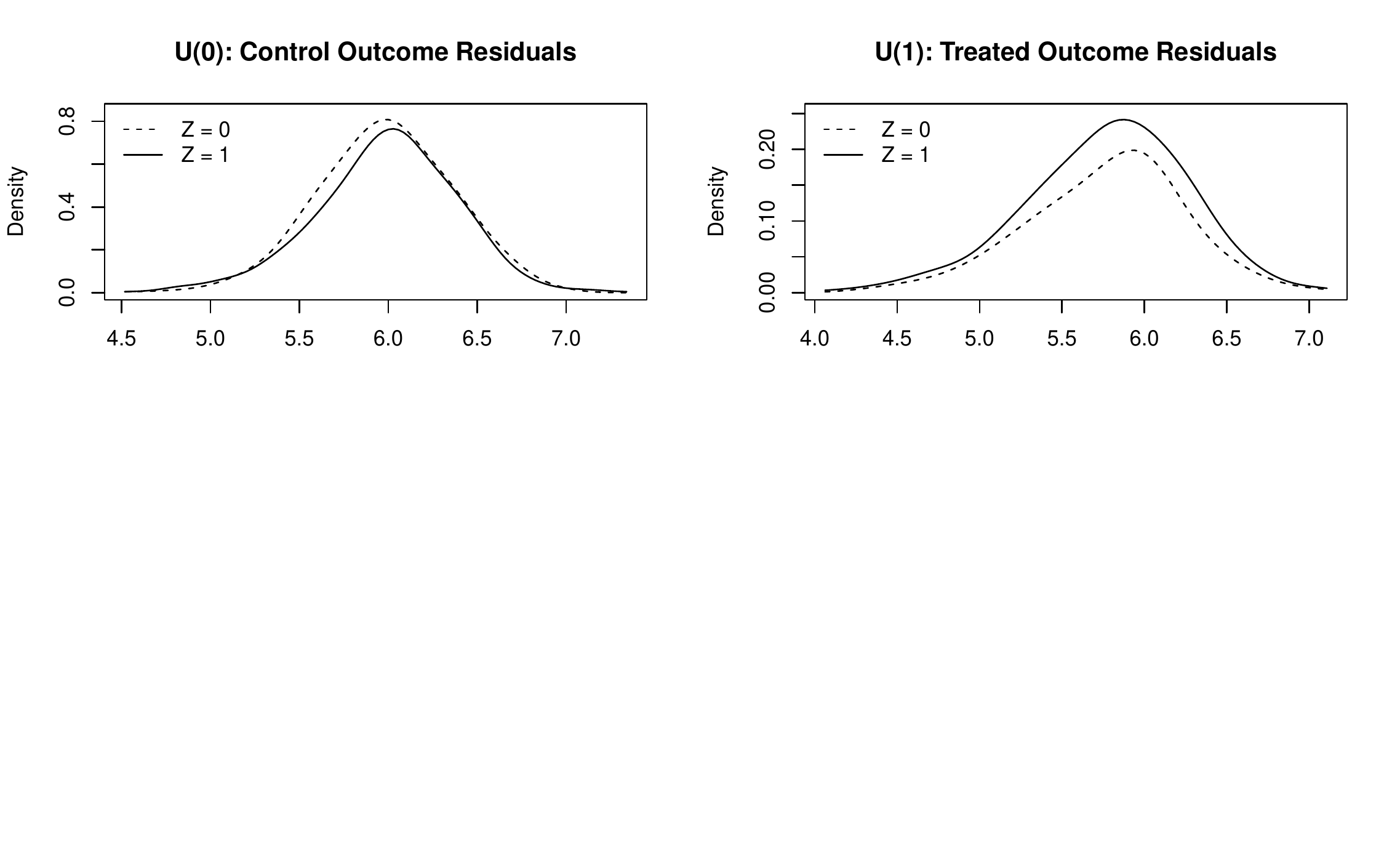}
    \caption{Nesting Inequality Densities} \label{fig:Card nesting}
    \centering
    \end{subfigure}\\
    
    \begin{subfigure}{\textwidth}
    \centering
          \includegraphics[trim = 50 240 50 0, width = \linewidth, keepaspectratio]{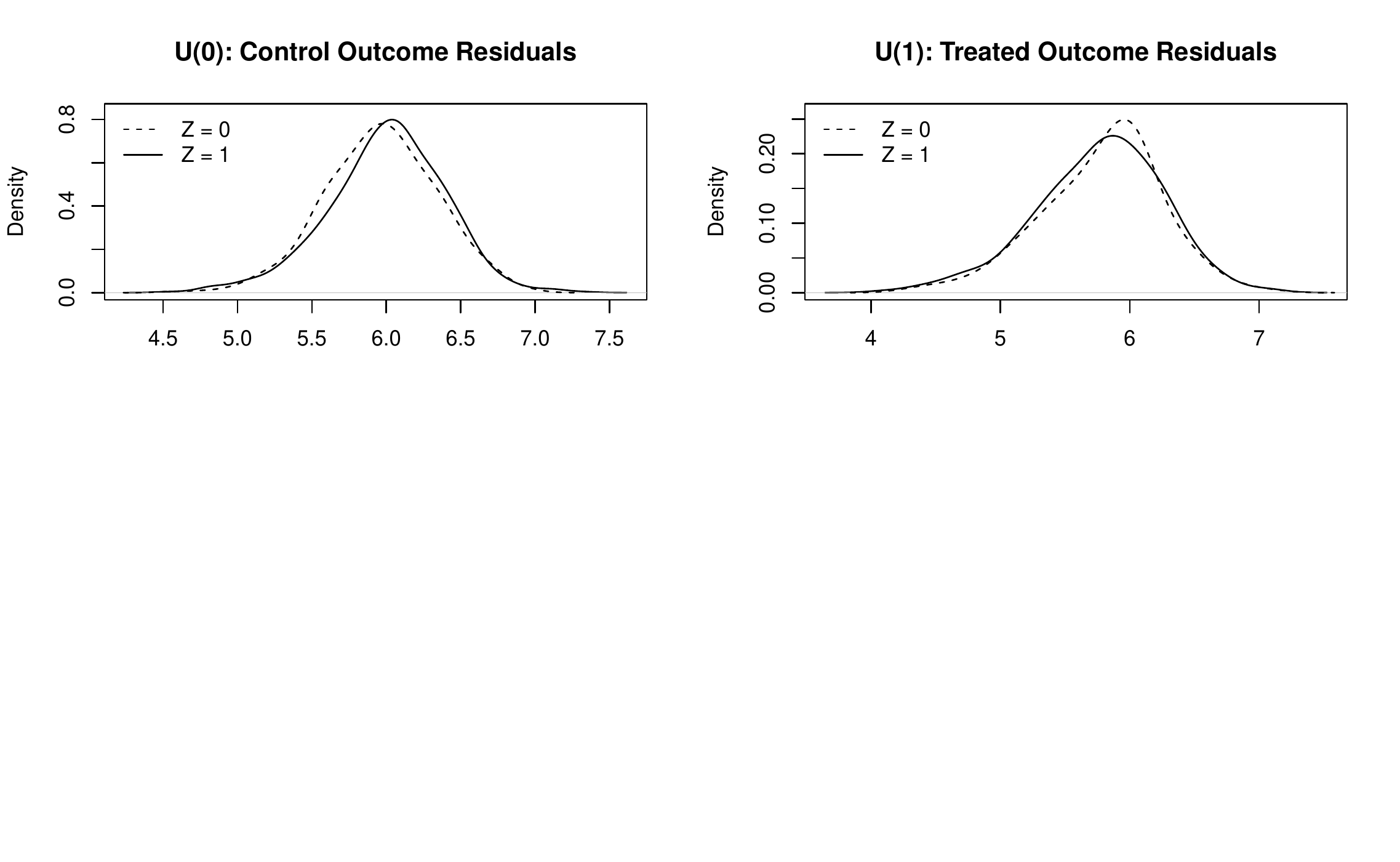}
    \caption{Rewighted} \label{fig:Card index}
    \end{subfigure}\\
    \caption{\citet{Card1993} subdensities. \ref{fig:Card outcome} plots conditional $(Y,D)$ subdensities for $Z=0$ and $Z=1$ where the outcome $Y$ is log weekly earnings. \ref{fig:Card nesting} and \ref{fig:Card index} plot condidtional subdensities for $(U,D)$, where $U$ are the partial residuals obtained by controlling for residence in the South, residence in a standard metropolitan statistical area in 1966, residence in a standard metropolitan statistical area in 1976, race, potential experience and its square, living in a single mother household at age 14, living with both parents at age 14, father's and mother's years of education, indicators for when father's and mother's years of indication are missing, a nine value categorical variable constructed by interacting father's and mother's education class, and a category for region of residence in 1966. \ref{fig:Card nesting} are nesting inequality subdensites which are trimmed to ensure first-order stochastic dominance. \ref{fig:Card index} are index sufficiency subdensities which are reweighted by $P(Z=1|p(X,Z))$. Densities use the Gaussian kernel with the rule-of-thumb bandwidth.} \label{fig:Card subdensities}
\end{figure}

\begin{figure}
    \centering
    \includegraphics[trim = 50 220 50 25, width = \linewidth, keepaspectratio]{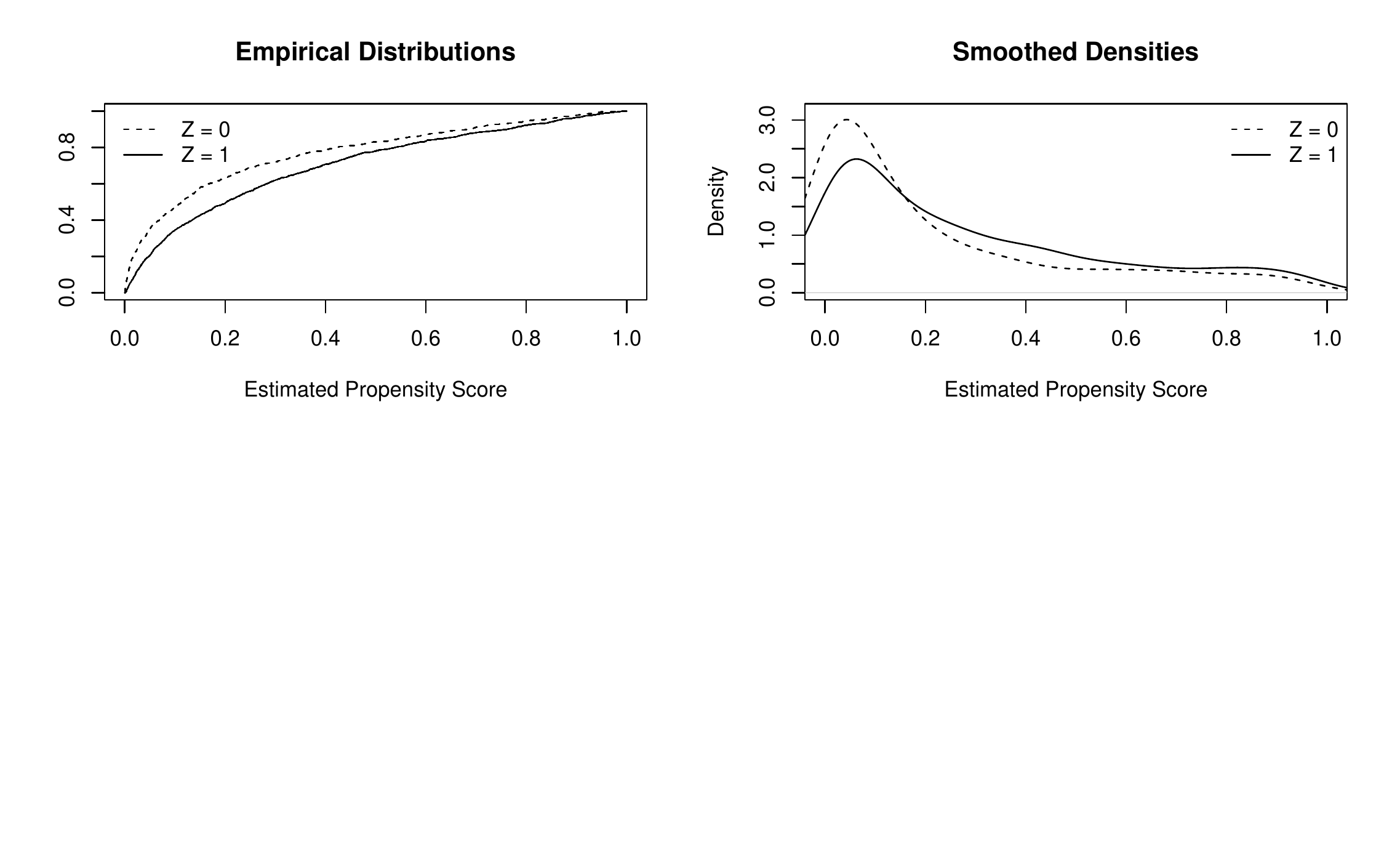}
    \caption{\citet{Card1993} distribution and density of the estimated propensity score. The propensity score is estimated by probit with controls consisting of the instrument, the conditioning covariates, and interactions between the two. Conditioning covariates consist of residence in the South, residence in a standard metropolitan statistical area in 1966, residence in a standard metropolitan statistical area in 1976, race, potential experience and its square, living in a single mother household at age 14, living with both parents at age 14, father's and mother's years of education, indicators for when father's and mother's years of indication are missing, a nine value categorical variable constructed by interacting father's and mother's education class, and a category for region of residence in 1966. The distribution is the simple empirical distribution of estimated propensity scores. The density is an estimated kernel density using the Gaussian kernel and the rule-of-thumb bandwidth.} \label{fig:Card pscore}
\end{figure}

\subsection{\citet{angristevans1998}}

\citet{angristevans1998} estimate the effect of family size on parental labour supply using the same-sex instrument, which takes a value of one when the two eldest children are the same sex and zero otherwise. Their justification for this instrument is that the sex of children is randomly determined and that there is a parental preference for mixed sex children, so parents whose first two children are the same sex are more likely to have a third. The validity of this instrument was subsequently challenged in \citet{RosenzweigWolpin2000}. Using a simple model of labour supply choice, they show that additional restrictions on preferences are needed for this instrument to satisfy the exclusion restriction.

In response to this debate, the validity of the instrument has been tested by  \citet{huber2015} and \citet{mourifiewan}. In both cases, the \citet{angristevans1998} sample is split into subgroups based on values of a subset of the covariates, and then instrument validity is tested for each subgroup. \citet{huber2015} defines 66 subgroups, and finds that the \citet{hubermellace15Testing} test rejects the null in a small number of them, but the \citet{Kitagawa2015} test never rejects the null. \citet{mourifiewan} reports no rejections at the 10\% level of significance in 24 subgroups. Neither \citet{huber2015} or \citet{mourifiewan} uses the full set of covariates from the original specification of \citet{angristevans1998}, andsome of the variables that are used are coarsened.  

We use the replication files for \citet{angristevans1998} to construct the sample and variables. We focus on the 1980 Census sample of married mothers. This consists of married mothers aged between 21 and 35 who have at least two children and were at least 15 years old at the time of their first birth. The sample contains 254,652 observations. To keep estimation of the semiparametric model feasible, we use a randomly chosen 1 / 20 subsample of this data, which results in 12,732 observations. $Z$ is the same-sex instrument. $D$ is an indicator set to 1 for mothers with more than two children. \citet{angristevans1998} presents results for several measures of labour supply. Here we set $Y$ to be the logarithm of annual family labour income. The headline specification controls for demographic characteristics of the mother, the sex of the first child, and the sex of the second child. 

Panel \subref{fig:Angrist Evans outcome} of Figure \ref{fig:Angrist Evans subdensities} plots subdensities for the outcomes without controlling for covariates. In this case there is no visual evidence of violation of the instrument validity. \citet{angristevans1998} shows that the same-sex instrument almost exactly balances covariate values, which suggests that the instrument and covariates are independent. As discussed in Section 4, when the instrument and covariates are independent, failure to control for covariates can mask violations of instrument validity. 

Figure \ref{fig:Angrist Evans pscore} plots the conditional distributions and smoothed densities for the estimated propensity scores. There is no obvious violation of first-order stochastic dominance, and the propensity score densities overlap. One observation is removed when constructing a distilled sample for the test of nesting inequalities.
However, as seen in Table \ref{tab:sample sizes for each test}, only a small fraction of the observations are used in the test of index sufficiency. Since \citet{angristevans1998} uses only a small number of discrete covariates, there are few estimated propensity score values where there is a mixture of observations with $Z=0$ and $Z=1$. Hence, the estimate of $\Pr(Z=1|p(X,Z))$ is either 0 or 1 for the majority of the sample, and these observations are removed for the test of index sufficiency.  
Panels \subref{fig:Angrist Evans nesting} and \subref{fig:Angrist Evans index} of Figure \ref{fig:Angrist Evans subdensities} plot smoothed subdensities for the samples used in the nesting inequalities with a distilled sample and index sufficiency tests. There is no visual evidence of a violation of the nesting inequalities. For index sufficiency, both the $U(0)$ and $U(1)$ subdensities for $Z=0$ and $Z=1$ appear to differ. However, the results in Table \ref{tab:pvalues} show that this difference is not statistically significant. 


\begin{figure}
    \centering
    \begin{subfigure}{\textwidth}
    \centering
    \includegraphics[trim = 50 240 50 0, width = \linewidth]{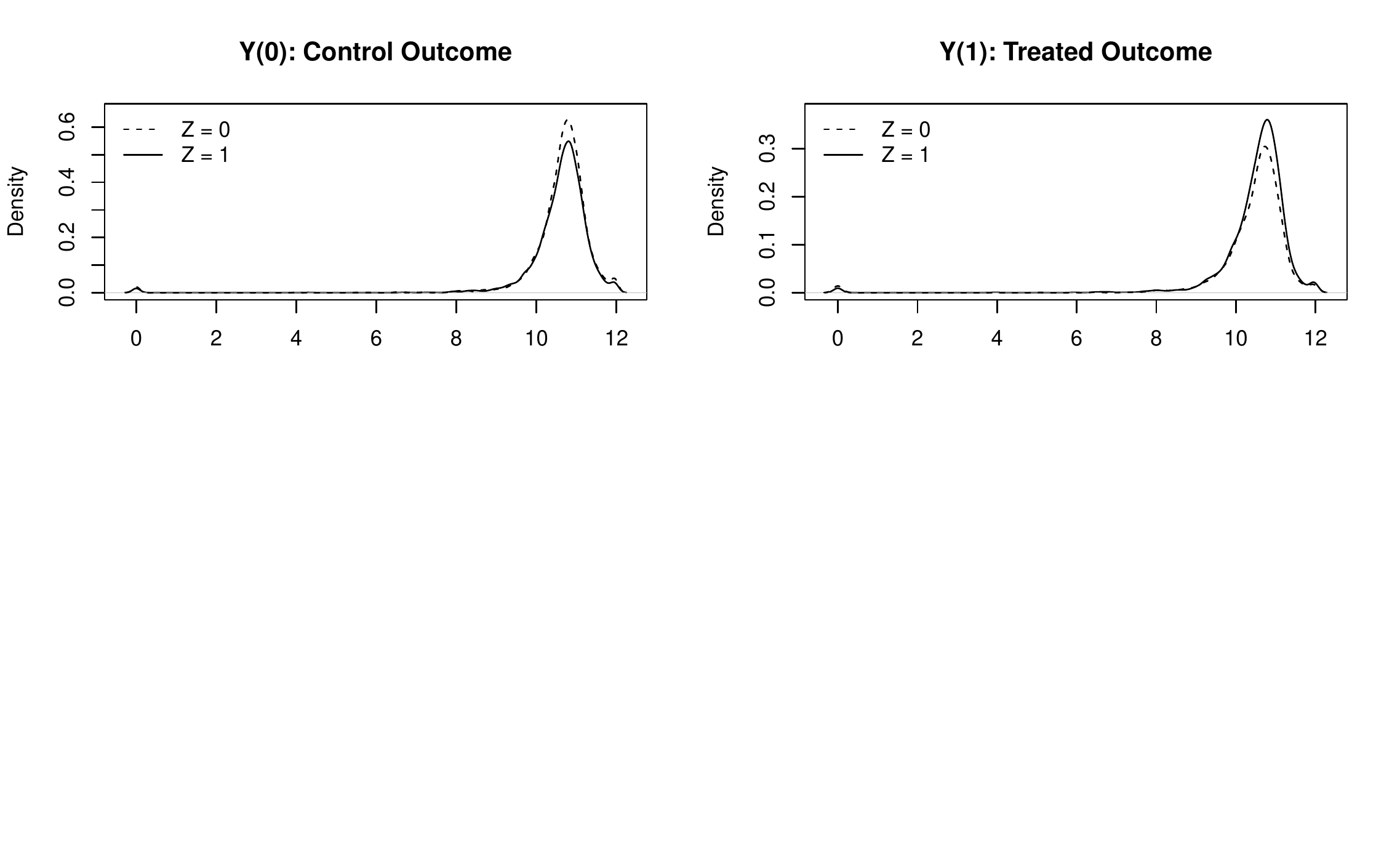}
    \caption{Outcome Densities}\label{fig:Angrist Evans outcome}
    \end{subfigure}\\
    
    \begin{subfigure}{\textwidth}
    \includegraphics[trim = 50 240 50 0, width = \linewidth, keepaspectratio]{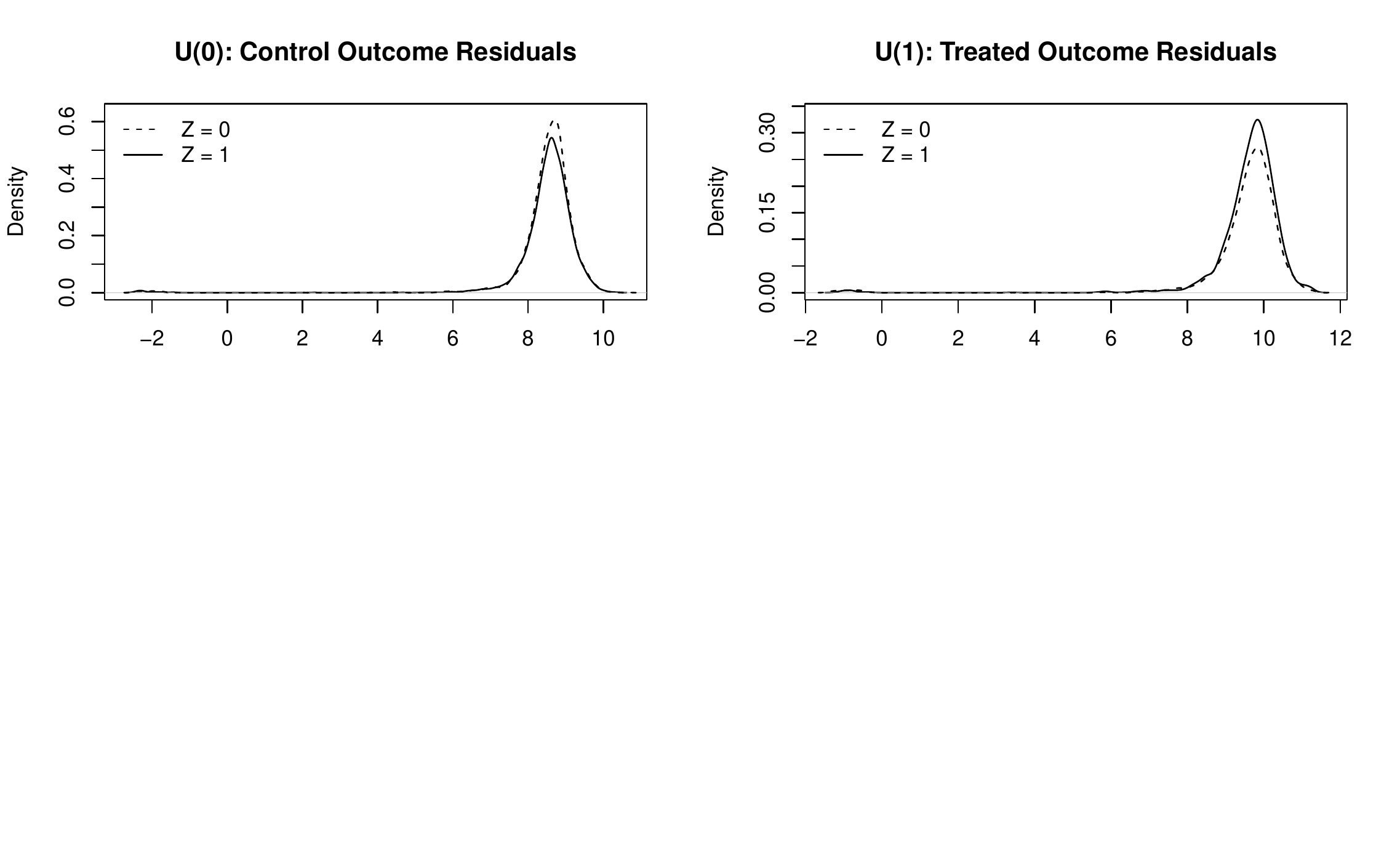}
    \caption{Nesting Inequality Densities}\label{fig:Angrist Evans nesting}
    \centering
    \end{subfigure}\\
    
    \begin{subfigure}{\textwidth}
    \centering
          \includegraphics[trim = 50 240 50 0, width = \linewidth, keepaspectratio]{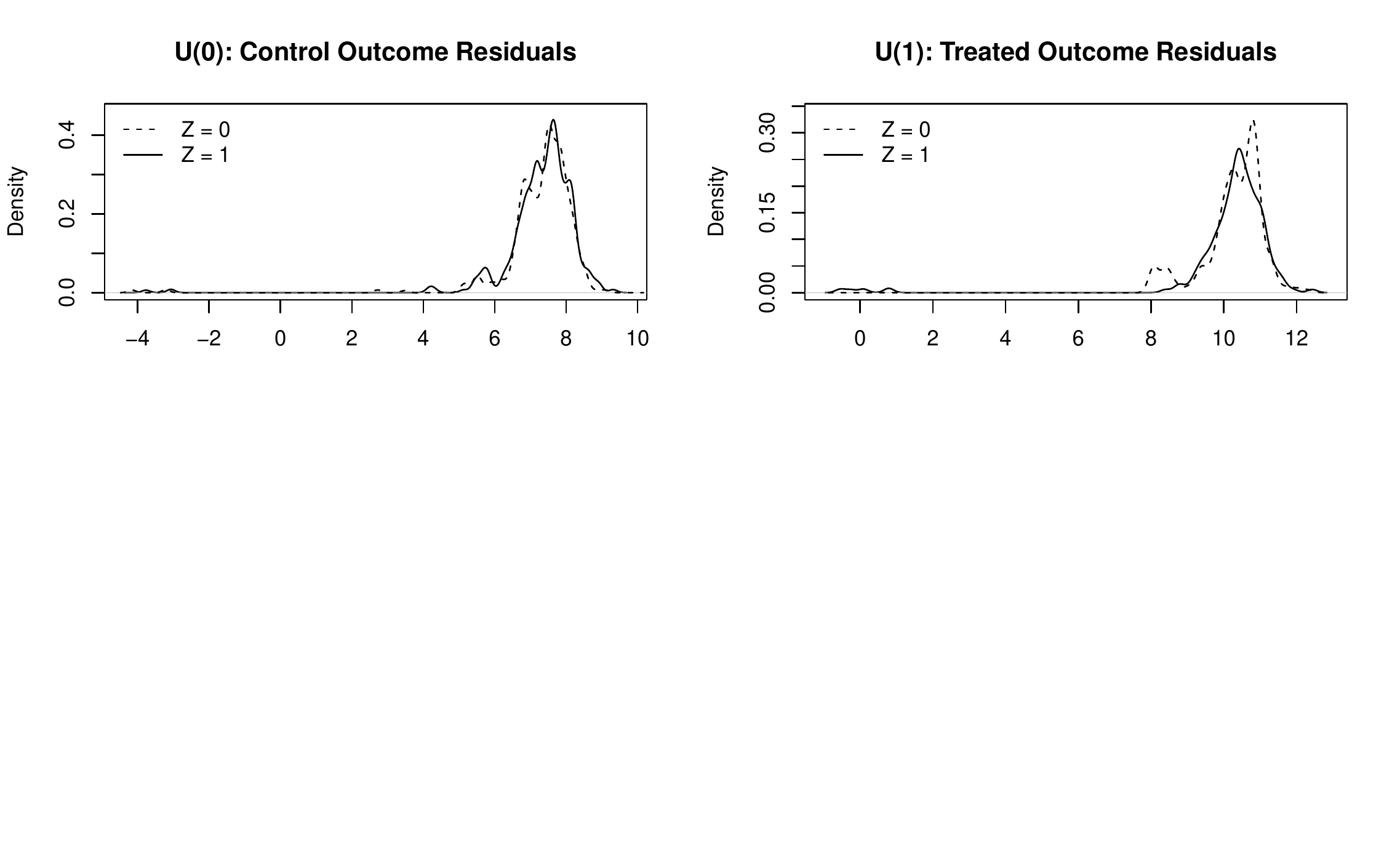}
    \caption{Index Sufficiency Densities}\label{fig:Angrist Evans index}
    \end{subfigure}\\
    \caption{\citet{angristevans1998} subdensities. \ref{fig:Angrist Evans outcome} plots conditional $(Y,D)$ subdensities for $Z=0$ and $Z=1$ where the outcome $Y$ is log annual family earnings. \ref{fig:Angrist Evans nesting} and \ref{fig:Angrist Evans index} plot condidtional subdensities for $(U,D)$, where $U$ are the partial residuals obtained by controlling for mother's age, mother's age at the time of the first birth, two dummies for mother's race, and dummies for the first and second child being a boy. \ref{fig:Angrist Evans nesting} are nesting inequality subdensites which are trimmed to ensure first-order stochastic dominance. \ref{fig:Angrist Evans index} are index sufficiency subdensities which are reweighted by $P(Z=1|p(X,Z))$. Densities use the Gaussian kernel with the rule-of-thumb bandwidth.}\label{fig:Angrist Evans subdensities}
\end{figure}

\begin{figure}
    \centering
    \includegraphics[trim = 50 220 50 25, width = \linewidth, keepaspectratio]{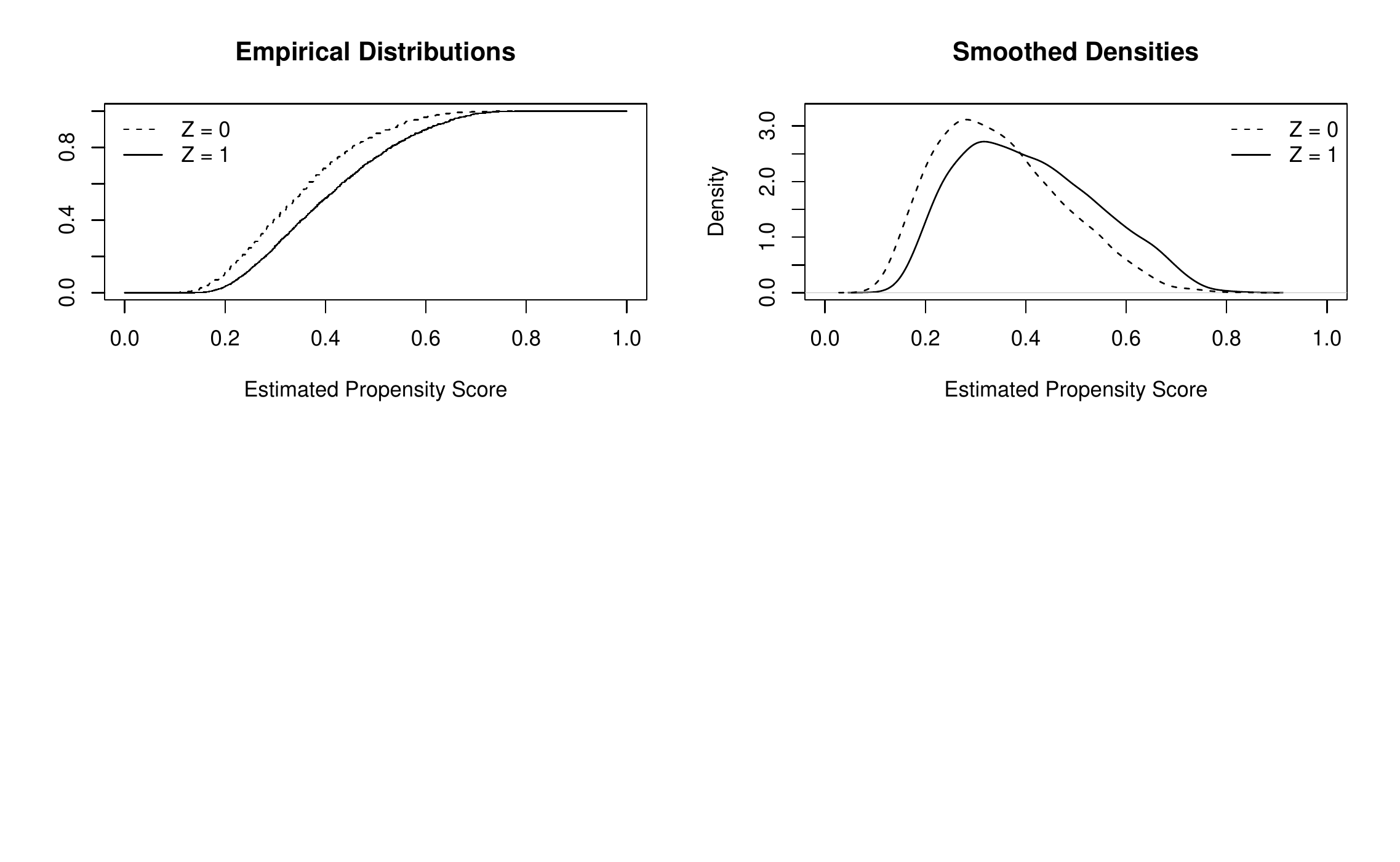}
    \caption{\citet{angristevans1998} distribution and density of the estimated propensity score. The propensity score is estimated by probit with controls consisting of the instrument, the conditioning covariates, and interactions between the two. Conditioning covariates consist of mother's current age, mother's age at the time of the birth of their first child, the sex of the first child, and two dummies for race. The distribution is the simple empirical distribution of estimated propensity scores. The density is an estimated kernel density using the Gaussian kernel and the rule-of-thumb bandwidth.}\label{fig:Angrist Evans pscore}
\end{figure}

\subsection{\citet{oreopoulos2006}}

Changes in schooling laws have been used as an instrument to estimate the effects of education for a variety of outcomes. \citet{acemogluangrist2000} uses this instrument to estimate human capital externalities, \citet{lochnermoretti2004} estimates effects on crime, \citet{LlerasMuney2005} effects on health, and \citet{oreopoulos2006} the returns to education. These designs exploit variation in the timing of changes in compulsory schooling laws across locations. Conditional on covariates, such as location and time fixed effects, this variation is assumed to be unrelated to other changes. \citet{stephensyang2014} presents evidence that this assumption does not hold in data for US states and, in particular, that changes in compulsory schooling laws are correlated with changes in school quality. \citet{brunello2013testing} develops a test of uncorrelatedness between changes in compulsory schooling laws and school quality. They apply this test to European data and find that they cannot reject the null of uncorrelatedness. \citet{BolzernHuber2017} tests the validity of the compulsory schooling law instrument using the \citet{hubermellace15Testing} and \citet{Kitagawa2015} tests with data from seven European countries, and does not reject the null of instrument validity. 


We apply our test to the data of \citet{oreopoulos2006}. The instrument is the legal minimum school leaving age for the UK. This was increased from 14 to 15 in 1947 in Great Britain and in 1957 in Northern Ireland. The dataset consists of individuals aged between 16 and 65. It is constructed by combining UK and Northern Ireland household surveys for the years 1984 through to 2006.\footnote{The published paper \citet{oreopoulos2006} used data from the UK General Household Survey for 1983-1998, and data from the Northern Ireland Continuous Household Survey for 1985-1998. A corrigendum with subsequent editions of each survey incorporated into the data was later published online. We use the revised data.} $Y$ is log real earnings, $D$ is age of leaving full-time education, and $Z$ is a dummy for the legal minimum leaving age faced at age 14 being 15. Conditioning covariates are year aged 14, current age, sex, survey year and residence in Northern Ireland. As with \citet{Card1993}, the treatment variable must be binarized. Here we define the treatment to be one if the age of leaving education is at least 15. We investigate the possibility that this induces coarsening bias, but the binarised treatment seems to capture well the variation in the underlying treatment induced by the instrument. 

The data consists of 82,908 observations which are aggregated into 30,587 cells. Each cell consists of individuals with common values of the treatment, instrument and conditioning covariates. The outcome variable is averaged at the cell level, and each cell is assigned a weight equal to the number of individuals it contains. The identifying assumptions stated in Section 2.2 were not for data that has been aggregated in this way. However, the aggregation procedure preserves all variables apart from the outcome, which implies that the propensity score is also preserved. The identifying assumptions may then be simply restated in terms of the averaged outcome, and the same testable implications derived.\footnote{Note, however, that failure to reject the LATE assumptions for the aggregated data does not imply that they would not be rejected for the unaggregated data and vice versa.}

Panel \subref{fig:Oreopoulos outcome} of Figure \ref{fig:Oreopoulos subdensities} plots subdensities for the outcome variable. It is immediately noticeable that there are very few observations with $Z=1$ and $D=0$. Figure \ref{fig:Oreopoulos pscore} plots the conditional distributions and densities for the propensity score. First-order stochastic dominance clearly holds, but the distributions barely overlap\footnote{The minimum estimated propensity score for a cell with Z = 1 is 0.678, and the estimated maximum propensity score with Z = 0 is 0.683. The overlap contains 8 cells (21 observations) with Z = 0 and 12 cells (94 observations) with Z = 1.}. The estimated $\Pr(Z=1|p(X,Z))$ is either 1 or 0 for all cells. The entire sample is used in the test of nesting inequalities with a distilled sample, but there are no observations for which to perform the test of index sufficiency. Table \ref{tab:pvalues} therefore present p-values for the test of nesting inequalities with a distilled sample only. We find that the null is rejected at even the 1\% level of significance. The violation is in the lower tail of the $U(1)$ subdensities, which are plotted in panel \subref{fig:Oreopoulos nesting} of Figure \ref{fig:Oreopoulos subdensities}. The subdensity of $U(1)$ for $Z=0$ does not lie entirely within the support of the subdensity for $Z=1$. 





\begin{figure}
    \centering
    \begin{subfigure}{\textwidth}
    \centering
    \includegraphics[trim = 50 240 50 0, width = \linewidth]{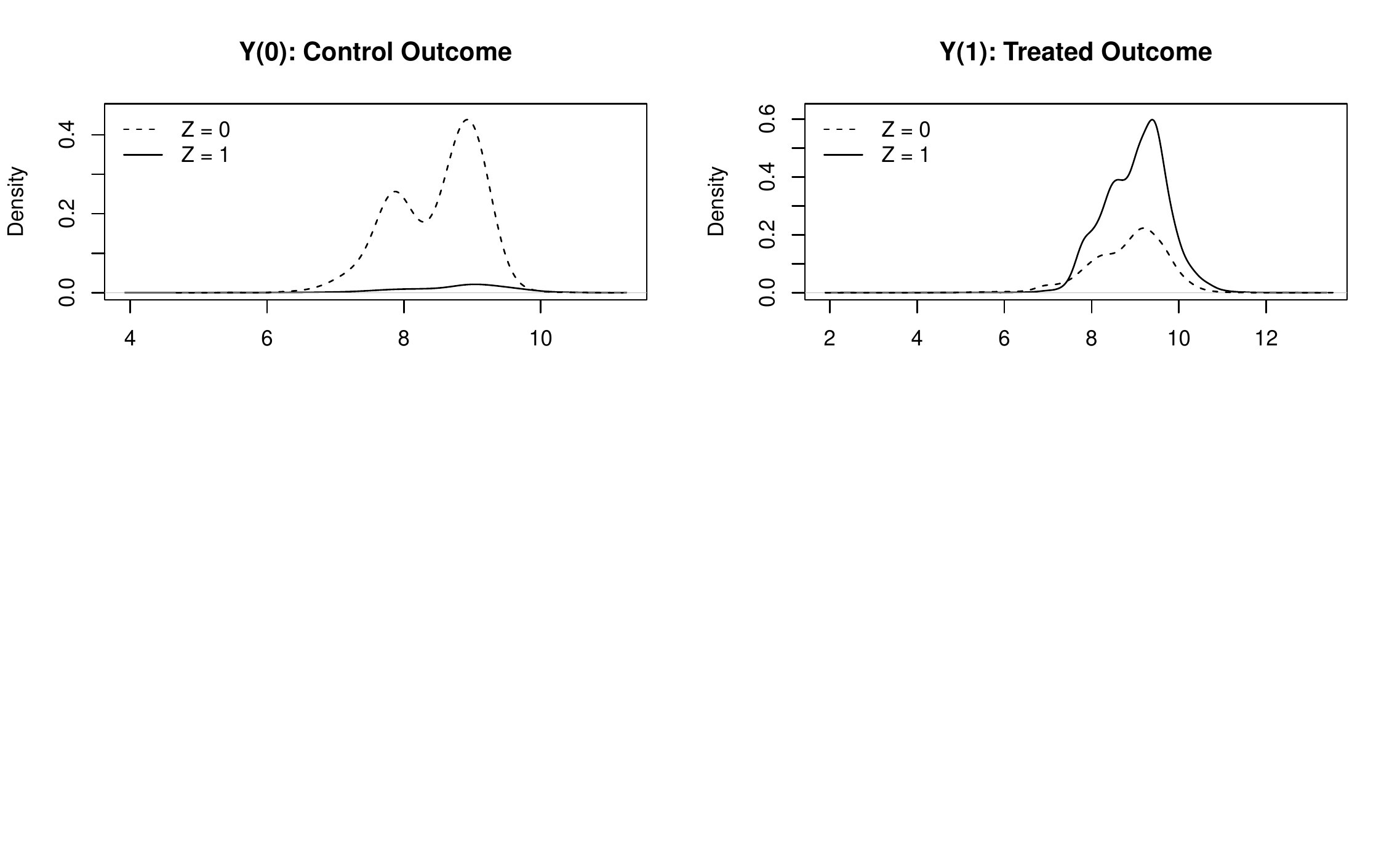}
    \caption{Outcome Densities}\label{fig:Oreopoulos outcome}
    \end{subfigure}\\
    
    \begin{subfigure}{\textwidth}
    \includegraphics[trim = 50 240 50 0, width = \linewidth, keepaspectratio]{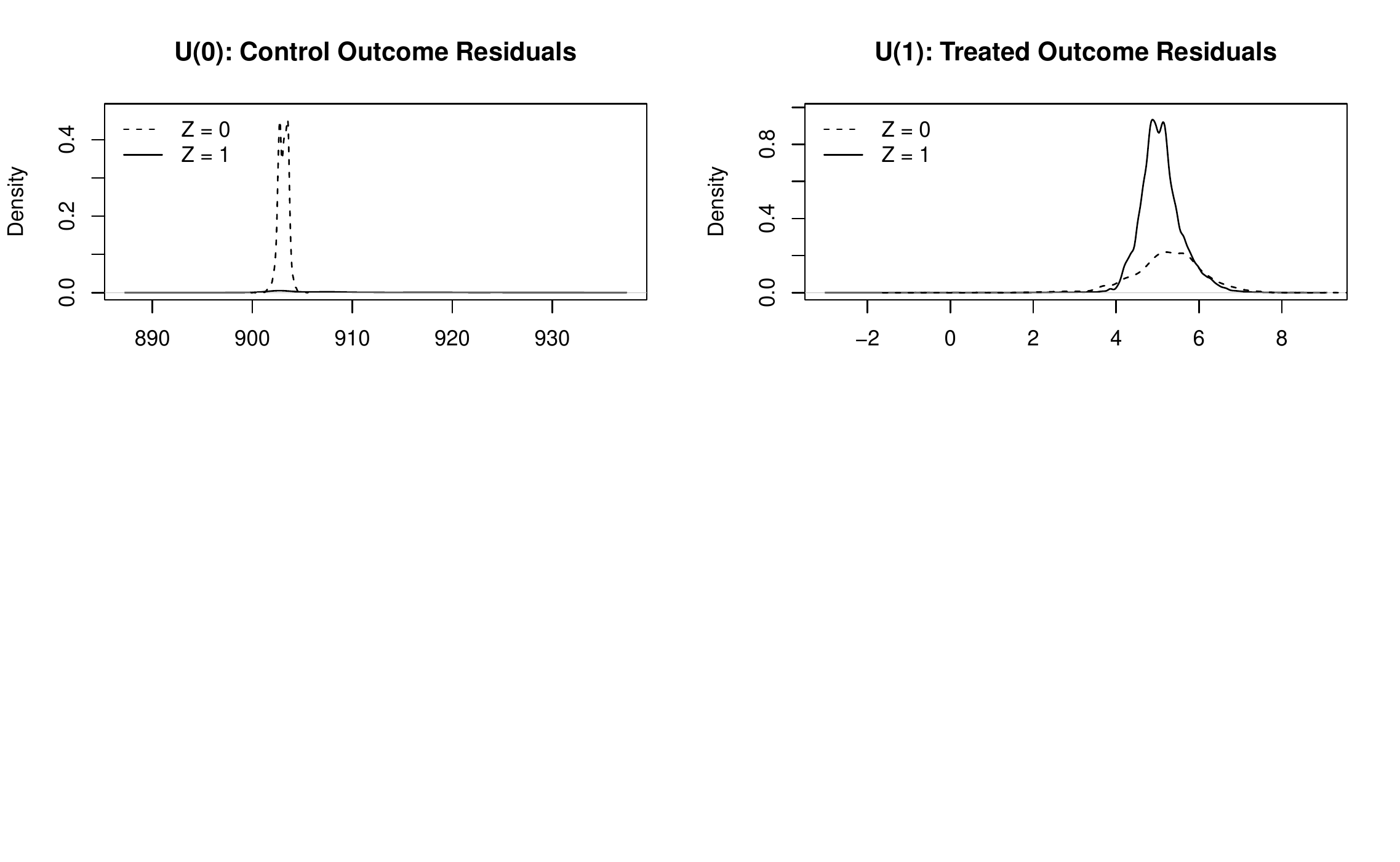}
    \caption{Nesting Inequality Densities}\label{fig:Oreopoulos nesting}
    \centering
    \end{subfigure}\\
    \caption{\citet{oreopoulos2006} subdensities. \ref{fig:Oreopoulos outcome} plots conditional $(Y,D)$ subdensities for $Z=0$ and $Z=1$ where the outcome $Y$ is log real earnings. \ref{fig:Oreopoulos nesting} plots condidtional subdensities for $(U,D)$, where $U$ are the partial residuals obtained by controlling for year aged 14 fixed effects, a quartic polynomial for current age, sex, survey year fixed effects and residence in Northern Ireland. \ref{fig:Oreopoulos nesting} are nesting inequality subdensites which are trimmed to ensure first-order stochastic dominance.  Densities use the Gaussian kernel with the rule-of-thumb bandwidth.}\label{fig:Oreopoulos subdensities}
\end{figure}

\begin{figure}
    \centering
    \includegraphics[trim = 50 220 50 25, width = \linewidth, keepaspectratio]{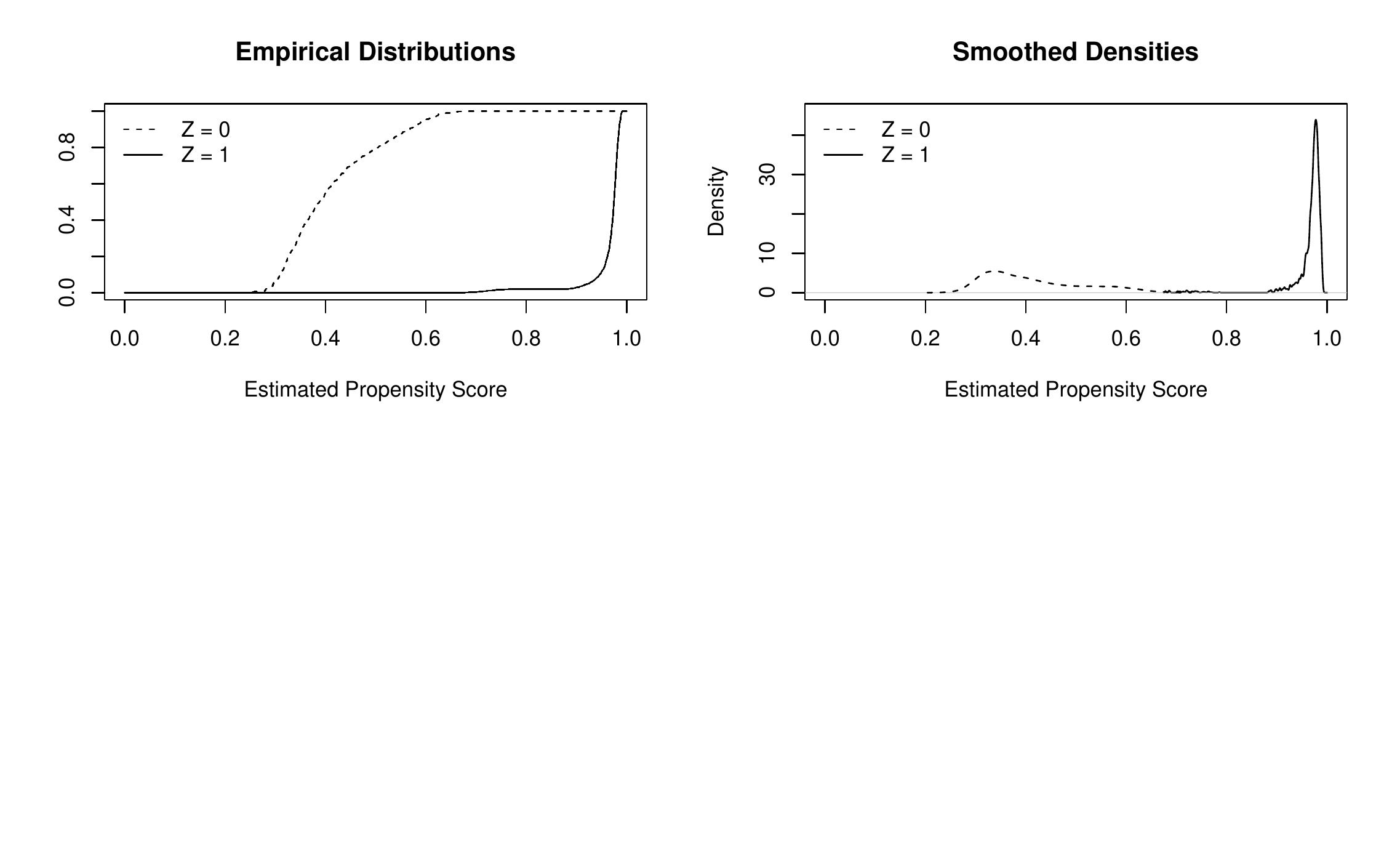}
    \caption{\citet{oreopoulos2006} distribution and density of the estimated propensity score. The propensity score is estimated by probit with controls consisting of the instrument, the conditioning covariates, and interactions between the two. Conditioning covariates consist of year aged 14 fixed effects, a quartic polynomial for current age, sex, survey year fixed effects and residence in Northern Ireland6. The distribution is the simple empirical distribution of estimated propensity scores. The density is an estimated kernel density using the Gaussian kernel and the rule-of-thumb bandwidth.}\label{fig:Oreopoulos pscore}
\end{figure}

\section{Conclusion}

This paper develops a novel test for instrument validity in the LATE/MTE framework that can accommodate conditioning covariates. We follow \citet{CHV11} in assuming a linear relationship between potential outcomes and conditioning covariates and that unobserved heterogeneities are independent of the instrument and covariates.. Expected outcomes then follow a partially linear model, depending linearly on conditioning covariates and non-parametrically on the propensity score. With this functional form, partial residuals can be obtained by partialing out the conditioning covariates. We derive testable implication for the subdensities of these partial residuals: index sufficiency and the nesting inequalities. We propose a test procedure consisting of an initial step where the partially linear model is estimated and partial residuals obtained, followed by a joint test of index sufficiency and the nesting inequalities. Crucially, as conditioning covariates are partialed out, they increase only the complexity of estimating the propensity score and the partially linear model. In this way, the test procedure overcomes the computational hurdles of implementing tests such as \citet{Kitagawa2015}, \citet{hubermellace15Testing} and \citet{mourifiewan} in the presence of conditioning covariates. 

The testable implications are complementary. Index sufficiency holds the propensity score fixed and compares subdensities for different values of the instrument, whereas the nesting inequalities are distributional monotonicity relations ordered with respect to the propensity score. A straightforward approach to testing the nesting inequalities is to coarsen the propensity score and check the inequalities for subdensities indexed by this coarsened variable. However, in cases where the propensity score is strongly correlated with the conditioning covaraites, the distribution of instrument values can become homogenised across propensity score bins, leading to a lack of power. We showed that the nesting inequalities hold with the original instrument in place of the propensity score if the conditional distribution of the propensity score given the instrument is stochastically monotonic in the sense of first-order stochastic dominance. Since first-order stochastic dominance is not implied by the identifying assumptions, we proposed testing the nesting inequalities using a trimmed sample where first-order stochastic dominance does hold. We refer to this process as distillation, and present an algorithm for constructing distilled samples. We show analytically a class of data generating processeses where distillation is guaranteed to outperform a test of the nesting inequalities over propensity score bins in terms of detecting violations of instrument validity. Monte Carlo exercises confirm this finding, with consistent gains in power when using the proposed test procedure in place of testing the nesting inequalities with a coarsened propensity score.

We believe that there are two areas which future research could address. First, formal characterisation of the size and power of this test, incorporating uncertainty over estimation of the partial residuals and the distillation process, is left for future research. Second, as noted in Section 2, in general there will be many possible distilled samples that satisfy the first-order stochastic dominance condition. In Proposition \ref{prop:distillation better}, we showed a restriction on trimming that is guaranteed to improve detection of violations of instrument validity for a class of data generating processes, and we provided an algorithm that satisfies this restriction. However, other possibilities exist, and the optimal choice of distillation process is an open question. 

\begin{appendix}

\section*{Appendix}

\section{Proofs for Propositions 1 and 2}\label{sec:proposition 1 and 2 proofs}
\begin{proof}[Proof of Proposition 1]
Fix the value of the conditioning variable $X$ at $x$, and consider an
arbitrary Borel subset $A\subset \mathcal{Y}$.\ \ Then, the additive error
representation of the selection equation in (\ref{Model1}) implies that the event $\{Y\in A,D=1|X=x, Z=z \}$ is equivalent to $%
\{Y_{1}\in A, V\leq p|p\left( X,Z \right) =p,X=x, Z=z \}$. We therefore 
obtain 
\begin{eqnarray*}
\Pr \left( Y\in A,D=1|X=x, Z=z \right) &=&\Pr \left(
Y_{1}\in A, V\leq p|p\left( X, Z\right) =p,X=x, Z=z \right) \\
&=& \Pr \left( Y_{1}\in A,V\leq p| X=x \right),
\end{eqnarray*}%
where the second equality follows from the random assignment condition (A2). Similarly, for $D=0$ case, we have 
\begin{equation*}
    \Pr \left( Y\in A,D=0|X=x, Z=z \right) = \Pr \left( Y_{0}\in A,V> p| X=x \right).
\end{equation*}
These prove index sufficiency, i.e., conditional on $X$, the distribution of $(Y,D)$ depends on $Z$ through the propensity score $p(X,Z)$ only.  

To obtain the nesting inequalities, we note
\begin{eqnarray}
&&\Pr (Y\in A,D=1|p\left( Z,X\right) =p,X=x)-\Pr (Y\in A,D=1|p\left(
Z,X\right) =p^{\prime },X=x)  \notag \\
&=&\Pr (Y_{1}\in A,V\leq p|X=x)-\Pr (Y_{1}\in A,V\leq p^{\prime
}|X=x)  \notag \\
&=&\Pr (Y_{1}\in A,p^{\prime }<V\leq p|X=x)\geq 0,  \label{Inequality1}
\end{eqnarray}%
follows. The inequality for $D=0$ case is proven in a similar manner.
\end{proof}

\bigskip

\begin{proof}[Proof of Proposition 2]
Noting that the event $\{ U \in A,D=1 | X=x, Z=z \}$ is equivalent to $\{ U \in A, V \leq p | p(X,Z) = p, X=x, Z=z \}$, we have
\begin{eqnarray*}
\Pr \left( U\in A,D=1|X=x, Z=z \right) &=&\Pr \left(
U_{1}\in A, V\leq p|p\left( X, Z\right) =p,X=x, Z=z \right) \\
&=& \Pr \left( U_{1}\in A,V\leq p \right),
\end{eqnarray*}%
where the second equality follows from the strong exogeneity condition of (A4). A similar argument gives 
\begin{equation*}
    \Pr \left( U\in A,D=0|X=x, Z=z \right) = \Pr \left( U_{0}\in A,V> p \right)
\end{equation*}
Hence, we obtain index sufficiency: the distribution of $(U,D)$ depends on $(X,Z)$ through the propensity score $p(X,Z)$ only.

To obtain the nesting inequalities, by (\ref{eq.U1}) and recalling that $V$ is uniformly distributed, we have
\begin{align}
&\Pr(U \in A,D=1|p\left( X,Z \right) =p) - \Pr(U \in A,D=1|p\left(X,Z \right) =p^{\prime }) \notag \\
=& \Pr(U_1 \in A, U_D \leq p) - \Pr(U_1 \in A, U_D \leq p') \notag \\
= & \Pr(U_1 \in A, p'<U_D \leq p) \notag \\ 
\geq & 0.
\end{align}
The other inequality of the proposition can be shown similarly. 
\end{proof}

\section{Proof of Proposition \ref{prop:distillation better}\label{sec:distillation power improvement}}

This section presents a proof for Proposition \ref{prop:distillation better}. We first introduce notation, then show several lemmas, and then use these lemmas to prove each of the three statements of Proposition \ref{prop:distillation better}. We prove the statements of Proposition \ref{prop:distillation better}. for $D = 1$. The proof for $D = 0$ is identical. 

\subsection{Notation}

Let $P$ be an inerval of propensity scores. Define
\begin{align*}
     &\mathcal{F}(P, z) \equiv \int_P \Pr(U \in A, D = 1 | p, Z = z) f(p | p \in P, Z=z) dp, \\
    & w(z|P) \equiv \Pr(Z = z | p \in P). 
\end{align*}
The joint probability of $(U \in A, D = 1)$ conditional on $p \in P$ may be written as
\begin{equation}
   \Pr(U \in A, D = 1|p \in P) = (1 - w(1 |P)) \mathcal{F}(P, 0) + w(1 |P)  \mathcal{F}(P, 1) .
\end{equation}
Thus the nesting inequality with a coarsened propensity score, \eqref{eq:nesting_comparison_1}, subtracts a weighted sum of $\mathcal{F}(P^{+}, 0)$ and $\mathcal{F}(P^{+}, 1)$ from a weighted sum of $\mathcal{F}(P^{-}, 0)$ and $\mathcal{F}(P^{-}, 1)$
\begin{equation}
   (1 - w(1 | P^{-})) \mathcal{F}(P^{-}, 0) + w(1 | P^{-})  \mathcal{F}(P^{-}, 1)  - (1 - w(1 |P^{+})) \mathcal{F}(P^{+}, 0) - w(1 |P^{+})  \mathcal{F}(P^{+}, 1). \label{eq:rewritten nesting}
\end{equation}
We refer to \eqref{eq:rewritten nesting} as the coarsened propensity score nesting violation.

Define
\begin{align*}
     &\mathcal{G}(P, z) \equiv \int_P \Pr(U \in A, D = 1 | p, S_1 = 1, Z = z) f(p | p \in P, S_1 = 1, Z=z) dp, \\
    & v(P|z) \equiv \Pr(p \in  P | Z = z, S_1 = 1, p \in  P^{-} \cup P^{+}). 
\end{align*}
The joint probability of $(U \in A, D = 1)$ conditional on $Z = z$, $S_1 = 1$ and $p \in P^{-} \cup P^{+}$ can be written 
\begin{equation}
      \Pr(U \in A, D = 1| Z = z, p \in  P^{-} \cup P^{+}) = v(P^{-}|z) \mathcal{G}(P^{-}, z) + v(P^{+}| z) \mathcal{G}(P^{+}, z). 
\end{equation}
Thus the nesting inequality with a distilled sample, \eqref{eq:distilled_comparison_1}, subtracts a weighted sum of $\mathcal{G}(P^{-}, 1)$ and $\mathcal{G}(P^{+}, 1)$ from a weighted sum of $\mathcal{G}(P^{-}, 0)$ and $\mathcal{G}(P^{+}, 0)$,
\begin{equation}\label{eq:rewritten distilled}
    v(P^{-}|0) \mathcal{G}(P^{-}, 0) + v(P^{+}| 0) \mathcal{G}(P^{+}, 0) -  v(P^{-}|1) \mathcal{G}(P^{-}, 1) - v(P^{+}| 1) \mathcal{G}(P^{+}, 1).
\end{equation}
We refer to \eqref{eq:rewritten distilled} as the distilled nesting violation. 
\subsection{Lemmas}

\begin{lemma}\label{lemma:weights inequality}
The labeling condition of Proposition \ref{prop:distillation better} implies
\begin{equation}
    w(1|P^{+}) \geq  w(1|P^{-}). \label{eq:weights inequality 1} \\
\end{equation}
\end{lemma}

\begin{proof}[Proof of Lemma \ref{lemma:weights inequality}]

Write out the labeling condition \eqref{eq:labeling restriction} and rearrange 
\begin{align*}
    \frac{\Pr(Z = 1, p \in P^{+})}{\Pr(Z = 1, p \in P^{+}) + \Pr(Z = 1, p \in P^{-})} &\geq  \frac{\Pr(Z = 0, p \in P^{+})}{\Pr(Z = 0, p \in P^{-}) + \Pr(Z = 0, p \in P^{+})} \\ \implies 
      \Pr(Z = 1, p \in P^{+})\Pr(Z = 0, p \in P^{-}) &\geq  \Pr(Z = 1, p \in P^{-})\Pr(Z = 0, p \in P^{+}) \\ \implies 
      \Pr(Z = 1, p \in P^{+})\Pr(p \in P^{-}) &\geq  \Pr(Z = 1, p \in P^{-})\Pr(p \in P^{+}) \\ \implies 
        \Pr(Z=1|p \in P^{+}) &\geq \Pr(Z=1|p \in P^{-}).
\end{align*}
The third line adds $\Pr(Z=1, p \in P^{+})\Pr(Z=1, p \in P^{-})$ to both sides and simplifies. By definition, the left hand side of the last line $w(1|P^{+})$, and the right hand side is  $w(1|P^{-})$.\end{proof}

\begin{lemma}\label{lemma:v inequality}
The labeling condition and trimming rule jointly imply
\begin{equation}
     v(P^{+}|1) \geq v(P^{+}|0). \label{eq:weights inequality 2}
\end{equation}

\end{lemma}

\begin{proof}[Proof of Lemma \ref{lemma:v inequality}]

Consider first $v(P^{+}|1)$
\begin{align*}
    v(P^{+}|1) &\equiv \frac{\Pr(Z = 1, p \in P^{+}, S_1 = 1)}{\Pr(Z = 1, p \in P^{+}, S_1 = 1) + \Pr(Z = 1, p \in P^{-}, S_1 = 1)} \\
    &= \frac{\Pr(Z = 1, p \in P^{+})}{\Pr(Z = 1, p \in P^{+}) + \Pr(Z = 1, p \in P^{-}, S_1 = 1)} \\
    &\geq \frac{\Pr(Z = 1, p \in P^{+})}{\Pr(Z = 1, p \in P^{+}) + \Pr(Z = 1, p \in P^{-})} 
\end{align*}
The second line uses that the trimming under imposes that all observations with $Z = 1$ and $p \in P^{+}$ are retained. The third line follows from $\Pr(Z = 1, p \in P^{-}) \geq \Pr(Z = 1, p \in P^{-}, S_1 = 1) \geq 0$. The expression on the last line is equivalent to $\Pr(p \in P^{+}| Z = 1)$, hence
\begin{equation}
      v(P^{+}|1) \geq \Pr(p \in P^{+}| Z = 1) \label{eq:v inequality proof 1}
\end{equation}
Following the same steps with  $v(P^{+}|0)$, and using that the trimming rule imposes that all observations $Z = 0$ and $p \in P^{-}$ are retained yields
\begin{equation}
   \Pr(p \in P^{+}| Z = 1) \geq v(P^{+}|0) \label{eq:v inequality proof 2}
\end{equation}
Combine \eqref{eq:v inequality proof 1}, \eqref{eq:v inequality proof 2} and the labeling condition \eqref{eq:labeling restriction}
\begin{equation*}
     v(P^{+}|1) \geq \Pr(p \in P^{+}| Z = 1) \geq \Pr(p \in P^{+}| Z = 1) \geq v(P^{+}|0) 
\end{equation*} \end{proof}

\begin{lemma}\label{lemma:conditional nesting implication}
The conditional nesting assumption implies 
\begin{equation}
    \mathcal{F}(P^{+}, z) \geq \mathcal{F}(P^{-}, z), \label{eq:F inequality} 
\end{equation}
for $z \in \{0, 1\}$.
\end{lemma}

\begin{proof}[Proof of Lemma \ref{lemma:conditional nesting implication}]

Let $\hat{p}$ be the upper bound of $P^{-}$. Equation \eqref{eq:conditional nesting D = 1} implies
\begin{align*}
     \mathcal{F}(P^{+}, z) & \equiv \int_{P^{+}} \Pr(U \in A, D = 1|Z = z, p) f(p|p \in P^{+}, Z = z) dp \\ 
    &\geq \Pr(U \in A, D = 1|Z = z, p = \hat{p}) \\
    &\geq \int_{P^{-}} \Pr(U \in A, D = 1|Z = z, p) f(p|p \in P^{-}, Z = z) dp \\
    &\equiv \mathcal{F}(P^{-}, z).
\end{align*} \end{proof}

\begin{lemma}\label{lemma:trimming inequality}

Under the trimming rule
\begin{align}
    &\mathcal{G}(P^{-}, 0) = \mathcal{F}(P^{-}, 0), \label{eq:FG Z=0 inequality} \\
    &\mathcal{F}(P^{+}, 1) = \mathcal{G}(P^{+}, 1). \label{eq:FG Z=1 inequality} 
\end{align}
\end{lemma}

\begin{proof}[Proof of Lemma \ref{lemma:trimming inequality}]

Equation \eqref{eq:FG Z=0 inequality} is shown as follows
\begin{align*}
    \mathcal{G}(P^{-}, 0) &\equiv \int_{P^{-}} \Pr(U \in A, D = 1 | p, S_1 = 1, Z = 0) f(p | p \in P^{-} S_1 = 1, Z=0) dp, \\
    &= \int_{P^{-}} \Pr(U \in A, D = 1 | p, Z = 0) f(p | p \in P^{-}, Z=0) dp \\
    &\equiv \mathcal{F}(P^{-}, 0).
\end{align*}
The second line uses that, under the trimming rule, $S_1 = 1$ for all observations with $Z = 0$ and $p \in P^{-}$.  By using that $S_1 = 1$ for all observations with $Z = 1$ and $p \in P^{+}$, \eqref{eq:FG Z=1 inequality} can be shown similarly. \end{proof}

\begin{lemma}\label{lemma:G conditional nesting implication}
The conditional nesting assumption and trimming rule jointly imply  
\begin{equation}
    \mathcal{G}(P^{+}, 1) \geq \mathcal{G}(P^{-}, 1), \label{eq:G inequality}
\end{equation}
and either 
\begin{align}
    \mathcal{G}(P^{+}, 0) &\geq \mathcal{G}(P^{-}, 0), \label{eq:G Z = 0 case 1}\\
    v(P^{+}, 0) &> 0, \label{eq:G Z = 0 case 1 v}
\end{align}
or
\begin{align}
    &\mathcal{G}(P^{+}, 0) = 0, \label{eq:G Z = 0 case 2} \\
    &v(P^{+}, 0) = 0. \label{eq:G Z = 0 case 2 v}
\end{align}
\end{lemma}

\begin{proof}[Proof of Lemma \ref{lemma:G conditional nesting implication}]

To show \eqref{eq:G inequality}, use Lemma \ref{lemma:trimming inequality} to impose $\mathcal{F}(P^{+}, 1) = \mathcal{G}(P^{+}, 1)$. The inequality can then be shown by following the proof for Lemma \ref{lemma:conditional nesting implication}. 

For \eqref{eq:G Z = 0 case 1} and \eqref{eq:G Z = 0 case 2}, note that under the trimming rule it is possible that all observations with $Z = 0, p \in P^{+}$ are trimmed. If the trimmed sample includes observations with $Z = 0, p \in P^{+}$, then $\Pr(p \in P^{+}, Z = 1, S_1 = 1) > 0$, which implies \eqref{eq:G Z = 0 case 1 v}. Equation \eqref{eq:G Z = 0 case 1} can ten be shown by following the proof for Lemma \ref{lemma:conditional nesting implication}. If all observations with $Z = 0, p \in P^{+}$ are trimmed, then $\Pr(p \in P^{+}, Z = 1, S_1 = 1) = 0$, and \eqref{eq:G Z = 0 case 2} and \eqref{eq:G Z = 0 case 2 v} follow immediately. \end{proof}

\subsection{Proof of Proposition \ref{prop:distillation better}}

\begin{proof}[Proof of Statement 1]
    
Suppose that the coarsened propensity score nesting violation \eqref{eq:rewritten nesting} is positive. For a weighted sum of $\mathcal{F}(P^{-}, 0)$ and $\mathcal{F}(P^{-}, 1)$ to exceed a weighted sum of $\mathcal{F}(P^{+}, 0)$ and $\mathcal{F}(P^{+}, 1)$, it must be that
\begin{equation}
    \text{max}(\mathcal{F}(P^{-}, 0),  \mathcal{F}(P^{-}, 1)) > \text{min}(\mathcal{F}(P^{+}, 0), \mathcal{F}(P^{+}, 1)). \label{eq:proposition 5 inequality 3}
\end{equation}
Combining \eqref{eq:proposition 5 inequality 3} and Lemma \ref{lemma:conditional nesting implication}, there are two possible orderings of $\mathcal{F}(P^{+}, 0)$, $\mathcal{F}(P^{-}, 0)$, $\mathcal{F}(P^{+}, 1)$, and $\mathcal{F}(P^{-}, 1)$
\begin{align}
   \mathcal{F}(P^{-}, 0) \leq \mathcal{F}(P^{+}, 0) \leq \mathcal{F}(P^{-}, 1) \leq \mathcal{F}(P^{+}, 1), \label{eq:proposition 5 ordering 1} \\
   \mathcal{F}(P^{-}, 1) \leq \mathcal{F}(P^{+}, 1) \leq \mathcal{F}(P^{-}, 0) \leq \mathcal{F}(P^{+}, 0). \label{eq:proposition 5 ordering 2} 
\end{align}

Rewrite \eqref{eq:rewritten nesting} as
\begin{multline}
    w(1 | P^{-})  (\mathcal{F}(P^{-}, 1) - \mathcal{F}(P^{+}, 1))  + (1 - w(1 | P^{-})) (\mathcal{F}(P^{-}, 0) - \mathcal{F}(P^{+}, 0)) \\   -(w(1|P^{+}) - w(1 |P^{-})) (\mathcal{F}(P^{+}, 1) -  \mathcal{F}(P^{+}, 0)).
\end{multline} 
Lemma \ref{lemma:conditional nesting implication} implies that the first two terms are (weakly) negative. Turning to the final term, by Lemma \ref{lemma:weights inequality} 
\begin{equation*}
    w(1|P^{+}) - w(1 |P^{-}) \geq 0,
\end{equation*}
If \eqref{eq:proposition 5 ordering 1} holds, $\mathcal{F}(P^{+}, 1) > \mathcal{F}(P^{+},0)$, and this last term will also be negative, which implies that the nesting inequality for a coarsened propensity can never be positive. Hence, for this nesting inequality to be positive, \eqref{eq:proposition 5 ordering 2} must hold. 

Assuming \eqref{eq:proposition 5 ordering 2} holds, we obtain the inequality
\begin{align}
    v(P^{+}|0) \mathcal{G}(P^{+}, 0) + v(P^{-}| 0) \mathcal{G}(P^{-}, 0) &\geq
    \mathcal{F}(P^{-}, 0) \notag \\ &\geq 
     (1 - w(1 | P^{-})) \mathcal{F}(P^{-}, 0) + w(1 | P^{-})  \mathcal{F}(P^{-}, 1).  \label{eq:statement 1 inequality 1}
\end{align}
The first inequality holds by Lemma \ref{lemma:trimming inequality} and Lemma \ref{lemma:G conditional nesting implication}. The second inequality follows directly from \eqref{eq:proposition 5 ordering 2}. Maintaining \eqref{eq:proposition 5 ordering 2}, we can obtain the additional inequality
\begin{align}
  (1 - w(1 |P^{+})) \mathcal{F}(P^{+}, 0) + w(1 |P^{+})  \mathcal{F}(P^{+}, 1) &\geq 
   \mathcal{F}(P^{+}, 1) \notag \\  &\geq
   v(P^{+}|1) \mathcal{G}(P^{+}, 1) + v(P^{-}| 1) \mathcal{G}(P^{-}, 1). \label{eq:statement 1 inequality 2}
\end{align}
Here the first inequality follows from \eqref{eq:proposition 5 ordering 2}, while the second uses Lemma \ref{lemma:trimming inequality} and Lemma \ref{lemma:G conditional nesting implication}.

Summing \eqref{eq:statement 1 inequality 1} and \eqref{eq:statement 1 inequality 2}
\begin{multline*}
     v(P^{+}|0) \mathcal{G}(P^{+}, 0) + v(P^{-}| 0) \mathcal{G}(P^{-}, 0) +  (1 - w(1 |P^{+})) \mathcal{F}(P^{+}, 0) + w(1 |P^{+})  \mathcal{F}(P^{+}, 1) \geq \\
     (1 - w(1 | P^{-})) \mathcal{F}(P^{-}, 0) + w(1 | P^{-})  \mathcal{F}(P^{-}, 1) + v(P^{+}|1) \mathcal{G}(P^{+}, 1) + v(P^{-}| 1) \mathcal{G}(P^{-}, 1),
\end{multline*}
which can be rearranged to
\begin{multline*}
   v(P^{+}|0) \mathcal{G}(P^{+}, 0) + v(P^{-}| 0)\mathcal{G}(P^{-}, 0) - v(P^{+}|1) \mathcal{G}(P^{+}, 1) - v(P^{-}| 1) \mathcal{G}(P^{-}, 1) \geq \\
  (1 - w(1 | P^{-})) \mathcal{F}(P^{-}, 0) + w(1 | P^{-})  \mathcal{F}(P^{-}, 1) - (1 - w(1 |P^{+})) \mathcal{F}(P^{+}, 0) - w(1 |P^{+})  \mathcal{F}(P^{+}, 1).
\end{multline*}
The left-hand side is the distilled nesting violation \eqref{eq:rewritten distilled}. The right-hand side is the coarsened propensity score nesting violation \eqref{eq:rewritten nesting}.

\end{proof}

\begin{proof}[Proof of Statement 2]

Our strategy to show the second statement is as follows. First, we make assumptions such that, for a given $A$, the coarsened propensity score nesting violation \eqref{eq:rewritten nesting} is not positive. We then show that this does not preclude the distilled nesting violation \eqref{eq:rewritten distilled} from being positive.

The above argument established that \eqref{eq:proposition 5 ordering 2} must hold for the nesting inequality conditional on bins of $p$ to be positive. Therefore, to guarantee that the nesting inequality conditional on the coarsened propensity score is negative, we assume that  \eqref{eq:proposition 5 ordering 2} does \emph{not} hold i.e.
\begin{equation}
    \mathcal{F}(P^{-}, 0) < \mathcal{F}(P^{+}, 1). \label{eq:nesting negative sufficient condition}
\end{equation}

 Write the distilled nesting violation \eqref{eq:rewritten distilled} as
 \begin{multline}
    v(P^{-}|1)(\mathcal{G}(P^{-}, 0) - \mathcal{G}(P^{-}, 1)) + (v(P^{-}|0) - v(P^{-}|1))(\mathcal{G}(P^{-}, 0) -  \mathcal{G}(P^{+}, 1)) \\ + (1 - v(P^{-}| 0))(\mathcal{G}(P^{+}, 0) - \mathcal{G}(P^{+}, 1)). \label{eq:more rewritten distilled}
 \end{multline}
Expressed this way, \eqref{eq:rewritten distilled} is a weighted sum of the differences $\mathcal{G}(P^{-}, 0) - \mathcal{G}(P^{-}, 1)$, $\mathcal{G}(P^{-}, 0) -  \mathcal{G}(P^{+}, 1)$ and $\mathcal{G}(P^{+}, 0) - \mathcal{G}(P^{+}, 1)$. Lemma \ref{lemma:v inequality} guarantees that $v(P^{-}|0) - v(P^{-}|1) \geq 0$, so the weights are all (weakly) positive. By Lemma \ref{lemma:trimming inequality}, \eqref{eq:nesting negative sufficient condition} is equivalent to
\begin{equation}
    \mathcal{G}(P^{-}, 0) < \mathcal{G}(P^{+}, 1). \label{eq:G restriction nesting negative}
\end{equation}
Assuming \eqref{eq:G restriction nesting negative}, \eqref{eq:more rewritten distilled} can only be positive if at least one of $\mathcal{G}(P^{-}, 0) - \mathcal{G}(P^{-}, 1)$ and  $\mathcal{G}(P^{+}, 0) - \mathcal{G}(P^{+}, 1)$ is positive. 

Take the four $\mathcal{G}(P, z)$ as given and solve for $v(P^{-}|0)$ and $ v(P^{-}|1)$ such that \eqref{eq:more rewritten distilled} is positive subject to the constraints
\begin{align*}
   &0 \leq v(P^{-}|0) \leq 1, \\
   &0 \leq v(P^{-}|1) \leq 1, \\
   &v(P^{-}|1)\leq v(P^{-}|0).
\end{align*}

Doing so yields the following bounds for $v(P^{-}|0)$
\begin{equation}
    v(P^{-}|1)\leq v(P^{-}|0) \leq \text{min}\left\{1, \frac{v(P^{-}, 1)(\mathcal{G}(P^{+}, 1) - \mathcal{G}(P^{-}, 1)) + \mathcal{G}(P^{+}, 0) - \mathcal{G}(P^{+}, 1)}{\mathcal{G}(P^{+}, 0) - \mathcal{G}(P^{-}, 0)}\right\}. \label{eq:statement 2 bounds 1}
\end{equation}
and for $v(P^{-}|1)$
\begin{equation}
    \mathbbm{1}\left\{\mathcal{G}(P^{+}, 0) \leq \mathcal{G}((P^{+}, 1) \right\} \bar{v} \leq v(P^{-}|1) \leq 1 - \mathbbm{1}\left\{\mathcal{G}(P^{-}, 0) \leq \mathcal{G}(P^{-}, 1) \right\}(1 + \bar{v}). \label{eq:statement 2 bounds 2}
\end{equation}
with 
\begin{equation*}
    \bar{v} \equiv  \frac{\mathcal{G}(P^{+}, 1) -  \mathcal{G}(P^{+}, 0)}{ \mathcal{G}(P^{+}, 1) -  \mathcal{G}(P^{-}, 1) -  \mathcal{G}(P^{+}, 0) +  \mathcal{G}(P^{-}, 0)}.
\end{equation*}
Taking $v(P^{-}|1)$ as given,  \eqref{eq:statement 2 bounds 1} gives bounds for $v(P^{-}|0)$ such that \eqref{eq:more rewritten distilled} is positive. Equation \eqref{eq:statement 2 bounds 2} gives bounds for $v(P^{-}|1)$ such that the bounds in \eqref{eq:statement 2 bounds 1} are non-empty. If $\mathcal{G}(P^{-}, 0) \geq \mathcal{G}(P^{-}, 1)$ and $\mathcal{G}(P^{+}, 0) \geq \mathcal{G}(P^{+}, 1)$, for any $v(P^{-}|1)$ between 0 and 1 there exists $v(P^{-}|0)$ such that \eqref{eq:more rewritten distilled} is positive. If $\mathcal{G}(P^{-}, 0) \geq \mathcal{G}(P^{-}, 1)$ and $\mathcal{G}(P^{+}, 0) \leq \mathcal{G}(P^{+}, 1)$, a $v(P^{-}|0)$ where  \eqref{eq:more rewritten distilled} is positive exists if $v(P^{-}|1)$ is above $\bar{v}$.  If $\mathcal{G}(P^{-}, 0) \leq \mathcal{G}(P^{-}, 1)$ and $\mathcal{G}(P^{+}, 0) \geq \mathcal{G}(P^{+}, 1)$,  a $v(P^{-}|0)$ where  \eqref{eq:more rewritten distilled} is positive exists if $v(P^{-}|1)$ is below $-\bar{v}$. \end{proof}


\begin{proof}[Proof of Statement 3]

To prove the third statement, we first find a sufficient condition for the coarsened propensity score nesting violation \eqref{eq:rewritten nesting} to be negative for any $A$. We then show that it is possible to construct a process where this condition holds and the distilled nesting violation \eqref{eq:rewritten distilled} is positive.

Consider again \eqref{eq:rewritten nesting}. Lemma \ref{lemma:weights inequality} and Lemma \ref{lemma:conditional nesting implication} together imply 
\begin{equation*}
     w(1 | P^{-})  \mathcal{F}(P^{-}, 1) \leq  
     w(1 | P^{+})  \mathcal{F}(P^{+}, 1) .
\end{equation*}
A sufficient condition for \eqref{eq:rewritten nesting} to be negative for a given $A$ is then
\begin{equation*}
    (1 - w(1 | P^{-})) \mathcal{F}(P^{-}, 0) \leq (1 - w(1 | P^{+})) \mathcal{F}(P^{+}, 0).
\end{equation*}
and a sufficient condition for \eqref{eq:rewritten nesting} to be negative for all $A$ where $\mathcal{F}(P^{+}, 0) > 0$ is 
\begin{equation}\label{eq:statement 3 condition 1}
  \underset{A}{\text{sup}}\left\{\frac{\mathcal{F}(P^{-}, 0)}{\mathcal{F}(P^{+}, 0)}\right\} \leq  \frac{ 1 - w(1 | P^{+})}{1 - w(1 | P^{-})}.
\end{equation}
Note that Lemma \ref{lemma:conditional nesting implication} implies the left-hand side is weakly below 1. Hence \eqref{eq:statement 3 condition 1} will always hold if $w(1 | P^{+}) = w(1 | P^{-})$. That is, if the mixture of $Z = 0$ and $Z = 1$ observations is sufficiently similar in $P^{-}$ and $P^{+}$, then \eqref{eq:rewritten nesting} will never be positive.\\

Now consider the distilled nesting violation \eqref{eq:rewritten distilled}. Assume that \eqref{eq:proposition 5 ordering 2} holds.  
Lemmas \ref{lemma:v inequality}, \ref{lemma:conditional nesting implication}, \ref{lemma:trimming inequality} and \ref{lemma:G conditional nesting implication} together imply
\begin{equation*}
     v(P^{-}| 0) \mathcal{G}(P^{-}, 0) \geq  v(P^{-}|1) \mathcal{G}(P^{-}, 1).
\end{equation*}
A sufficient condition for \eqref{eq:rewritten distilled} to be positive at a given $A$ is
\begin{equation*}
    v(P^{+}|0) \mathcal{G}(P^{+}, 0) \geq   v(P^{+}|1) \mathcal{G}(P^{+}, 1).
\end{equation*}
A sufficient condition for \eqref{eq:rewritten distilled} to be positive for some $A$ where $\mathcal{G}(P^{+}, 1) > 0$ is thus
\begin{equation}\label{eq:statement 3 condition 2}
    \underset{A}{\text{sup}}\left\{\frac{\mathcal{G}(P^{+}, 0)}{\mathcal{G}(P^{+}, 1)}\right\} \geq \frac{ v(P^{+}|1)}{v(P^{+}|0)}.
\end{equation}
For any $A$ where \eqref{eq:proposition 5 ordering 2} holds, the left-hand side is bounded from below by 1. Equation \eqref{eq:statement 3 condition 2} thus states that, if \eqref{eq:proposition 5 ordering 2} holds for some $A$ and the mixture of observations from $P^{-}$ and $P^{+}$ is sufficiently similar for $Z = 0$ and $Z = 1$, then \eqref{eq:rewritten distilled} is positive for that $A$.

A case where \eqref{eq:statement 3 condition 1} and \eqref{eq:statement 3 condition 2} are immediate is when there exists $A$ such that $\mathcal{G}(P^{+}, 0) > \mathcal{G}(P^{+}, 1)$ and $Z$ is an irrelevant instrument that is independent of $X$. Then
\begin{align*}
    w(1 | P^{+}) &= w(1 | P^{-}), \\
    v(P^{+}|1) &= v(P^{+}|0).
\end{align*}
Hence, any violation can only be detected by conditioning on $Z$. \end{proof}

\section{Numerical Example}\label{app:numerical example}


This section derives properties of the conditional joint densities of partial residuals and treatment assignment for the process introduced in Section \ref{sec:example process}.

Let $\tilde{\theta}_1$ and $\tilde{\theta}_0$ be the asymptotic bias of the estimates of $\theta_1$ and $\theta_0$
\begin{align*}
  \tilde{\theta}_1 &= \hat{\theta}_1^{*} - \theta_1, \\
  \tilde{\theta}_0 &= \hat{\theta}_0^{*} - \theta_0.
\end{align*}
Using this notation , the asymptotic partial residuals are
\begin{equation*}
    U = \nu Z - \nu (1 - Z) + D(\bar{U}_1 - X^{\prime}\tilde{\theta}_1) + (1 - D)(\bar{U}_0 - X^{\prime}\tilde{\theta}_0) .
\end{equation*}
Given  process of Section \ref{sec:example process}, the following expressions can be derived for the joint density of $(U, D = 1)$ conditional on $p$ and $Z$
\begin{align*}
    &f(U = u, D = 1|p , Z = 1) = \int_{- \Phi^{-1}(p)}^{\infty}  \phi(U_D) \frac{1}{ \sqrt{2 \pi  \sigma_{\hat{U}_1|U_{D}, Z = 1}^2}}  \text{exp}\left(-\frac{\left(u - \nu - \mu_{\tilde{U}_1|U_D}(U_D) - \frac{\delta_1^{\prime}\tilde{\theta_1}}{\delta_1^{\prime}\delta_1}X^{\prime}\delta_1 \right)^2}{2  \sigma_{\hat{U}_1|U_{D}, Z = 1}^2} \right)  d U_D, \\
      &f(U = u, D = 1|p , Z = 0) = \int_{- \Phi^{-1}(p)}^{\infty}  \phi(U_D) \frac{1}{ \sqrt{2 \pi  \sigma_{\hat{U}_1|U_{D}, Z = 0}^2}}  \text{exp}\left(-\frac{\left(u + \nu -\mu_{\tilde{U}_1|U_D}(U_D) - \frac{\delta_0^{\prime}\tilde{\theta_1}}{\delta_0^{\prime}\delta_0}X^{\prime}\delta_0 \right)^2}{2  \sigma_{\hat{U}_1|U_{D}, Z = 0}^2} \right)  d U_D.
\end{align*}
where $\phi(\cdot)$ is the standard univariate normal density and
\begin{align*}
    &\mu_{\tilde{U}_1|U_D}(U_D) = \mu_1 + \sigma_1 \rho_{D, 1} U_D, \\
    &\sigma_{\hat{U}_1|U_{D}, Z = 1}^2 = \sigma_1^2(1 - \rho_{D,1} ^2) + \left(1 - \frac{(\tilde{\theta}_1^{\prime}\delta_1)^2}{\tilde{\theta}_1^{\prime}\tilde{\theta}_1 \delta^{\prime}\delta} \right) \tilde{\theta}_1^{\prime}\tilde{\theta}_1, \\
    &\sigma_{\hat{U}_1|U_{D}, Z = 0}^2 = \sigma_1^2(1 - \rho_{D,1} ^2) + \left(1 - \frac{(\tilde{\theta}_1^{\prime}\delta_0)^2}{\tilde{\theta}_1^{\prime}\tilde{\theta}_1 \delta^{\prime}\delta} \right) \tilde{\theta}_1^{\prime}\tilde{\theta}_1,
\end{align*}
with
\begin{align}
    X^{\prime}\delta_1 &= \Phi^{-1}(p) - \alpha, \label{eq:Xd_1 p relation}\\
    X^{\prime}\delta_0 &= \Phi^{-1}(p) + \alpha. \label{eq:Xd_0 p relation}
\end{align}
The expressions for the conditional joint densities of $(U, D = 0)$ are analogous. 
If the asymptotic estimates are unbiased (i.e. $\tilde{\theta}_1 = 0$), then the propensity score only enters through the lower bound of each integral. The derivative of the joint density with respect to $p$ is then weakly positive for all $A$ and $p$ and both $Z = 0$ and $Z = 1$. Condition (iv) of Proposition \ref{prop:distillation better} must then hold for any non-overlapping intervals of $p$ we consider. However, if $\tilde{\theta}_1 \neq 0$, there is an additional effect through \eqref{eq:Xd_1 p relation} and \eqref{eq:Xd_0 p relation}, changing $p$ affects $X^{\prime}\delta_1$ and $X^{\prime}\delta_0$. Intuitively, with biased estimates, partialling out does not remove the effect of covariates from the residuals. Changing propensity scores changes the distribution of covariates, which shifts the distribution of partial residuals. 

$X^{\prime}\delta_1$ and $ X^{\prime}\delta_0$ enter the density multiplied by  $\tilde{\theta}_1^{\prime}\delta_1$ and $\tilde{\theta}_1^{\prime}\delta_0$ respectively. Rather than characterising $\tilde{\theta}_1$, we characterise these quantities. The following proposition holds
\begin{proposition} \label{prop:bias simplification}
Consider the terms  $\delta_1^{\prime}\tilde{\theta}_1$, $\delta_0^{\prime}\tilde{\theta}_1$, $\delta_1^{\prime}\tilde{\theta}_0$ and $\delta_0^{\prime}\tilde{\theta}_0$. Assume \\
\noindent (i) The partially linear model is estimated as described in Section \ref{sec:test procedure}. \\
\noindent (ii) The propensity score function $p(X, Z)$ and the conditional expectations $E[X|p(X,Z)]$ and $E[Y|p(X,Z)]$ are consistently estimated.\\
\noindent (iii) There is no perfect multicollinearity among the components of $X - E[X|p(X,Z)]$.\\
\noindent (iv) The matrix $E$ defined in \eqref{eq:matrix that must be invertible 1} and the matrix $G$ defined in \eqref{eq:matrix that must be invertible 2} are invertible.\\
Then
\begin{enumerate}
    \item The following equalities hold
    \begin{align*}
        &\tilde{\theta}_1^{\prime}\delta_1 = \tilde{\theta}_0^{\prime}\delta_0, \\
        &\tilde{\theta}_0^{\prime}\delta_1 = \tilde{\theta}_1^{\prime}\delta_0.
    \end{align*}
    \item $\tilde{\theta}_1^{\prime}\delta_1$ and $\tilde{\theta}_1^{\prime}\delta_0$ have the form
    \begin{align*}
        \tilde{\theta}_1^{\prime}\delta_1 &= \mathcal{H}_1(\alpha, \rho_{\delta}, \delta^{\prime}\delta) \nu, \\
        \tilde{\theta}_1^{\prime}\delta_0 &= \mathcal{H}_0(\alpha, \rho_{\delta}, \delta^{\prime}\delta) \nu.
    \end{align*}
    
\end{enumerate}
\end{proposition}
where $\mathcal{H}_1$ and $\mathcal{H}_0$ are functions. In the numerical exercise, the matrices introduced in assumption (iv) are always invertible. The first statement implies that the coefficients for $X^{\prime}\delta_1$ and $ X^{\prime}\delta_0$ in the subdensities for $(U, D = 0)$ will mirror those in the subdensities for $(U, D = 1)$, so it is sufficient to only examine one set of subdensities. The second statement implies that, once we have fixed $\nu, \alpha, \delta^{\prime}\delta$ and $\rho_{\delta}$, the actual values of $\delta_0$ and $\delta_1$ and all other parameters are irrelevant.

\subsection{Proof of Proposition \ref{prop:bias simplification}}\label{sec:bias simplification proof}

The proof proceeds as follows. First, we state several properties of density of the propensity score and the conditional density of $X$ given the propensity score. These properties are then used to show a series of lemmas. Finally, the lemmas are used to show the two statements of Proposition \ref{prop:bias simplification}. 

\subsubsection{Propensity Score Density}

The propensity score function for this process is
 \begin{equation}
     p(X, Z) = \Phi(Z(X^{\prime}\delta_1 + \alpha) + (1 - Z)(X^{\prime}\delta_0 - \alpha)). \label{eq:numerical example propensity score}
 \end{equation}

Given the distributions of $X$ and $Z$, and this propensity score function, the joint density of $p$ and $Z$ is 
\begin{align}
    f(p, Z = 0) &= \frac{1}{2}  \frac{1}{\sqrt{2 \pi \delta^{\prime}\delta}} \text{exp}\left(-\frac{(\Phi^{-1}(p) + \alpha)^{2}}{2 \delta^{\prime}\delta} \right) \frac{1}{\phi(\Phi^{-1}(p))} \label{eq:joint density p Z = 0},\\
    f(p, Z = 1) &= \frac{1}{2}  \frac{1}{\sqrt{2 \pi \delta^{\prime}\delta}} \text{exp}\left(-\frac{(\Phi^{-1}(p) - \alpha)^{2}}{2 \delta^{\prime}\delta} \right) \frac{1}{\phi(\Phi^{-1}(p))}. \label{eq:joint density p Z = 1}
\end{align}
Summing \eqref{eq:joint density p Z = 0} and \eqref{eq:joint density p Z = 1}, the unconditional density of the propensity score is  
\begin{equation}
    f(p) = \frac{1}{\sqrt{2 \pi \delta^{\prime}\delta}}\text{exp}\left(-\frac{\Phi^{-1}(p)^2 + \alpha^{2}}{2 \delta^{\prime}\delta} \right)\text{cosh}\left(\frac{\Phi^{-1}(p)\alpha}{\delta^{\prime}\delta} \right) \frac{1}{\phi(\Phi^{-1}(p))}.  \label{eq:numerical example p unconditional density}
\end{equation}
Notice that this density satisfies 
\begin{equation}
     f(p) = f(1 - p), \label{eq:symmetry 1}
\end{equation}
which implies
\begin{equation}
     \int_0^1 p^2 f(p) dp = \int_0^1 (1 - p)^2 f(p) dp. \label{eq: EXX symmetry 1}
\end{equation}

Combining \eqref{eq:joint density p Z = 1} and \eqref{eq:numerical example p unconditional density}, the probability of $Z = 1$ conditional on the propensity score is
 \begin{equation}
       \Pr(Z=1|p) =  \frac{1}{2} + \frac{\text{tanh}(\alpha \Phi^{-1}(p) / \delta^{\prime}\delta) }{2},
 \end{equation}   
 which satisfies
 \begin{equation}
     \Pr(Z=1|p) = 1 - \Pr(Z=1|1-p). \label{eq:symmetry 2}
 \end{equation}
The symmetry relations \eqref{eq:symmetry 1} and \eqref{eq:symmetry 2}, and $\Pr(Z=0|p) = 1  - \Pr(Z=1|p)$ together imply
\begin{align}
     & \int_0^1 p^2 \Pr(Z=1|p) f(p) dp = \int_0^1 (1 - p)^2 \Pr(Z=0|p) f(p) dp, \label{eq: EXX symmetry 2}\\
     & \int_0^1 p^2 \Pr(Z=0|p) f(p) dp = \int_0^1 (1 - p)^2 \Pr(Z=1|p) f(p) dp,  \label{eq: EXX symmetry 3}\\
      & \int_0^1 p(1 - p) \Pr(Z = 1 |p)f(p) dp =  \int_0^1 p(1 - p) \Pr(Z = 0 |p)f(p) dp, \label{eq: EXX symmetry 4}\\
      & \int_0^1 p \Pr(z=0|p)\Pr(z=1|p) f(p) dp =  \int_0^1 (1 - p) \Pr(z=0|p)\Pr(z=1|p) f(p) dp, \label{eq:E XZ symmetry 1}\\ 
       & \int_0^1 p^2 \Pr(Z = 0|p) \Pr(Z =1|p) f(p) dp = \int_0^1 (1 - p)^2 \Pr(Z = 0|p) \Pr(Z =1|p) f(p) dp. \label{eq: EXX symmetry 5}
\end{align}

Finally, as $\Phi^{-1}(p) = - \Phi^{-1}(1 - p)$ and $\Phi^{-1}(p)^2 = \Phi^{-1}(1 - p)^2$
\begin{align}
    & \int_0^1 p (1 - p) \Pr(Z=1|p)\Pr(Z=0|p) \Phi^{-1}(p) f(p) dp = 0, \label{eq: EXX symmetry 6}\\
    & \int_0^1 (1 - p) \Pr(z=0|p)\Pr(z=1|p) \Phi^{-1}(p) f(p) dp = -  \int_0^1 p \Pr(z=0|p)\Pr(z=1|p) \Phi^{-1}(p) f(p) dp, \label{eq:E XZ symmetry 2} \\
   &   \int_0^1 p^2 \Pr(Z = 0|p) \Pr(Z =1|p) \Phi^{-1}(p) f(p) dp = - \int_0^1 (1 - p)^2 \Pr(Z = 0|p) \Pr(Z =1|p) \Phi^{-1}(p) f(p) dp, \label{eq: EXX symmetry 7}\\
   & \int_0^1 p^2 \Pr(Z = 0|p) \Pr(Z =1|p) \Phi^{-1}(p)^2 f(p) dp = \int_0^1 (1 - p)^2 \Pr(Z = 0|p) \Pr(Z =1|p) \Phi^{-1}(p)^2 f(p) dp. \label{eq: EXX symmetry 8}
\end{align}

Notice that the only parameters which enter the propensity score density and conditional distribution of $Z$ are $\alpha$ and $\delta^{\prime}\delta$. As a consequence, these are the only parameters which affect the integrals above. 

\subsubsection{Conditional Density of $X$}

Given the propensity score function \eqref{eq:numerical example propensity score}, conditioning on $(p, Z = 1)$ is equivalent to conditioning on $X^{\prime}\delta_1 = \Phi^{-1}(p) - \alpha$, and conditioning on $(p, Z = 0)$ is equivalent to conditioning on $X^{\prime}\delta_0 = \Phi^{-1}(p) + \alpha$. Using that $X$ is normally distributed
\begin{align}
     & E[X|p, Z = 1] =  E[X|X^{\prime}\delta_1 = \Phi^{-1}(p) - \alpha] = \frac{\delta_1}{\delta^{\prime}\delta} (\Phi^{-1}(p) - \alpha ), \label{eq:EX|p. Z = 1}\\
    & E[X|p, Z = 0] = E[X|X^{\prime}\delta_0 = \Phi^{-1}(p) + \alpha] = \frac{\delta_0}{\delta^{\prime}\delta} (\Phi^{-1}(p) + \alpha ), \label{eq:EX|p. Z = 0}  \\
    & E[X|p] = \Pr(Z=1|p)\frac{\delta_1}{\delta^{\prime}\delta} (\Phi^{-1}(p) - \alpha ) + \Pr(Z = 0|p) \frac{\delta_0}{\delta^{\prime}\delta} (\Phi^{-1}(p) + \alpha ). \label{eq:EX|p}
\end{align}
Let
\begin{equation*}
    \text{Cov}(X - E[X|p, Z]|p, Z) \equiv E\left[(X - E[X|p, Z])(X - E[X|p, Z])^{\prime}|p, Z\right], \\
\end{equation*}
then
\begin{align}
    \text{Cov}(X - E[X|p, Z]|p, Z = 1)  &=   \text{Cov}(X - E[X|p, Z]|X^{\prime}\delta_1 = \Phi^{-1}(p) - \alpha) = I - \frac{\delta_1\delta_1^{\prime}}{\delta^{\prime}\delta},  \label{eq:cov X|p. Z = 1}\\
    \text{Cov}(X - E[X|p, Z]|p, Z = 0)  &=   \text{Cov}(X - E[X|p, Z]|X^{\prime}\delta_0 = \Phi^{-1}(p) + \alpha) = I - \frac{\delta_0\delta_0^{\prime}}{\delta^{\prime}\delta}.  \label{eq:cov X|p. Z = 0} 
\end{align}
Similarly, define
\begin{equation*}
    \text{Cov}(X - E[X|p]|p) \equiv E[(X - E[X|p])(X - E[X|p])^{\prime}|p].
\end{equation*}
$\text{Cov}(X - E[X|p]|p)$ satisfies the decomposition
\begin{equation*}
  \text{Cov}(X - E[X|p]|p) =  E[\text{Cov}(X - E[X|p, Z]|p, Z)|p] + E[\text{Cov}(E[X|p,Z] - E[X|p]|p, Z)|p].    
\end{equation*}
Using \eqref{eq:EX|p. Z = 1}, \eqref{eq:EX|p. Z = 0}, \eqref{eq:EX|p}, \eqref{eq:cov X|p. Z = 1} and \eqref{eq:cov X|p. Z = 0} and simplifying
\begin{multline}
    \text{Cov}(X - E[X|p]|p) = I - \Pr(Z =1 |p) \frac{\delta_1\delta_1^{\prime}}{\delta^{\prime}\delta} - \Pr(Z = 0|p) \frac{\delta_0\delta_0^{\prime}}{\delta^{\prime}\delta} \\
   + \frac{ \delta_1\delta_1^{\prime}}{(\delta^{\prime}\delta)^2}(\Phi^{-1}(p) - \alpha)^2 - \frac{\delta_0\delta_1^{\prime} + \delta_1\delta_0^{\prime}}{(\delta^{\prime}\delta)^2} (\Phi^{-1}(p) - \alpha)(\Phi^{-1}(p) + \alpha) + \frac{ \delta_0\delta_0^{\prime}}{(\delta^{\prime}\delta)^2} (\Phi^{-1}(p) + \alpha)^2. \label{eq:cov X - E[X|p] | p}
\end{multline}

\subsubsection{Lemmas}

Let $\mathbf{X}$ be a column vector of length $2k_X$ which stacks $p(X - E[X|p])$ and $(1-p)(X-E[X|p])$. Let  $\mathbf{Z} \equiv Z - E[Z|p]$.

\begin{lemma}\label{lemma:E XZ}
\begin{equation*}
    E[\mathbf{X}\mathbf{Z}] = \begin{pmatrix}
        w_1 \delta_1 + w_0 \delta_0 \\
        w_1 \delta_0 + w_0 \delta_1
    \end{pmatrix},
\end{equation*}
where $w_0$ and $w_1$ are scalars that depend on $\alpha$ and $\delta^{\prime}\delta$ only.  
\end{lemma}

\begin{lemma}\label{lemma:E XX}
\begin{equation*}
     E[\mathbf{X}\mathbf{X}^{\prime}]^{-1} = \begin{pmatrix}
        \Omega_1 & \Omega_2 \\
        \Omega_2^{\prime} & \Omega_3
    \end{pmatrix},
\end{equation*} 
where $\Omega_1$, $\Omega_2$, and $\Omega_3$ are $K \times K$ matrices.  
\end{lemma}

\begin{lemma}\label{lemma:equalities}

The following equalities hold
\begin{align*}
    &\delta_1^{\prime}\Omega_1\delta_1 = \delta_0^{\prime}\Omega_3\delta_0, \\
    &\delta_0^{\prime}\Omega_1 \delta_0 =  \delta_1^{\prime}\Omega_3\delta_1,\\
    &\delta_0^{\prime}\Omega_1 \delta_1 =  \delta_1^{\prime}\Omega_3\delta_0,\\
    &\delta_1^{\prime}\Omega_1 \delta_0 =  \delta_0^{\prime}\Omega_3\delta_1,\\
    &\delta_1^{\prime}\Omega_2\delta_1 =  \delta_0^{\prime}\Omega_2^{\prime} \delta_0.
\end{align*}

\end{lemma}

\begin{lemma}\label{lemma:E XX inverse parameters}
The constants $\delta_i^{\prime} \Omega_k \delta_j$ for $i,j = \{0, 1\}$, $k = \{1, 2, 3\}$ depend on $\alpha$, $\delta^{\prime}\delta$ and $\rho_{\delta}$ only.
\end{lemma}

\begin{proof}[Proof of Lemma \ref{lemma:E XZ}]

\begin{align}
    E[\mathbf{X}\mathbf{Z}] &= \begin{pmatrix}
        E[p(X-E[X|p])(Z-E[Z|p]] \\
        E[(1 - p)(X-E[X|p])(Z-E[Z|p]] 
    \end{pmatrix}, \notag \\
    &= \begin{pmatrix}
        \int_0^1 p \Pr(z = 0|p) \Pr(z = 1|p) (E[X|p, Z = 1] - E[X|p, Z = 0]) f(p) dp \\
        \int_0^1 (1 - p) \Pr(z = 0|p) \Pr(z = 1|p) (E[X|p, Z = 1] - E[X|p, Z = 0]) f(p) dp
    \end{pmatrix}, \notag \\
    & =  \frac{1}{\delta^{\prime}\delta}\begin{pmatrix}
       \int_0^1 p \Pr(z = 0|p) \Pr(z = 1|p) (\delta_1(\Phi^{-1}(p) - \alpha) - \delta_0 (\Phi^{-1}(p) + \alpha)) f(p) dp \\
        \int_0^1 (1 - p) \Pr(z = 0|p) \Pr(z = 1|p) (\delta_1(\Phi^{-1}(p) - \alpha) - \delta_0 (\Phi^{-1}(p) + \alpha)) f(p) dp
    \end{pmatrix}, \notag \\
    &= \begin{pmatrix}
        w_1 \delta_1 + w_0 \delta_0 \\
        w_1 \delta_0 + w_0 \delta_1
    \end{pmatrix}, \label{eq:EXZ expression}
\end{align}
where $w_1$ and $w_0$ are the scalars 
\begin{align*}
    w_1 &= \frac{\int_0^1 p \Pr(z=0|p)\Pr(z=1|p) (\Phi^{-1}(p) - \alpha) f(p) dp}{\delta^{\prime}\delta}, \\
    w_0 &= - \frac{\int_0^1 p \Pr(z=0|p)\Pr(z=1|p) (\Phi^{-1}(p) + \alpha) f(p) dp}{\delta^{\prime}\delta}.
\end{align*}
The first line uses the definitions of $\mathbf{X}$ and $\mathbf{Z}$. The second line uses the law of iterated expectations, and the $Z$ is binary so $E[Z|p] = \Pr(Z = 1|p)$.  The third uses \eqref{eq:EX|p. Z = 1} and \eqref{eq:EX|p. Z = 0}. Equations \eqref{eq:E XZ symmetry 1} and \eqref{eq:E XZ symmetry 2} are used in the fourth.

$w_1$ and $w_0$ are integrals of functions that depend on $\alpha$ and $\delta^{\prime}\delta$ only. Hence, they themselves will depend on $\alpha$ and $\delta^{\prime}\delta$ only. \end{proof}

\subsubsection{Characterising  $E[\mathbf{X}\mathbf{X}^{\prime}]$}

By definition
\begin{equation*}
   E[\mathbf{X}\mathbf{X}^{\prime}] = \begin{pmatrix}
         E[p^2 (X - E[X|p])(X-E[X|p])^{\prime}] & E[p(1-p) (X - E[X|p])(X-E[X|p])^{\prime}] \\
         E[p (1-p) (X - E[X|p])(X-E[X|p])^{\prime}] & E[(1-p)^2 (X - E[X|p])(X-E[X|p])^{\prime}] 
     \end{pmatrix}.    
\end{equation*}
All four components of this matrix can be written in the form $E[p^{i}(1 - p)^{j} (X - E[X|p])(X-E[X|p])^{\prime}]$ where $i$ and $j$ take values $\{0, 1, 2\}$. Notice that
\begin{align*}
    E[p^{i}(1 - p)^{j} (X - E[X|p])(X-E[X|p])^{\prime}] &=  \int_0^1 p^{i}(1 - p)^{j} E[(X - E[X|p])(X-E[X|p])^{\prime}|p] f(p) dp, \\ 
    &=  \int_0^1 p^{i}(1 - p)^{j} \text{Cov}(X - E[X|p]|p) f(p) dp.
\end{align*}
Using \eqref{eq:cov X - E[X|p] | p} to substitute for $\text{Cov}(X - E[X|p]|p)$, we can obtain
\begin{align}
    & E[\mathbf{X}\mathbf{X}^{\prime}] = \begin{pmatrix}
        \Psi_1 & \Psi_2 \\
        \Psi_2 & \Psi_3 
     \end{pmatrix},   \label{eq: EXX not inverted} \\
    & \Psi_1 =  I a_1 + \delta_1\delta_1^{\prime} a_2 +  \delta_0\delta_0^{\prime} a_3 +  (\delta_0\delta_1^{\prime} +  \delta_1\delta_0^{\prime})a_4, \notag \\
    & \Psi_3 =  I a_1 + \delta_1\delta_1^{\prime} a_3 +  \delta_0\delta_0^{\prime} a_2 +  (\delta_0\delta_1^{\prime} +  \delta_1\delta_0^{\prime})a_4, \notag \\
    & \Psi_2 =  I b_1 + (\delta_1\delta_1^{\prime} +  \delta_0\delta_0^{\prime}) b_2 +  (\delta_0\delta_1^{\prime} +  \delta_1\delta_0^{\prime}) b_3, \notag
\end{align}
with
\begin{align*}
    a_1 &= \int p^2 f(p) dp, \\
    a_2 &= -\frac{\int p^2 \Pr(Z=1|p) f(p) dp}{\delta^{\prime}\delta} + \frac{\int p^2 \Pr(Z=1|p) \Pr(Z=0|p) (\Phi^{-1}(p) - \alpha)^2 f(p) dp}{(\delta^{\prime}\delta)^2}, \\
    a_3 &= -\frac{\int p^2 \Pr(Z=0|p) f(p) dp}{\delta^{\prime}\delta} + \frac{\int p^2 \Pr(Z=1|p) \Pr(Z=0|p) (\Phi^{-1}(p) + \alpha)^2 f(p) dp}{(\delta^{\prime}\delta)^2}, \\
    a_4 &= \frac{\int p^2 \Pr(Z=1|p) \Pr(Z=0|p) (\alpha^2 - \Phi^{-1}(p)^2) f(p) dp}{(\delta^{\prime}\delta)^2}, \\
    b_1 &= \int p(1 - p) f(p) dp, \\
    b_2 &= -\frac{\int p(1 - p) \Pr(Z=1|p) f(p) dp}{\delta^{\prime}\delta} + \frac{\int p(1 - p) \Pr(Z=1|p) \Pr(Z=0|p) (\Phi^{-1}(p)^2 + \alpha^2) f(p) dp}{(\delta^{\prime}\delta)^2}, \\
    b_3 &=\frac{\int p(1 - p) \Pr(Z=1|p) \Pr(Z=0|p) (\alpha^2 - \Phi^{-1}(p)^2) f(p) dp}{(\delta^{\prime}\delta)^2},
\end{align*}
where \eqref{eq: EXX symmetry 1}, \eqref{eq: EXX symmetry 2}, \eqref{eq: EXX symmetry 3}, \eqref{eq: EXX symmetry 4}, \eqref{eq: EXX symmetry 5}, \eqref{eq: EXX symmetry 6}, \eqref{eq: EXX symmetry 7}, and \eqref{eq: EXX symmetry 8} have all been imposed.  By pre-multiplying the components of $E[\mathbf{X}\mathbf{X}^{\prime}]$ by $\delta_1$ and $\delta_0$, we can obtain the following
\begin{align}
    \delta_1^{\prime}\Psi_1 &= c_1 \delta_1^{\prime} + c_0 \delta_0^{\prime}, \label{eq: push through 1}\\
    \delta_0^{\prime}\Psi_1 &= d_0 \delta_1^{\prime} + d_1 \delta_0^{\prime}, \label{eq: push through 2}\\
    \delta_1^{\prime}\Psi_3 &= d_1 \delta_1^{\prime} + d_0 \delta_0^{\prime}, \label{eq: push through 3}\\
    \delta_0^{\prime}\Psi_3 &= c_0 \delta_1^{\prime} + c_1 \delta_0^{\prime}, \label{eq: push through 4}\\
    \delta_1^{\prime}\Psi_2 &= e_1 \delta_1^{\prime} + e_0 \delta_0^{\prime}, \label{eq: push through 5}\\ 
    \delta_0^{\prime}\Psi_2 &= e_0 \delta_1^{\prime} + e_1 \delta_0^{\prime}, \label{eq: push through 6}
\end{align}
where
\begin{align*}
    c_0 &= a_3  \rho_{\delta}\delta^{\prime}\delta + a_4 \delta^{\prime}\delta, & d_0 &= a_2  \rho_{\delta}\delta^{\prime}\delta + a_4 \delta^{\prime}\delta, &  e_0 &= b_2 \rho_{\delta} \delta^{\prime}\delta + b_3 \delta^{\prime}\delta, \\
    c_1 &= a_1 + a_2 \delta^{\prime}\delta + \rho_{\delta}\delta^{\prime}\delta a_4, & d_1 &= a_1 + a_3 \delta^{\prime}\delta + \rho_{\delta}\delta^{\prime}\delta a_4, & e_1 &= b_1 + b_2 \delta^{\prime}\delta + b_3 \rho_{\delta}\delta^{\prime}\delta. 
\end{align*}
Note that
\begin{equation}
    c_0 + \rho_{\delta} c_1 = d_0 + \rho_{\delta} d_1. \label{eq:symmetry condtion}
\end{equation}

\begin{proof}[Proof of Lemma  \ref{lemma:E XX}]

 $E[\mathbf{X}\mathbf{X}^{\prime}]$ is invertible by assumption (iii). Notice that the off-diagonal blocks in \eqref{eq: EXX not inverted} are identical, and that each block is a symmetric matrix. Applying the formulas for block matrix inversion and imposing these properties gives
\begin{equation}
   \begin{pmatrix}
        \Psi_1 & \Psi_2 \\
        \Psi_2 & \Psi_3 
     \end{pmatrix}^{-1} = \begin{pmatrix}
        \Omega_1 & \Omega_2 \\
        \Omega_2^{\prime} & \Omega_3 \label{eq:EXX inverse}
    \end{pmatrix}.
\end{equation} \end{proof}

\begin{proof}[Proof of Lemma \ref{lemma:equalities}]

Equation \eqref{eq:EXX inverse} implies 
\begin{equation}
\begin{pmatrix}
    \Psi_1 \Omega_1 + \Psi_2 \Omega_2^{\prime} & \Psi_1 \Omega_2 + \Psi_2 \Omega_3 \\
      \Psi_2 \Omega_1 + \Psi_3 \Omega_2^{\prime}  & \Psi_2 \Omega_2 + \Psi_3 \Omega_3
\end{pmatrix} = \begin{pmatrix}
    I & 0 \\
    0 & I
\end{pmatrix}.  \label{eq:EXX inverse block relations}
\end{equation}

Treating each block of \eqref{eq:EXX inverse block relations} as a system of equations, consider pre and post-multiplying the top-left blocks by $\delta_1$
\begin{equation*}
     \delta_1^{\prime}\Psi_1 \Omega_1\delta_1 + \delta_1^{\prime}\Psi_2 \Omega_2^{\prime}\delta_1 = \delta^{\prime}\delta .
\end{equation*}
Using \eqref{eq: push through 1} and \eqref{eq: push through 5}, this becomes
\begin{equation*}
    c_1   \delta_1^{\prime}\Omega_1\delta_1 + c_0 \delta_0^{\prime}\Omega_1\delta_1 + e_1  \delta_1^{\prime} \Omega_2^{\prime}\delta_1 + e_0  \delta_0^{\prime} \Omega_2^{\prime}\delta_1 = \delta^{\prime}\delta. 
\end{equation*}
Repeat this with all four blocks and combinations of $\delta_1$ and $\delta_0$ to obtain sixteen conditions. These sixteen conditions can be written as the system of equations
\begin{equation*}
       E
     \begin{pmatrix}
        \delta_1^{\prime}\Omega_1\delta_1 & \delta_1^{\prime}\Omega_2\delta_1 & \delta_1^{\prime}\Omega_1\delta_0 & \delta_1^{\prime}\Omega_2\delta_0  \\
        \delta_0^{\prime}\Omega_1\delta_1 & \delta_0^{\prime}\Omega_2\delta_1 & \delta_0^{\prime}\Omega_1\delta_0 & \delta_0^{\prime}\Omega_2\delta_0 \\
        \delta_1^{\prime}\Omega_2^{\prime}\delta_1 & \delta_1^{\prime}\Omega_3\delta_1 & \delta_1^{\prime}\Omega_2^{\prime}\delta_0 & \delta_1^{\prime}\Omega_3\delta_0 \\
        \delta_0^{\prime}\Omega_2^{\prime}\delta_1 & \delta_0^{\prime}\Omega_3\delta_1 &  \delta_0^{\prime}\Omega_2^{\prime}\delta_0 & \delta_0^{\prime}\Omega_3\delta_0
    \end{pmatrix} = \begin{pmatrix}
        \delta^{\prime}\delta & 0 &  \rho_{\delta}\delta^{\prime}\delta & 0\\
        0 &  \delta^{\prime}\delta & 0 & \rho_{\delta}\delta^{\prime}\delta\\
        \rho_{\delta}\delta^{\prime}\delta & 0 & \delta^{\prime}\delta & 0\\
        0 &  \rho_{\delta}\delta^{\prime}\delta & 0 & \delta^{\prime}\delta 
    \end{pmatrix},
\end{equation*}
where
\begin{equation} \label{eq:matrix that must be invertible 1}
    E \equiv \begin{pmatrix}
        E_1 & E_2 \\
        P E_2 P & P E_1 P 
    \end{pmatrix},
    \quad
    E_1 = \begin{pmatrix} 
    c_1 & c_0 \\
    e_1 & e_0
    \end{pmatrix}, \; 
    E_2 = \begin{pmatrix}
       e_1 & e_0 \\
       d_1 & d_0 \\
    \end{pmatrix},
\end{equation}
and $P$ is the perturbation matrix
\begin{equation*}
      P = \begin{pmatrix}
        0 & 1 \\
        1 & 0
    \end{pmatrix}.
\end{equation*}

By assumption (iv), $E$ is invertible. Let $F = E^{-1}$. Inverting $E$ as a block matrix, and using that $P$ is its own inverse 
\begin{equation}
    F = 
    \begin{pmatrix}
        F_1 & F_2 \\
        P F_{2} P & P F_1 P 
    \end{pmatrix}, 
\quad
  F_1 = \begin{pmatrix} 
  f_{11} & f_{12} \\
    f_{21} & f_{22} 
    \end{pmatrix}, \; 
    F_2 = \begin{pmatrix}
       f_{13} & f_{14} \\
       f_{23} & f_{24} \\
    \end{pmatrix}.  \label{eq:E block matrix inverse}
\end{equation}

Hence,
\begin{equation*}
   \begin{pmatrix}
        \delta_1^{\prime}\Omega_1\delta_1 & \delta_1^{\prime}\Omega_2\delta_1 & \delta_1^{\prime}\Omega_1\delta_0 & \delta_1^{\prime}\Omega_2\delta_0  \\
        \delta_0^{\prime}\Omega_1\delta_1 & \delta_0^{\prime}\Omega_2\delta_1 & \delta_0^{\prime}\Omega_1\delta_0 & \delta_0^{\prime}\Omega_2\delta_0 \\
        \delta_1^{\prime}\Omega_2^{\prime}\delta_1 & \delta_1^{\prime}\Omega_3\delta_1 & \delta_1^{\prime}\Omega_2^{\prime}\delta_0 & \delta_1^{\prime}\Omega_3\delta_0 \\
        \delta_0^{\prime}\Omega_2^{\prime}\delta_1 & \delta_0^{\prime}\Omega_3\delta_1 &  \delta_0^{\prime}\Omega_2^{\prime}\delta_0 & \delta_0^{\prime}\Omega_3\delta_0
    \end{pmatrix} =   \begin{pmatrix}
        F_1 & F_2 \\
        P F_{2} P & P F_1 P 
    \end{pmatrix} \begin{pmatrix}
        \delta^{\prime}\delta & 0 &  \rho_{\delta}\delta^{\prime}\delta & 0\\
        0 &  \delta^{\prime}\delta & 0 & \rho_{\delta}\delta^{\prime}\delta\\
        \rho_{\delta}\delta^{\prime}\delta & 0 & \delta^{\prime}\delta & 0\\
        0 &  \rho_{\delta}\delta^{\prime}\delta & 0 & \delta^{\prime}\delta 
    \end{pmatrix}.
\end{equation*}

From which we can obtain the relations stated in  Lemma \ref{lemma:equalities}
\begin{align}
    &\delta_1^{\prime}\Omega_1\delta_1 = \delta_0^{\prime}\Omega_3\delta_0 = f_{11} \delta^{\prime}\delta + f_{13} \rho_{\delta} \delta^{\prime}\delta, \label{eq:equality one} \\
    &\delta_0^{\prime}\Omega_1 \delta_0 =  \delta_1^{\prime}\Omega_3\delta_1 = f_{21} \rho_{\delta} \delta^{\prime}\delta + f_{23} \delta^{\prime}\delta, \label{eq:equality two} \\
    &\delta_0^{\prime}\Omega_1 \delta_1 =  \delta_1^{\prime}\Omega_3\delta_0 = f_{21}  \delta^{\prime}\delta + f_{23} \rho_{\delta} \delta^{\prime}\delta, \label{eq:equality three} \\
    &\delta_1^{\prime}\Omega_1 \delta_0 =  \delta_0^{\prime}\Omega_3\delta_1 = f_{11}  \rho_{\delta}\delta^{\prime}\delta + f_{13}  \delta^{\prime}\delta, \label{eq:equality four} \\
    &\delta_1^{\prime}\Omega_2\delta_1 =  \delta_0^{\prime}\Omega_2^{\prime} \delta_0 = f_{12}  \delta^{\prime}\delta + f_{14} \rho_{\delta} \delta^{\prime}\delta, \label{eq:equality five}\\
    &\delta_1^{\prime}\Omega_2^{\prime} \delta_1 =  \delta_0^{\prime}\Omega_2 \delta_0 = f_{22}  \rho_{\delta}\delta^{\prime}\delta + f_{24}  \delta^{\prime}\delta, \label{eq:equality six}\\
     &\delta_0^{\prime}\Omega_2^{\prime} \delta_1 =  \delta_1^{\prime}\Omega_2 \delta_0 = f_{12}  \rho_{\delta}\delta^{\prime}\delta + f_{14}  \delta^{\prime}\delta, \label{eq:equality seven} \\
    &\delta_0^{\prime}\Omega_2 \delta_1 =  \delta_1^{\prime}\Omega_2^{\prime} \delta_0 = f_{22}  \delta^{\prime}\delta + f_{24}  \rho_{\delta}\delta^{\prime}\delta. \label{eq:equality eight}
\end{align}

\end{proof}

\begin{proof}[Proof of Lemma \ref{lemma:E XX inverse parameters}]

The formulas for $\delta_i^{\prime} \Omega_k \delta_j$ $i,j = \{0, 1\}$, $k = \{1, 2, 3\}$ above contain $\delta_{\prime}\delta$, $\rho_{\delta}$, and the elements of the inverse of $E$. The elements of $E$ are given by combinations of $\rho_{\delta}$, $\delta^{\prime}\delta$ and integrals of the propensity score density and the conditional distribution of $Z$. Recall that the propensity score density and the conditional distribution of $Z$ depend only on $\alpha$ and $\delta^{\prime}\delta$. Hence, $E$ depends only on $\alpha$, $\delta^{\prime}\delta$ and $\rho_{\delta}$, which implies its also inverse depends only on these parameters. Thus, the only parameters affecting each $\delta_i^{\prime} \Omega_k \delta_j$ are $\alpha$, $\delta^{\prime}\delta$ and $\rho_{\delta}$. \end{proof}

\begin{proof}[Proof of Statements 1 and 2]

Under assumption (i), estimates of $\theta_0$ and $\theta_1$ are obtained by regressing $Y - E[Y|p]$ on $p(X - E[X|p])$ and $(1 - p)(X - E[X|p])$. The asymptotic limit of these estimates is
\begin{equation*}
    \left(\begin{array}{c}
         \hat{\theta}_0 \\
         \hat{\theta}_1 
    \end{array} \right) = \begin{pmatrix}
           \theta_0 \\
          \theta_1 
    \end{pmatrix} + E[\mathbf{X}\mathbf{X}^{\prime}]^{-1}E[\mathbf{X}\mathbf{Z}] 2 \nu,
\end{equation*}
so that
\begin{equation*}
    \begin{pmatrix}
        \tilde{\theta}_1 \\
        \tilde{\theta}_0
    \end{pmatrix} =  - E[\mathbf{X}\mathbf{X}^{\prime}]^{-1}E[\mathbf{X}\mathbf{Z}] 2 \nu.
\end{equation*}

Use Lemma \ref{lemma:E XZ} to substitute for $E[\mathbf{X}\mathbf{Z}]$, Lemma \ref{lemma:E XX} to substitute for  $E[\mathbf{X}\mathbf{X}^{\prime}]^{-1}$ and simplify to obtain
\begin{equation}
    \begin{pmatrix}
        \tilde{\theta}_1 \\
        \tilde{\theta}_0
    \end{pmatrix} = - \begin{pmatrix}
        w_1 \Omega_1 \delta_1 + w_0 \Omega_1 \delta_0 + w_1 \Omega_2 \delta_0 + w_0 \Omega_2 \delta_1 \\
        w_1 \Omega_2^{\prime}\delta_1 + w_0 \Omega_2^{\prime}\delta_0 + w_1 \Omega_3 \delta_0 + w_0 \Omega_3 \delta_1
    \end{pmatrix} 2 \nu. \label{eq:asymptotic bias expression}
\end{equation}

Pre-mutliply the two elements in \eqref{eq:asymptotic bias expression} by $\delta_1^{\prime}$ and $\delta_0^{\prime}$ to obtain
\begin{align}
    \delta_1^{\prime} \tilde{\theta}_1 &= -(w_1 \delta_1^{\prime}\Omega_1 \delta_1 + w_0 \delta_1^{\prime}\Omega_1 \delta_0 + w_1 \delta_1^{\prime}\Omega_2 \delta_0 + w_0 \delta_1^{\prime}\Omega_2 \delta_1)2\nu, \label{eq: d1 theta1}\\
    \delta_1^{\prime} \tilde{\theta}_0 &= -(w_1 \delta_1^{\prime}\Omega_2^{\prime}\delta_1 + w_0 \delta_1^{\prime}\Omega_2^{\prime}\delta_0 + w_1 \delta_1^{\prime}\Omega_3 \delta_0 + w_0 \delta_1^{\prime}\Omega_3 \delta_1)2\nu, \label{eq: d1 theta0} \\
  \delta_0^{\prime} \tilde{\theta}_1 &= -(w_1 \delta_0^{\prime}\Omega_1 \delta_1 + w_0 \delta_0^{\prime}\Omega_1 \delta_0 + w_1 \delta_0^{\prime}\Omega_2 \delta_0 + w_0 \delta_0^{\prime}\Omega_2 \delta_1)2\nu, \label{eq: d0 theta1}\\
   \delta_0^{\prime} \tilde{\theta}_0 &= -(w_1 \delta_0^{\prime}\Omega_2^{\prime}\delta_1 + w_0 \delta_0^{\prime}\Omega_2^{\prime}\delta_0 + w_1 \delta_0^{\prime}\Omega_3 \delta_0 + w_0 \delta_0^{\prime}\Omega_3 \delta_1)2\nu. \label{eq: d0 theta0}
\end{align}

Using \eqref{eq: d1 theta1} and \eqref{eq: d0 theta0}
\begin{multline*}
     \delta_1^{\prime}  \tilde{\theta}_1 -   \delta_0^{\prime}  \tilde{\theta}_0 = -2(w_1 (\delta_1^{\prime} \Omega_1 \delta_1 - \delta_0^{\prime} \Omega_3 \delta_0) + w_0 (\delta_1^{\prime} \Omega_1 \delta_0 - \delta_0^{\prime} \Omega_3 \delta_1) + w_1 (\delta_1^{\prime} \Omega_2 \delta_0 - \delta_0^{\prime} \Omega_2^{\prime} \delta_1) +  w_0 (\delta_1^{\prime} \Omega_2 \delta_1 - \delta_0^{\prime} \Omega_2^{\prime} \delta_0))\nu.
\end{multline*}
Impose the equalities of Lemma \ref{lemma:equalities} and $\delta_1^{\prime} \Omega_2 \delta_0 = \delta_0^{\prime} \Omega_2^{\prime} \delta_1$, and this is zero. Repeating this process with \eqref{eq: d1 theta0} and \eqref{eq: d0 theta1} yields that $\delta_0^{\prime} \tilde{\theta}_1 = \delta_1^{\prime} \tilde{\theta}_0$. This proves the first statement of Proposition \ref{prop:bias simplification}.

To see the second statement, consider \eqref{eq: d1 theta1}. By Lemma \ref{lemma:E XZ}, $w_0$ and $w_1$ depend only on $\alpha$ and $\delta^{\prime}\delta$. By Lemma \ref{lemma:E XX inverse parameters}, $\delta_1^{\prime}\Omega_1 \delta_1$, $\delta_1^{\prime}\Omega_1 \delta_0$, $\delta_1^{\prime}\Omega_2 \delta_0$ and $\delta_1^{\prime}\Omega_2 \delta_1$ depend only on  $\alpha$, $\delta^{\prime}\delta$ and $\rho_{\delta}$. Hence, we can write
\begin{equation*}
    \delta_1^{\prime} \tilde{\theta}_1 = \mathcal{H}_1(\alpha, \rho_{\delta}, \delta^{\prime}\delta) \nu,
\end{equation*}
with
\begin{equation*}
     \mathcal{H}_1(\alpha, \rho_{\delta}, \delta^{\prime}\delta) \equiv -(w_1 \delta_1^{\prime}\Omega_1 \delta_1 + w_0 \delta_1^{\prime}\Omega_1 \delta_0 + w_1 \delta_1^{\prime}\Omega_2 \delta_0 + w_0 \delta_1^{\prime}\Omega_2 \delta_1)2 .
\end{equation*}
The same argument can be applied to \eqref{eq: d1 theta0} to obtain the second identity. \end{proof}

\subsubsection{Additional Results}

For \eqref{eq:equality three} and \eqref{eq:equality four} to both hold, it must be that
\begin{equation*}
     f_{21} + f_{23} \rho_{\delta}  =  f_{11}  \rho_{\delta} + f_{13}. 
\end{equation*}
Similarly, for \eqref{eq:equality five} and \eqref{eq:equality six} to both hold, we require
\begin{equation*}
      f_{12}   + f_{14} \rho_{\delta}  =  f_{22}  \rho_{\delta} + f_{24}.
\end{equation*}

As $F$ is the inverse of $E$
\begin{equation*}
     \begin{pmatrix}
        F_1 & F_2 \\
        P F_{2} P & P F_1 P 
    \end{pmatrix}
     \begin{pmatrix}
        E_1 & E_2 \\
        P E_2 P & P E_1 P 
    \end{pmatrix} =
   \begin{pmatrix}
        I & 0 \\
        0 & I 
    \end{pmatrix}. 
\end{equation*}

Post-multiply both sides by the matrix with columns $(\rho_{\delta}, 1, 0, 0)$ and $(0, 0, \rho_{\delta}, 1)$ to obtain
\begin{equation}
      \begin{pmatrix}
        F_1 & F_2 \\
        P F_{2} P & P F_1 P 
    \end{pmatrix}
    \begin{pmatrix}
        G_1  \\
        G_2
    \end{pmatrix} 
   =   \begin{pmatrix}
        \rho_{\delta} & 0 \\
        1 & 0 \\
        0 & \rho_{\delta} \\
        0 & 1
    \end{pmatrix}, \quad     G_1 = \begin{pmatrix}
        g_1 & g_2 \\
        g_2 & g_1 \\
    \end{pmatrix}, \; G_2 = \begin{pmatrix}
        g_4 & g_5  \\
        g_5 & g_3
    \end{pmatrix}, \label{eq:F G system}
\end{equation}
with
\begin{align*}
    g_1 &= \rho_{\delta} c_1 + c_0 = \rho_{\delta} d_1 + d_0, & g_3 &= \rho_{\delta} c_0 + c_1, \\
    g_2 &= \rho_{\delta} e_1 + e_0, & g_4 &= \rho_{\delta} d_0 + d_1, \\
    & & g_5 &= \rho_{\delta} e_0 + e_1.
\end{align*}
Where the $g_1$ uses \eqref{eq:symmetry condtion}. Notice that
\begin{align}
    & G_1 P = P G_1, \label{eq: G_1 P condition}\\
    & G_1 = P G_1 P = G_1^{T}, \label{eq:G_1 transpose}\\
    & G_2 = G_2^{T}. \label{eq:G_2 transpose}
\end{align}

The system \eqref{eq:F G system} can be rearranged to
\begin{equation}
    G \begin{pmatrix}
        F_1^T \\
        F_2^T
    \end{pmatrix} 
   =   \begin{pmatrix}
        \rho_{\delta} & 1 \\
        0 & 0 \\
        0 & 0 \\
        1 & \rho_{\delta}
    \end{pmatrix}, \label{eq: G F system}
\end{equation}
where
\begin{equation}
    G \equiv \begin{pmatrix}
        G_1 & G_2 \\
        G_2 P & G_1 P
    \end{pmatrix}. \label{eq:matrix that must be invertible 2}
\end{equation}

By assumption (iv), $G$ is invertible. Let $G^{-1} = H$. Applying the formulas for block matrix inversion
\begin{equation*}
     H = \begin{pmatrix}
        H_1 & H_2 \\
        P H_2 & P H_1
    \end{pmatrix}.
\end{equation*}
Furthermore, using  \eqref{eq: G_1 P condition}, \eqref{eq:G_1 transpose} and \eqref{eq:G_2 transpose}, it is possible to show 
\begin{align*}
    H_1^T = P H_1 P, \\
    H_2^T = P H_2 P,
\end{align*}
which implies that $H_1$ and $H_2$ are constant diagonal matrices
\begin{equation*}
    H_1 = \begin{pmatrix}
        h_{11} & h_{12} \\
        h_{21} & h_{11}
    \end{pmatrix}, \; 
    H_2 = \begin{pmatrix}
        h_{13} & h_{14} \\
        h_{23} & h_{13}
    \end{pmatrix}.
\end{equation*}
Multiplying both sides of \eqref{eq: G F system} by $H$
\begin{equation*}
     \begin{pmatrix}
        F_1^T \\
        F_2^T
    \end{pmatrix} 
   =  \begin{pmatrix}
        H_1 & H_2 \\
        P H_2 & P H_1
    \end{pmatrix} \begin{pmatrix}
        \rho_{\delta} & 1 \\
        0 & 0 \\
        0 & 0 \\
        1 & \rho_{\delta}
    \end{pmatrix}, 
\end{equation*}

which yields the following formulas 
\begin{align*}
    & f_{11} = \rho_{\delta} h_{11} + h_{14}, & & f_{21} = h_{11} + \rho_{\delta} h_{14}, \\
    & f_{12} = \rho_{\delta} h_{21} + h_{13}, & & f_{22} = h_{21} + \rho_{\delta} h_{13}, \\
    & f_{13} = \rho_{\delta} h_{23} + h_{11}, & & f_{23} = h_{23} + \rho_{\delta} h_{11}, \\
    & f_{14} = \rho_{\delta} h_{13} + h_{12}, & & f_{24} = h_{13} + \rho_{\delta} h_{12}.
\end{align*}
Using these formulas
\begin{align*}
    f_{21} + \rho_{\delta} f_{23} = \rho_{\delta} f_{11} + f_{13} = (1 + \rho_{\delta}^2) h_{11} + \rho_{\delta} (h_{14} + h_{23}),  \\
    f_{12} + \rho_{\delta} f_{14} = \rho_{\delta} f_{22} + f_{24} = (1 + \rho_{\delta}^2) h_{13} + \rho_{\delta} (h_{12} + h_{21}). 
\end{align*}

\section{Modified Distillation Algorithm}\label{sec:distillation algorithm details}

As in section \ref{sec:distillation procedure}, assume the sample is sorted in ascending order by propensity score, so that observation $i = 1$ has the lowest propensity score and $i = N$ the highest, with ties broken by placing observations with $Z=0$ above observations with $Z=1$. Maintain the definitions of $n_0$, $n_1$, and $\Delta_j$ of section \ref{sec:distillation procedure}, and let  $J^{-} = \{j:p_j \in P^{-}\}$ and Let $J^{+} = \{j:p_j \in P^{+}\}$.

Define the following `reaction' functions
\begin{align}
    R_{d_1}(d_0) &=  \underset{j \in J^{-}}{\text{max}} \left\{ n_1 \frac{n_0 - d_{0}}{n_0 - d_{0} - \sum_{i=1}^j (1 - Z_i)} \left(\Delta_j - \frac{d_{0}}{n_0 - d_{0}} \frac{\sum_{i = 1}^j (1 - Z_i)}{n_0}\right) \label{eq:d1 reaction} \right\} \\
    R_{d_0}(d_1) &=  \underset{j \in J^{+}}{\text{max}} \left\{ n_0 \frac{n_1 - d_{1}}{\sum_{i=1}^j Z_i - d_{1}} \left(\Delta_j - \frac{d_{1}}{n_1 - d_{1}} \frac{n_1 - \sum_{i = 1}^j Z_i}{n_1}\right) \right\} \label{eq:d0 reaction} 
\end{align}
Given that $d_{0}$ observations with $(Z_i = 0, p_i \in P^{+})$ are trimmed, \eqref{eq:d1 reaction} returns the number of observations with  $(Z_i = 1, p_i \in P^{-})$ that must be trimmed for first-order stochastic dominance to hold. Given that $d_{1}$ observations with $(Z_i = 1, p_i \in P^{-})$ are trimmed, \eqref{eq:d0 reaction} returns the number of observations with  $(Z_i = 0, p_i \in P^{+})$ that must be trimmed for first-order stochastic dominance to hold.   

The modified algorithm is as follows  
\begin{enumerate}
    \item Set $S_{1i} = 1 \; \forall \; i$
    \item Calculate $ \Delta_j$ for $j = 1,...,N$. If $\text{max}(\Delta_j) \leq 0$, stop. If $\text{max}(\Delta_j) > 0$, continue to steps 3, 4, and 5.
    \item Obtain $d_1$ and $d_0$ as follows
    \begin{enumerate}
        \item Calculate
        \begin{align*}
            & d_1^{\text{max}} = R_{d_1}(0) \\
            & d_0^{\text{max}} = R_{d_0}(0) \\
            & d_1^{\text{min}} = \text{max}\{R_{d_1}(d_0^{\text{max}}), 0\} \\
            & d_0^{\text{min}} = \text{max}\{R_{d_0}(d_1^{\text{max}}), 0\} \\
        \end{align*}
        \item  If $d_1^{\text{min}} > 0$ and  $d_0^{\text{min}} > 0$ or $d_1^{\text{min}} = 0$ and  $d_0^{\text{min}} = 0$, obtain all $d_0, d_1$ where $R_{d_1}(d_0) = R_{d_0}(d_1)$
        \item Evaluate $d_0 + d_1$ at $(d_0^{\text{min}}, d_1^{\text{max}})$, $(d_0^{\text{max}}, d_1^{\text{min}})$ and all $d_0, d_1$ where $R_{d_1}(d_0) = R_{d_0}(d_1)$, Select the combination that attains the minimum. 
        \item Obtain $j^{-}$ as the lowest $j \in J^{-}$ where $R_{d_1}(d_0) = d_1$ and $j^{+}$ as the highest $j \in J^{+}$ where $R_{d_0}(d_1) = d_0$. 
    \end{enumerate}    
    \item Loop over $j = 1,...,j^-$. If $Z_j = 1$
     \begin{enumerate}
         \item Calculate
         \begin{equation*}
              \delta_{1j} = \frac{\sum_{i=1}^j S_{1,i} Z_i}{n_1 -  d_{1}} - \frac{\sum_{i=1}^j (1 - Z_i)}{n_0 - d_0}
         \end{equation*}
         \item If $\delta_{1j} > 0$, set $S_{1j} = 0$
     \end{enumerate}
   \item Loop over $j = N, ..., j^+ + 1$. If $Z_j = 0$
        \begin{enumerate}
            \item Calculate
            \begin{equation*}
              \delta_{0} =  \frac{\sum_{i=j}^N S_{1,i}(1 - Z_i)}{n_0 - d_{0}} - \frac{\sum_{i=j}^N  Z_i}{n_1 - d_{1}}
            \end{equation*}
            \item If $\delta_{0j} > 0$, set $S_{1j} = 0$
        \end{enumerate}
    \end{enumerate}

Steps 1 and 2 are identical to the algorithm of Section \ref{sec:distillation procedure}. In step 3, $d_0$ and $d_1$ are jointly determined such that $d_1$ is the smallest number of observations with $(Z_i = 1, p_i \in P^{-})$ that must be trimmed given that $d_0$ observations with $(Z_i = 0, p_i \in P^{+})$ are trimmed, $d_0$ is the smallest number of observations with $(Z_i = 0, p_i \in P^{+})$ that must be trimmed given that $d_1$ observations with $(Z_i = 1, p_i \in P^{-})$ are trimmed, and the sum $d_0 + d_1$ is minimised. Step 4 corresponds to step 3b of the algorithm of Section \ref{sec:distillation procedure}, with the modification that the trimmed distribution for $Z = 0$ is used. Step 5 is identical to step 4b of the algorithm of Section \ref{sec:distillation procedure}.

Figure \ref{fig:subsample full example} illustrates the modified algorithm using the same artificial sample as Figure \ref{fig:subsample simple example}. Panel \subref{fig:subsample full reaction functions} plots $R_{d_0}(d_1)$ and $R_{d_1}(d_0)$ for this sample. The algorithm of Section \ref{sec:distillation procedure} leads to $d_1 = 6, d_0 = 11$. In comparison, the modified algorithm selects $d_1 = 4, d_0 = 12$, so we trim two less observations with $Z = 1$ at the cost of trimming an additional observation with $Z = 0$. Panels \subref{fig:subsample full P^-} and \subref{fig:subsample full P^+} show how steps 4 and 5 choose which points to trim. Panel \subref{fig:subsample full trimmed cdfs} plots the resulting trimmed conditional propensity score distributions. 

\begin{figure}
    \centering
    
    \begin{subfigure}[b]{0.49 \textwidth}
    \centering
    \includegraphics[width = \textwidth]{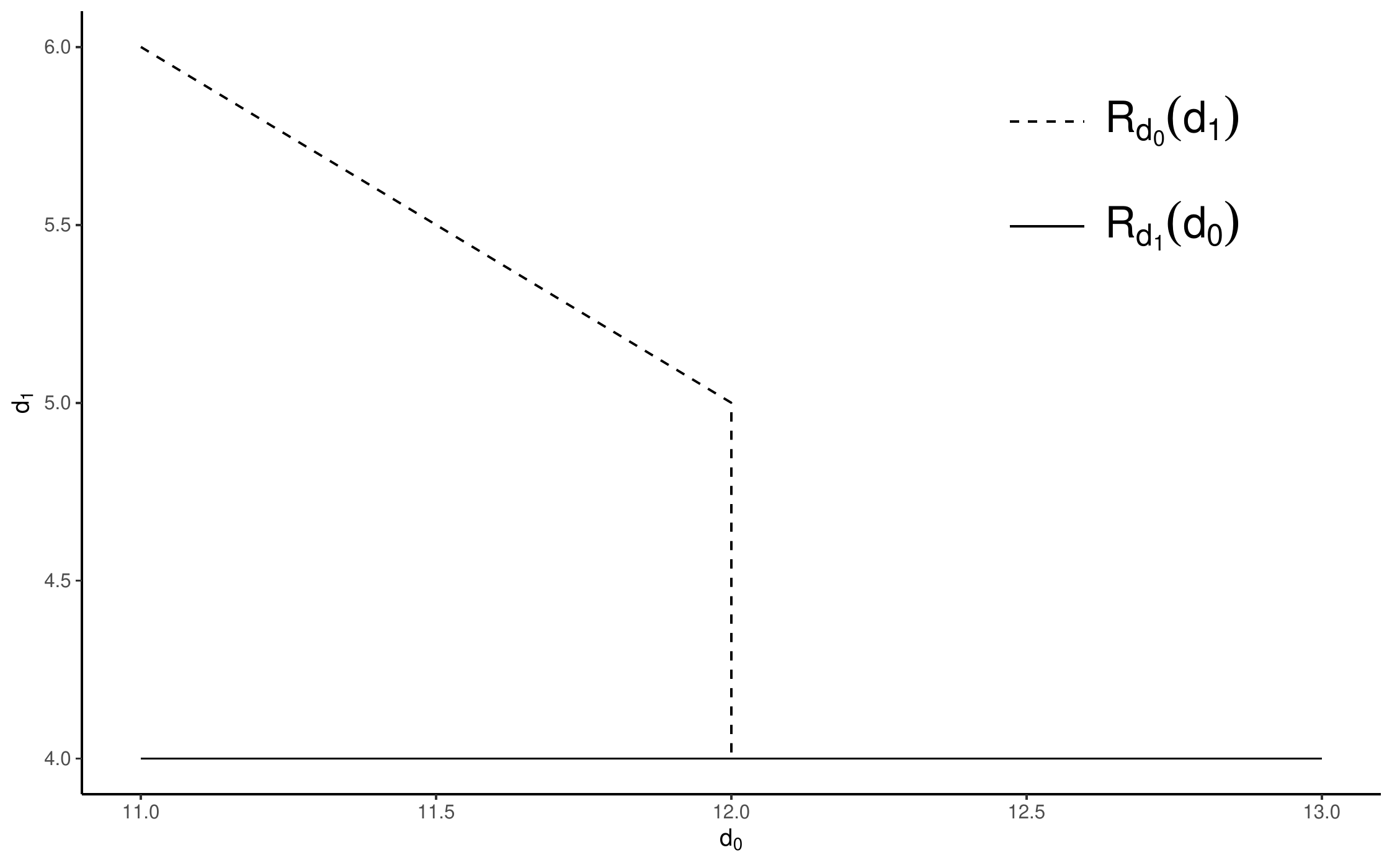}
    \caption{$R_{d_0}(d_1)$ and $R_{d_1}(d_0)$}\label{fig:subsample full reaction functions}
    \end{subfigure}
    \begin{subfigure}[b]{0.49 \textwidth}
    \centering
    \includegraphics[width = \textwidth]{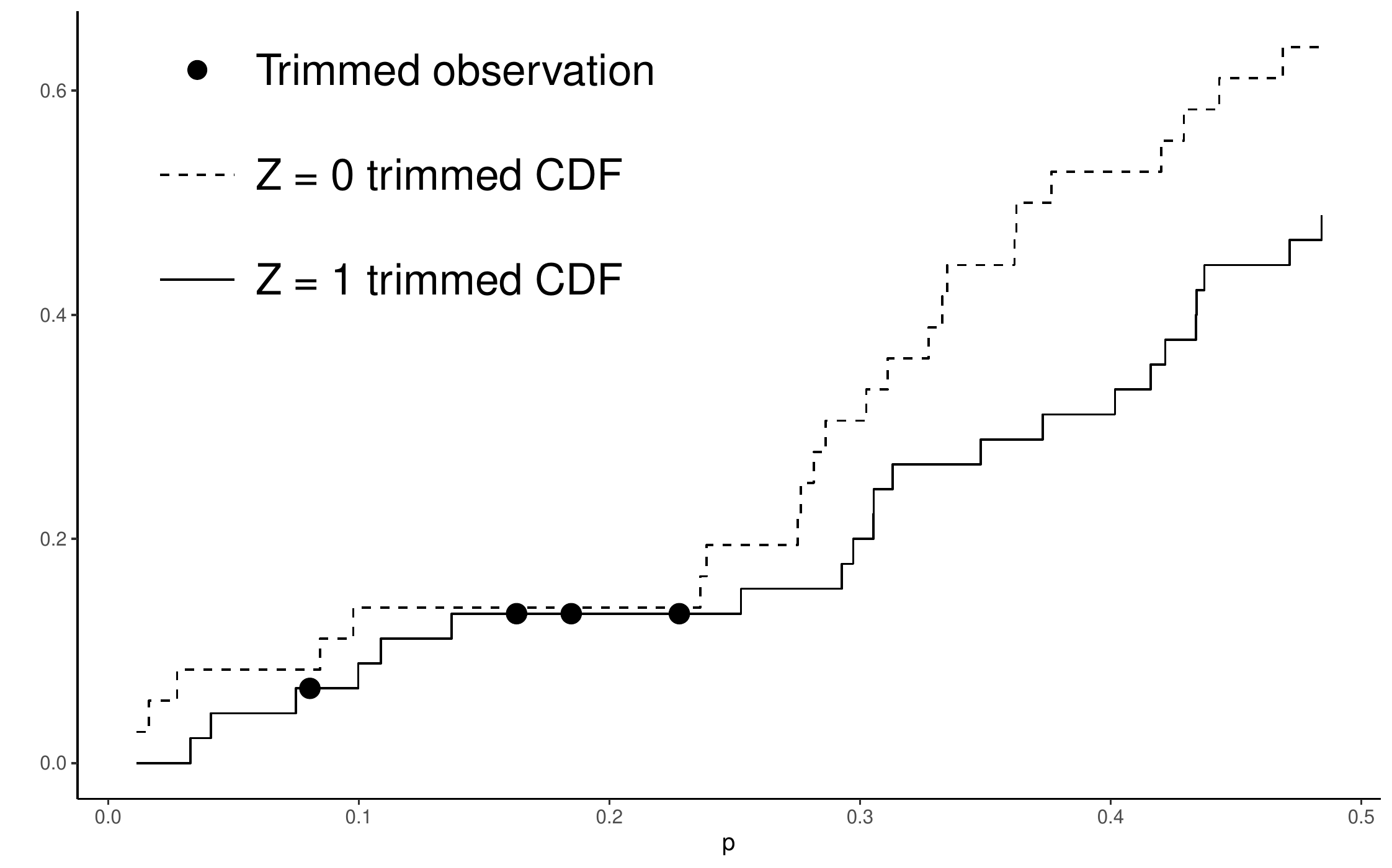}
    \caption{$P^-$ region}\label{fig:subsample full P^-}
    \end{subfigure}\\
    
    \begin{subfigure}[b]{0.49 \textwidth}
    \centering
    \includegraphics[width = \textwidth]{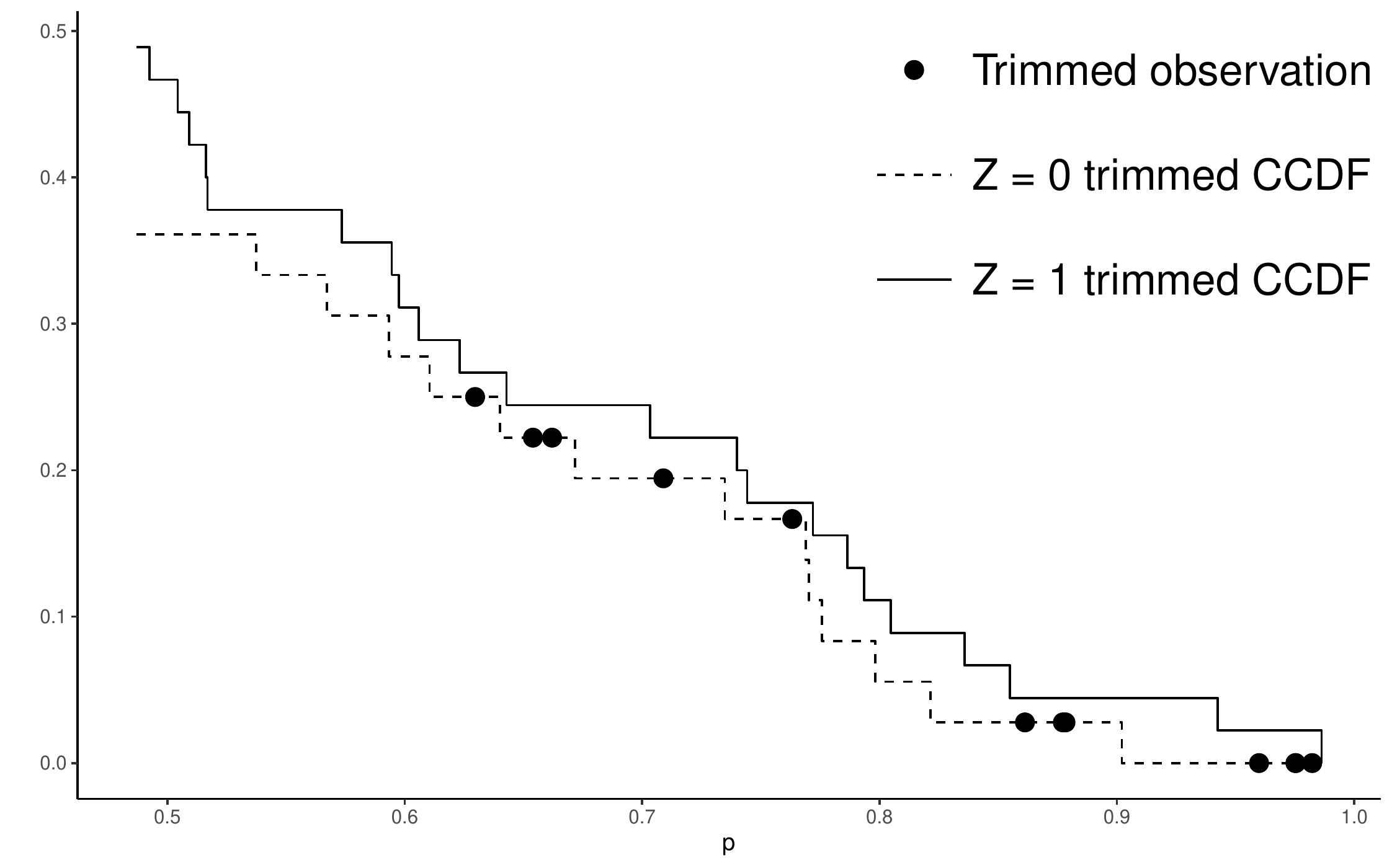}
    \caption{$P^+$ region}\label{fig:subsample full P^+}
    \end{subfigure}    
     \begin{subfigure}[b]{0.49 \textwidth}
    \centering
    \includegraphics[width = \textwidth]{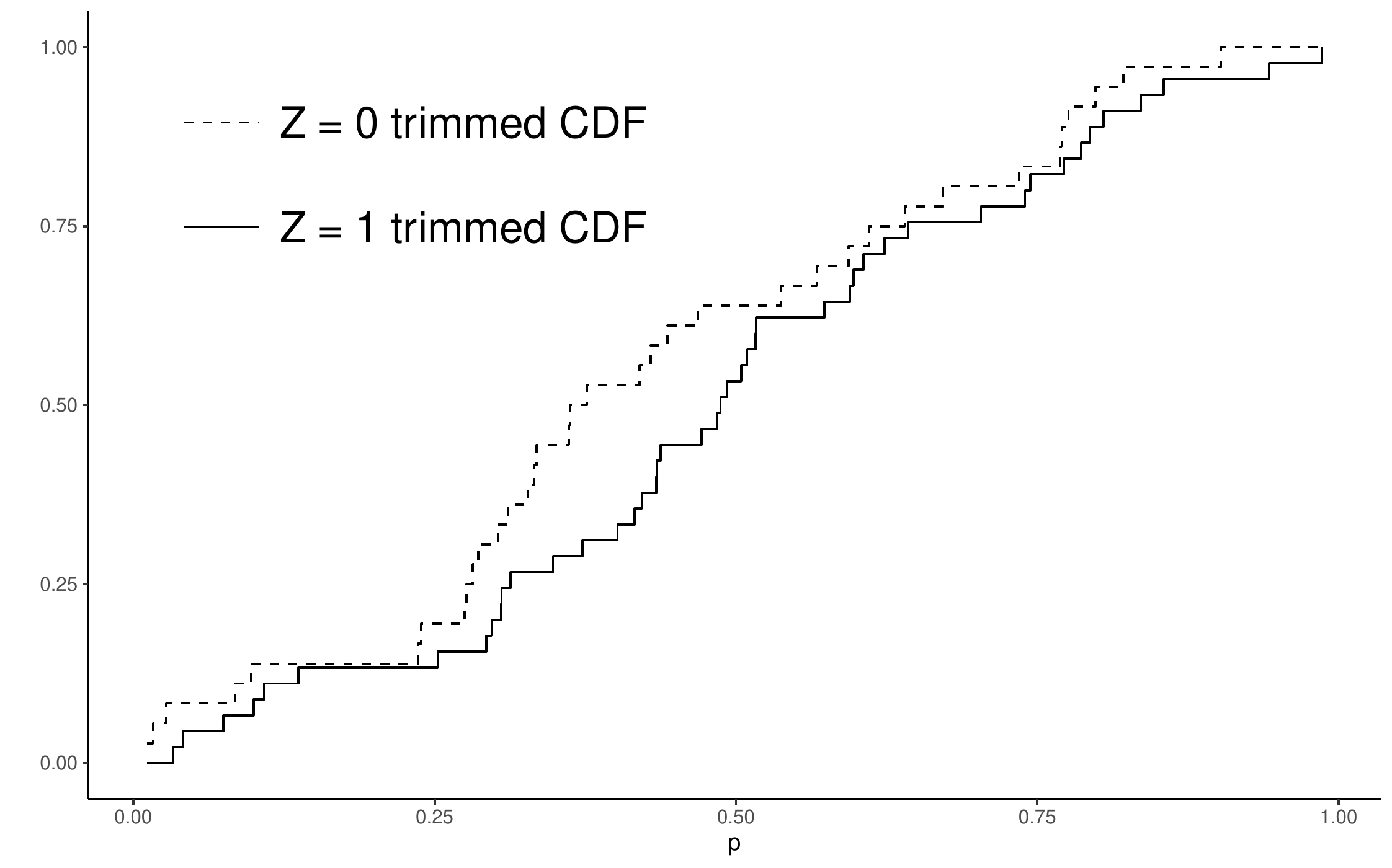}
    \caption{Trimmed distributions}\label{fig:subsample full trimmed cdfs}
    \end{subfigure}
    
    \caption{Example of the distillation algorithm. The sample consists of 100 observations. The data generating process sets $\Pr(Z = 1) = 1 / 2$ with propensity score uniformly distributed between 0 and 1 for $Z = 0, 1$. The interval $P^-$ includes all observations with a propensity weakly less than the median, while $P^+$ includes all observations with a propensity score above the median.} \label{fig:subsample full example}
\end{figure}

\section{Test Procedure for a Multi-valued Discrete Instrument}\label{sec:test procedure multi-valued discrete instrument}

The test statistic and procedure of sections \ref{sec:test statistic} and \ref{sec:test procedure} can be extended to the case of a single multi-valued instrument. This test procedure forms pairs of neighbouring instrument values and jointly tests the nesting inequalities and index sufficiency for all pairs. 

Let $Z \in \{z_1, z_2,...,z_K\}$, $n_k$ be the number of observations with $Z = z_k$, and $\lambda_{z_k} = \Pr(Z = z_k)$. Form $K-1$ pairs of neighbouring instrument values $(z_1, z_2)$, $(z_2, z_3)$,...,$(z_{K-1}, z_{L})$. For $k = 1, ..., K-1$, let $(S_{1k}, S_{2k})$ be a pair of sample inclusion indicators corresponding to the subsample containing only observations with $Z \in \{z_{k}, z_{k + 1}\}$. $S_{1k}$ can be constructed by applying a distillation procedure to each $(z_{k}, z_{k+1})$ subsample. The construction of $S_{2k}$ differs slightly from the headline test procedure in that it must now be based on the value of the two inverse probability weighting terms $\Pr(Z=z_{k}|p(X_i,Z_i))$ and $\Pr(Z=z_{k+1}|p(X_i,Z_i))$. For example $S_{2j} = 1$ if  $\Pr(Z=z_{k}|p(X_i,Z_i)) > 0.05$ and $\Pr(Z=z_{k + 1}|p(X_i,Z_i))> 0.05$.  


For the subsample of observations with $Z \in \{z_k, z_{k+1}\}$, denote the empirical conditional distributions of  $(U,D,X)$ given $Z=z_{k+1}$ and $Z=z_{k}$ by $P_{n_{k+1}}$ and $Q_{n_k}$. 

Define
\begin{align*}
    f_i^{z_k,S_{1k}}(A,d) &= \frac{1\{U_i \in A, D_i = d \} S_{1ik}}{\Pr(S_{1k} = 1 | Z=z_k )}, \\
    f_i^{z_{k+1},S_{1k}}(A,d) &= \frac{1\{U_i \in A, D_i = d \} S_{1ik}}{\Pr(S_{1k} = 1 | Z=z_{k+1} )}, \\
    g_i^{z_k,S_{2k}}(A,d) &= \frac{1\{U_i \in A, D_i = d \} \lambda_{z_k} S_{2ik}}{ \Pr(Z=z_k|p(X_i,Z_i))\Pr(S_{2k} = 1 | Z=z_k)},\\
    g_i^{z_{k+1},S_{2k}}(A,d) &= \frac{1\{U_i \in A, D_i = d \} \lambda_{z_{k+1}} S_{2ik}}{ \Pr(Z=z_k|p(X_i,Z_i))\Pr(S_{2k} = 1 | Z=z_{k+1})}.
\end{align*}
The test statistic is
\begin{align*}
    T \equiv &\underset{k \in \{1,...,K-1\}}{\max} T_k \\ 
    T_k \equiv & \max \left\{ \sup_{A \in \mathcal{C}(\mathcal{Y}),d \in \{0 , 1 \}} \frac{T_{1k}(A,d)}{\hat{\sigma}_{1k}(A,d)}, \sup_{A \in \mathcal{C}(\mathcal{Y}),d \in \{0 , 1 \}} \frac{T_{2k}(A,d)}{\hat{\sigma}_{2k}(A,d)}, \sup_{A \in \mathcal{C}(\mathcal{Y}), d \in \{0 , 1 \}} \frac{-T_{2k}(A,d)}{\hat{\sigma}_{2k}(A,d)} \right\}, \ \text{where} \\
    T_{1k}(A,d) = & \sqrt{\frac{n_{k + 1} n_{k}}{n_{k + 1}+n_{k}}} \left[ \left( Q_{n_k} f_i^{z_k,S_{1k}}(A,1) - P_{n_{k + 1}} f_i^{z_{k+1},S_{1k}}(A,1) \right) d  \right. \\ &+\left. (P_{n_{k+1}} f_i^{z_{k+1},S_{1k}}(A,0) - Q_{n_k} f_i^{z_k,S_{1k}}(A,0) )(1-d)  \right] \\
    T_{2k}(A,d) = & \sqrt{\frac{n_{k + 1} n_{k}}{n_{k + 1}+n_{k}}}  \left( Q_{n_k} g_i^{z_k, S_{2k}}(A,d) - P_{n_{k+1}} g_i^{z_{k+1}, S_{2k}}(A,d) \right) \\
    \hat{\sigma}_{1k}^2(A,d) = & \left\{ \lambda_{z_k}[ Q_{n_k} (f_i^{z_k,S_{1k}}(A,1))^2 - ( Q_{n_k} f_i^{z_k,S_{1k}}(A,1))^2] \right. \\ &+ \left. \lambda_{z_{k+1}}[ P_{n_{k+1}} (f_i^{z_{k+1},S_{1k}}(A,1))^2 - ( P_{n_{k+1}} f_i^{z_{k+1},S_{1k}}(A,1))^2] \right\} d \\
    & + \left\{ \lambda_{z_k}[ Q_{n_k} (f_i^{z_k,S_{1k}}(A,0))^2 - ( Q_{n_k} f_i^{z_k,S_{1k}}(A,0))^2] \right. \\ & + \left. \lambda_{z_{k + 1}}[ P_{n_{k+1}} (f_i^{z_{k+1},S_{1k}}(A,0))^2 - ( P_{n_{k+1}} f_i^{z_{k+1},S_{1k}}(A,0))^2] \right\} (1-d), \\
    \hat{\sigma}_{2k}^2(A,d) = & \lambda_{z_k}[ Q_{n_k} (g_i^{z_k,S_{2k}}(A,d))^2 - ( Q_{n_0} g_i^{z_kS_{2k}}(A,d))^2] + \lambda_{z_{k + 1}}[ P_{n_{k+1}} (g_i^{z_{k+1}S_{2k}}(A,d))^2 - ( P_{n_{k+1}} g_i^{z_{k+1}S_{2k}}(A,d))^2]. \\
\end{align*}

For the bootstrap, again let $(\hat{M}_1, \dots, \hat{M}_{n})$ be iid random bootstrap multipliers such that they are independent of the original sample and satisfy $E(\hat{M}_i) = 0$ and $Var(\hat{M}_i) = 1$. A bootstrap analogue of $T$ can be constructed as
\begin{align*}
    \hat{T} \equiv &\underset{k \in \{1,...,K-1\}}{\max} \hat{T}_k \\
    \hat{T}_k \equiv & \max \left\{ \sup_{A  \in \mathcal{C}(\mathcal{Y}),d \in \{0 , 1 \}} \frac{\hat{T}_{1k}(A,d)}{\sigma_{1k}(A,d)}, \sup_{A \in \mathcal{C}(\mathcal{Y}),d \in \{0 , 1 \}} \frac{\hat{T}_{2k}(A,d)}{\sigma_{2k}(A,d)}, \sup_{A \in \mathcal{C}(\mathcal{Y}),d \in \{0 , 1 \}} \frac{-\hat{T}_{2k}(A,d)}{\sigma_{2k}(A,d)} \right\}, \ \text{where} \\
    \hat{T}_{1k}(A,d) = & \sqrt{\frac{n_{k+1} n_{k}}{n_{k+1}+n_{k}}} \left( (\hat{Q}_{n_k} - Q_{n_k}) f_i^{z_k,S_{1k}}(A,1) - (\hat{P}_{n_{k+1}}- P_{n_{k+1}}) f_i^{z_{k+1},S_{1k}}(A,1) \right) d \\ 
    & + \sqrt{\frac{n_{k+1} n_{k}}{n_{k+1} + n_{k}}} \left( (\hat{P}_{n_{k+1}} - P_{n_{k+1}}) f_i^{z_{k+1},S_{1k}}(A,0) - (\hat{Q}_{n_k} - Q_{n_k} ) f_i^{z_k,S_{1k}}(A,0) \right) (1-d) \\
    \hat{T}_{2k}(A,d) = & \sqrt{\frac{n_{k + 1} n_{k}}{n_{k + 1}+n_{k}}}  \left[ ( \hat{Q}_{n_k} - Q_{n_k}) g_i^{z_k,S_{2k}}(A,d) - (\hat{P}_{n_{k + 1}} - P_{n_{k + 1}}) g_i^{z_{k + 1},S_{2k}}(A,d) \right],
    \end{align*}
where, for the random variable $(a_i: i=1, \dots, n)$, we define
\begin{align*}
(\hat{Q}_{n_k} - Q_{n_k}) a_i & = \frac{1}{n_k} \sum_{i: Z_i=z_k} \hat{M}_i a_i, \mspace{10mu}
(\hat{P}_{n_{k + 1}} - P_{n_{k + 1}}) a_i = \frac{1}{n_{k + 1}} \sum_{i: Z_i=z_{k + 1}} \hat{M}_i a_i.
\end{align*}

Note that common bootstrap draws are used in constructing for all $\hat{T}_{k}$, rather than a separate draw being generated for each.

\end{appendix}

\bibliographystyle{ecta}  
\bibliography{bibliography.bib}

\end{document}